\documentclass[sigconf, screen=true, nonacm=true]{acmart}

\AtBeginDocument{
  \providecommand\BibTeX{{
    \normalfont B\kern-0.5em{\scshape i\kern-0.25em b}\kern-0.8em\TeX}}}

\renewcommand\footnotetextcopyrightpermission[1]{}

\usepackage{amsopn}		
\usepackage{wasysym}	

\usepackage{mathtools}				

\usepackage{xspace}            

\usepackage{graphicx}   
\graphicspath{{figs/}}		

\usepackage{multirow} 

\usepackage{soul} 

\usepackage{caption}	
\usepackage{subcaption} 	
\usepackage{bm, upgreek}                     	

\usepackage{enumitem}   
\setlist{noitemsep}
\setlist{nolistsep}

\usepackage{booktabs}       
\usepackage{tabularx}

\let\footnotesize\scriptsize

\usepackage{mdframed}

\usepackage{balance}
\usepackage{tikz}
\usetikzlibrary{arrows,shapes,trees,backgrounds,automata,decorations.pathmorphing}
\usepackage{tkz-euclide}
\usetkzobj{all}

\usepackage[ruled,vlined]{algorithm2e}       
     
\DontPrintSemicolon     
\SetNlSty{}{}{}                
\SetAlgoInsideSkip{smallskip}   
\SetAlCapSkip{1.5mm}             
\SetInd{1.5mm}{1.5mm}             

\SetAlFnt{\small}			
\SetAlCapFnt{\small}		
\SetAlCapNameFnt{\small}

\usepackage{hyperref}   
\hypersetup{                                                            
	pdfpagemode=UseNone,               
    pdfstartview=Fit,
    pdfpagelayout=SinglePage,       
    bookmarks=false,
    bookmarksnumbered=true,
    bookmarksopen=false,             
    pdftex=true,
    unicode,
    colorlinks=true,       
    linkcolor=blue,          
    citecolor=Green4,  
    filecolor=magenta,      
    urlcolor=cyan           
}

\usepackage{aliascnt}

\newtheorem{theorem}{Theorem}          	

\newaliascnt{lemma}{theorem}				
\newtheorem{lemma}[lemma]{Lemma}            
\aliascntresetthe{lemma}  					

\newaliascnt{conjecture}{theorem}			
\newtheorem{conjecture}[conjecture]{Conjecture}   
\aliascntresetthe{conjecture}  				

\newaliascnt{remark}{theorem}				
         
\aliascntresetthe{remark}  					

\newaliascnt{corollary}{theorem}			

\aliascntresetthe{corollary}  				

\newaliascnt{definition}{theorem}			
\newtheorem{definition}[definition]{Definition}    
\aliascntresetthe{definition}  				

\newaliascnt{proposition}{theorem}			
\newtheorem{proposition}[proposition]{Proposition}
\aliascntresetthe{proposition}  		

\newaliascnt{example}{theorem}			
\newtheorem{example}[example]{Example}  
\aliascntresetthe{example}  			

\newaliascnt{axiom}{theorem}			
  
\aliascntresetthe{axiom}  			

\newaliascnt{problem}{theorem}			
  
\aliascntresetthe{problem}  			

\newaliascnt{fact}{theorem}		
  	
\aliascntresetthe{fact}  				

\newaliascnt{claim}{theorem}		
  	
\aliascntresetthe{claim}

\usepackage[capitalise,nameinlink]{cleveref}

\usepackage{tcolorbox}		
\tcbset{common/.style={
		enhanced jigsaw,	
		colback=gray!20,
		colframe=gray!40,
		arc=0mm,
		boxrule=1pt,		
		left=0mm,
		right=0mm,
		left skip=-1mm,
		right skip=-1mm,
		topsep at break=1mm,			
		top=-1mm,		
		bottom=0mm,		
		breakable,		
		parbox = false		
		}
	}

\makeatletter
\DeclareRobustCommand*\uell{\mathpalette\@uell\relax}
\newcommand*\@uell[2]{
  \setbox0=\hbox{$#1\ell$}
  \setbox1=\hbox{\rotatebox{10}{$#1\ell$}}
  \dimen0=\wd0 \advance\dimen0 by -\wd1 \divide\dimen0 by 2
  \mathord{\lower 0.1ex \hbox{\kern\dimen0\unhbox1\kern\dimen0}}
}
\makeatother

\usepackage{textcomp}	 
\DeclareRobustCommand{\perthousand}{
  \ifmmode
    \text{\textperthousand}
  \else
    \textperthousand
  \fi}

\newcommand{\introparagraph}[1]{\textbf{#1.}}        

\newcommand{\specificref}[2]{\hyperref[#2]{#1~\ref*{#2}}}

\newcommand{\commentout}[1]{}

\newcommand{\set}[1]{\{#1\}}

\newcommand{\sat}{\mbox{{\rm\sc SAT}}}

\renewcommand{\vec}[1]{\boldsymbol{\mathbf{#1}}}

\newcommand{\datarule}{{\,:\!\!-\,}}

\newcommand{\x}  { \mathrm{x}}			

\newcommand{\true}{\texttt{true}\xspace}		
\newcommand{\false}{\texttt{false}\xspace}

\newcommand{\witnesses}{\texttt{witnesses}\xspace}		
\newcommand{\arity}{\texttt{arity}\xspace}		
		
\newcommand{\dom}{\texttt{dom}\xspace}

\newcommand{\exSymb}{\textup{x}}
					
\newcommand{\ex}[1]{{#1}^{\exSymb}}

\usepackage{bigstrut}				

\newcommand{\Var}{\textup{\textit{Var}}}      	
\newcommand{\var}{\textup{\texttt{var}}}      	
      	
\renewcommand{\vec}[1]{\boldsymbol{\mathbf{#1}}}

\newcommand{\res}{\texttt{RES}} 
\newcommand{\rats}{q_{\textup{\textrm{rats}}}}  
\newcommand{\chain}{q_{\textup{\textrm{chain}}}} 
\newcommand{\achain}{q_{\textup{\textrm{chain}}}^{\textup{\textrm{a}}}} 
\newcommand{\bchain}{q_{\textup{\textrm{chain}}}^{\textup{\textrm{b}}}}
\newcommand{\cchain}{q_{\textup{\textrm{chain}}}^{\textup{\textrm{c}}}}
\newcommand{\abchain}{q_{\textup{\textrm{chain}}}^{\textup{\textrm{ab}}}}
\newcommand{\acchain}{q_{\textup{\textrm{chain}}}^{\textup{\textrm{ac}}}}
\newcommand{\bcchain}{q_{\textup{\textrm{chain}}}^{\textup{\textrm{bc}}}}
\newcommand{\abcchain}{q_{\textup{\textrm{chain}}}^{\textup{\textrm{abc}}}}
 
\newcommand{\perm}{q_\textup{perm}} 
\newcommand{\vc}{q_{\textup{\textrm{vc}}}} 
\newcommand{\lin}{q_{\textup{\textrm{lin}}}} 

\newcommand{\conv}{q_{\textup{\textrm{conf}}}}

\newcommand{\brats}{q_{\textrm{brats}}}

\newcommand{\NP}{\textup{\textsf{NP}}\xspace}
 
\newcommand{\PTIME}{\textup{\textsf{PTIME}}\xspace}

\newcommand{\pos}{\texttt{pos}}

\newcommand{\flow}{\textup{\texttt{flow}}\xspace}

\renewcommand{\vec}[1]{\boldsymbol{\mathbf{#1}}}

\renewcommand{\phi}{\varphi}      
\renewcommand{\angle}[1]{ \langle #1 \rangle }
\newcommand{\ov}{\overline}

\usepackage{wasysym}

\newcommand{\bigset}[2]{\bigl\{ #1 \,\bigm|\, #2 \bigr\} }

\newcommand{\qLra}{\quad\Leftrightarrow\quad}

\newcommand{\abs}[1]{ \vert #1 \vert }

\newcommand{\rewriteclosed}		

\usepackage{tikz}
\usetikzlibrary{arrows,shapes,trees,backgrounds,automata}
\definecolor{dg}{cmyk}{0.60,0,0.88,0.27}

\usepackage{marvosym}

\newcommand{\p}{\mbox{{\rm P}}}
\newcommand{\np}{\mbox{{\rm NP}}}

\newcommand{\eqdef}{\buildrel \mbox{\tiny\rm def} \over =}
\newcommand{\Tri}{\textup{\textrm{T}}}

\usepackage{soul}

\newcommand{\sj}{\textrm{sj}}

\newcommand{\allsj}{\CIRCLE}
\newcommand{\ssjbin}{\LEFTcircle}
\newcommand{\twoR}{\LEFTcircle^{:}}
\newcommand{\threeR}{\LEFTcircle^{\because}}

\begin{document}

\title[Resilience for Binary SJ Conjunctive Queries]{New Results for the Complexity of Resilience for Binary Conjunctive Queries with Self-Joins}

\author{Cibele Freire}
\affiliation{
  \institution{Wellesley College}
}

\author{Wolfgang Gatterbauer}
\affiliation{
  \institution{Northeastern University, Boston}
}

\author{Neil Immerman}
\affiliation{
  \institution{University of Massachusetts Amherst}
}

\author{Alexandra Meliou}
\affiliation{
 \institution{University of Massachusetts Amherst}}

\begin{abstract}	
	
The resilience of a Boolean query {on a database} is the minimum number of tuples that need to be deleted from the input tables in order to make the query false. A solution to this problem immediately translates into a solution for the more widely known problem of deletion propagation with source-side effects. In this paper, we give several novel results on the hardness of the resilience problem for conjunctive queries with self-joins, and, more specifically, we present a dichotomy result for the class of \emph{single-self-join binary queries} with exactly two repeated relations occurring in the query. Unlike in the self-join free case, the concept of triad is not enough to fully characterize the complexity of resilience.  We identify new structural properties, namely chains, confluences and permutations, which lead to various \np-hardness results. We also give novel involved reductions to network flow to show certain cases are in \p. Although restricted, our results provide important insights into the problem of self-joins that we hope can help solve the general case of all conjunctive queries with self-joins in the future.
\end{abstract}


\maketitle

\section{Introduction}\label{sec:introduction}

Various problems in database research, such as causality, explanations, and
deletion propagation, examine how \emph{interventions in the input} to a query
impact the query's output. An intervention constitutes a change (update,
addition, or deletion) to the input tuples. 
In this paper, we study the
\emph{resilience} of a Boolean query with respect to tuple deletions.
Resilience is a variant of deletion propagation that focuses on Boolean
queries: it corresponds to the minimum number of tuples whose deletion causes
the query to evaluate to false. In previous work~\cite{FreireGIM15}, we
provided a full characterization of the complexity of resilience for the
family of self-join-free conjunctive queries (sj-free CQs) with functional
dependencies. 
In this paper, we augment the previous results to account for a restricted class of
self-joins.

Self-joins have long plagued the complexity study of many problems in database
theory research: for example, on the topic of \emph{consistent query answering}, Kolaitis and
Pema~\cite{KolaitisP12} proved a dichotomy into PTIME and coNP-complete cases
for the family of queries with only two atoms and no self-joins. Koutris and
Suciu~\cite{KoutrisS14} extended the dichotomy to the larger class of
self-join-free conjunctive queries, where each atom has as primary key either
a single attribute or all the attributes. Koutris and
Wijsen~\cite{Koutris2017,Koutris2018b} further extended the dichotomy to the
full class of sj-free Boolean CQs, and queries with
negated atoms~\cite{Koutris2018a}. To the best of our knowledge, there
is no known result on this problem for a query family that permits self-joins.
As another example, complexity results on the problem of \emph{query-based pricing}~\cite{KoutrisUBHS15}
are also restricted to the class of sj-free CQs. 
On the closely related topic of \emph{deletion propagation} with view side-effects, Kimelfeld et
al.~\cite{KimelfeldVW12} used a characteristic of the query structure (head
domination) to formalize a complexity dichotomy for the family of
sj-free CQs,
and indicated that self-joins can significantly harden
approximation in the problem of deletion propagation. 
Extensions to the cases of functional
dependencies~\cite{Kimelfeld12} and multi-tuple
deletions~\cite{Kimelfeld:2013} also focused on the same query class.
These examples offer strong indication that self-joins introduce significant
hurdles in the study of a variety of problems, and progress
in cases that account for self-joins is 
rare.\footnote{While some prior work
on related problems does allow for
self-joins~\cite{Buneman:2002,Cong12,Amarilli:2017:CQP:3034786.3056121}, the
complexity characterizations in those results are not specific to the queries, but rather to high-level operators (e.g, join, projection, etc.).
In contrast, our work provides results that are fine-grained and identify
elements of the query structure that render the resilience problem NP-complete
or PTIME-computable.}

In this paper, we give several novel results on the hardness of the resilience problem for CQs with self-joins.
We show some results that hold for any CQ with self-join but later we focus on the class of \emph{binary CQs} 
(those where relations are either unary or binary). We provide various complexity results for binary CQs
where only one relation name can be repeated, which we denote by single-self-join (ssj).
We analyze the case of ssj binary queries in general but emphasize that for the case with at most 2 instances of the 
repeated relation, we prove that a \p\  versus \np-complete dichotomy exists.  
We further provide a unifying criterion for hardness (a ``proof template''),  
and we conjecture that it subsumes and generalizes the criterion of \emph{triads} from Sj-free queries, 
and that it provides a sufficient criterion of hardness for \emph{any} CQ.

\introparagraph{Contributions and outline}

\begin{itemize}[leftmargin=8pt]
    \item Contrasting with current knowledge about the resilience of
    CQs without self-joins (summarized in
    \Cref{sec:background}), we demonstrate how self-joins complicate
    the problem and invalidate several aspects and intuitions from the self-join-free case (\Cref{sec:themess}).
    
    \item We establish foundations for tackling the resilience problem for
    {\emph{conjunctive queries with self-joins}} by identifying important
    conditions on the minimality and connectedness of queries and by revising
    the fundamental notion of query domination (\Cref{sec:general}).
    
    \item We prove that resilience for queries that contain a triad (a
    structure that characterizes hardness in the sj-free case~\cite{FreireGIM15}) remains \np-complete
    in the presence of self-joins (\Cref{sec:triadsRemainHard}).
    
    \item By narrowing our target class to the class of \emph{binary conjunctive queries} (those where relations are either unary or binary)
    and single-self-join queries
    (i.e., only one relation can appear in multiple atoms of the query), we
    identify a new structure that implies hardness, thus expanding the
    \np-complete class compared to the sj-free case
    (\Cref{sec:paths}).

    \item We identify and define the fundamental structures of chains,
    confluences, and permutations, and use them to prove a complete dichotomy
    between \np-complete and PTIME cases for the class of {single-self-join} binary conjunctive
    queries where exactly two atoms in a query correspond to the same relation
    (\Cref{sec:2R}).    
    
    \item We prove several involved results using the chains, confluences, and
    permutations structures in the case of single-self-join binary conjunctive queries where exactly
    3 atoms correspond to the same relation. While a complete dichotomy
    for this class remains elusive, our work creates a roadmap and identifies
    remaining open problems (\Cref{sec:3R}).

    \item 
	{We provide the novel concept of \emph{Independent Join Paths}.
	This general ``proof template'' aims to 
	($i$) provide a sufficient criterion of hardness for \emph{any} CQs,
	($ii$) subsume the prior hardness criterion of triads for SJ-free CQs,
	and ($iii$) provide a hint for an approach that could possibly automate the search for hardness reductions.} 	
	(\cref{sec:generalization}).
	
\end{itemize}

 {Some of our results apply to the general class of self-join CQs, while others apply to more restricted query families. 
We annotate our theoretical results with the symbols detailed in}~\cref{tb: class of queries}  {to indicate the relevant assumptions.}

\begin{table}[t!]
\centering
\begin{tabularx}{\linewidth}{  @{\hspace{0pt}} >{$}l<{$}  @{\hspace{2mm}} X @{}} 	
	\toprule
      & \textbf{Query class}  \\
      \midrule
    $\allsj$  & all self-join conjunctive queries\\
    $\ssjbin$ &  single-self-join (ssj) binary conjunctive queries\\
    \twoR & ssj binary conjunctive queries with exactly 2 $R$-atoms\\
    \threeR & ssj binary conjunctive queries with exactly 3 $R$-atoms\\
    \bottomrule
    \end{tabularx}
   \caption{Annotations specifying the relevant classes of queries. }
   \label{tb: class of queries}
\end{table}

\section{Background and Prior Results}\label{sec:background}
This section introduces our notation, defines \emph{the resilience} of a query, and 
summarizes prior complexity results for sj-free queries.

\looseness-1
\introparagraph{Standard database notations}
We use boldface to denote tuples or ordered sets, 
(e.g., $\vec x = (x_1, \ldots, x_k)$)
and use both subscripts and superscripts as indices 
(e.g., $a^1$ and $a_1$). 
We fix a relational vocabulary $\vec R = (R_1, \ldots, R_\ell)$, and denote $\arity(R_i)$ the arity of a relation $R_i$. 
We call \emph{unary} and \emph{binary} those relations with arity 1 or 2, respectively.
We call ``\emph{binary queries}'' those queries that contain only unary or binary relations.
A database instance over $\vec R$ is $D = (R_1^D, \ldots, R_\ell^D)$, where each $R_i^D$ is a finite relation.
We call the elements of $R_i^D$ tuples and
write $R_i$ instead of $R_i^D$ when $D$ is clear from the context.
With some abuse of notation 
we also denote $D$ as the set of all tuples, 
i.e.\ $D = \bigcup_i R_i$,
where the union is understood to be a disjoint union (thus each tuple belongs to only one relation).
The active domain $\dom(D)$ is the set of all constants occurring in $D$.
The size of the database instance is $n = |D|$, i.e.\ the number of tuples in the database.\footnote{Notice that other work sometimes uses $\dom(D)$ as the size of the database. Our different definition has no implication on our complexity results but simplifies the discussions of our reductions.}

A \emph{conjunctive query} (CQ) is a 
first-order formula $q(\vec y)$ $= \exists
\vec x\,(g_1 \wedge \ldots \wedge g_m)$ where 
the variables $\vec x = (x_1, \ldots, x_k)$ 
are called existential variables,
$\vec y = (y_1, \ldots, y_c)$ are called the head variables (or free variables),
and each atom (also called subgoal) $g_i$ represents a relation 
$g_i= R(\vec z_i)$ where $\vec z_i \subseteq \vec x \cup \vec y$.\footnote{WLOG, we assume that
$\vec z_i$ is a tuple of only variables and don't write the constants.
Selections can always be directly pushed into the database before executing the query.
In other words, for any constant in
the query, we can first apply a selection on each relation and then consider the modified query with
a column removed.}
A \emph{self-join-free CQ} (sj-free CQ) is one where no relation symbol occurs more than once 
and thus every atom represents a different relation. In turn, a \emph{self-join CQ} is one where at least
one relation symbol is repeated, and a \emph{single-self-join (ssj) CQ} is one where only one relation symbol
can be repeated in a query. 
We write $\var(g_j)$ for the set of variables occurring in atom $g_j$.
As usual, we abbreviate a non-Boolean query in Datalog notation by 
$q(\vec y) \datarule g_1, \ldots, g_m$
where $q$ has head variables $\vec y$
and $g_1, \ldots, g_m$ represents the body of the query.

{Unless otherwise stated, a query in this paper denotes
a \emph{Boolean} CQ  $q$ (i.e., $\vec y = \emptyset$). 
We write $D \models q$ to denote that the query $q$ evaluates to \true over the database
instance $D$, and $D \not\models q$ to denote that $q$ evaluates to \false.
For a Boolean query $q$, we write $q(\vec x)$ to indicate that $\vec x$ represents the set of all existentially
quantified variables. We write $[k]$ as short notation for the set $\{1, \ldots, k\}$.}

\introparagraph{Additional notations}

We call a valuation of all existential variables that is permitted by $D$ and that makes $q$ \true
(i.e.\ $D \models q[\vec w/\vec x]$) a \emph{witness} $\vec w$.\footnote{Note that our notion of witness slightly differs from the one commonly seen in  provenance literature where a ``witness'' refers to a subset of the input database records that is sufficient to ensure that a given output tuple appears in the result of a query \cite{DBLP:journals/ftdb/CheneyCT09}.} 
The set of witnesses is then
\[
\witnesses(D,q) = \bigset{\vec w}{D \models q[\vec w /\vec x]}\; .
\]

\noindent
Since every witness implies exactly one set of at most $m$ tuples from $D$ that make the query true, 
we will slightly abuse the notation and also refer to this set of tuples as ``witnesses.'' 
For example, consider the query 
$\chain \datarule R(x,y), R(y,z)$
with $\vec x = (x,y,z)$ over the database 
$D = \{t_{1}:R(1,2),t_{2}:R(2,3), t_{3}:R(3,3)\}$.
Then one can easily see that
$$\witnesses(D,\chain) = \{(1,2,3), (2,3,3), (3,3,3)\}$$
and their respective tuples are 
$\set{t_{1},t_{2}}$, $\set{t_{2},t_{3}}$, and $\set{t_3}$.

In line with prior work~\cite{FreireGIM15,MeliouGMS11}, relations
may be specified as \emph{exogenous}, meaning that tuples from these relations
cannot be deleted.\footnote{In other words, tuples in these atoms provide
context and are outside the scope of possible ``interventions'' in the spirit
of causality~\cite{HalpernPearl:Cause2005}.}
We specify the atoms corresponding to exogenous relations with a superscript
``$\exSymb$''. The remaining atoms are \emph{endogenous}.

\introparagraph{Complexity theory}
We write $S\leq T$ to mean $S\leq_{\textrm{fo}} T$.\footnote{First-order reductions are not required, but 
it is the case that all reductions defined in the paper are expressible in first-order.}
We say that two problems have \emph{equivalent} complexity ($S\equiv T$) iff they are inter-reducible, i.e., $S\leq T$ and $T\leq S$.

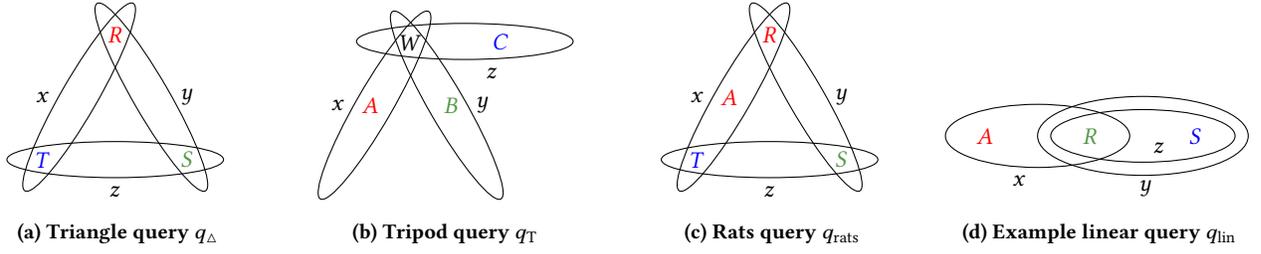
\begin{figure*}
\begin{subfigure}[b]{.24\linewidth}
	\centering
	\begin{tikzpicture}[ scale=.12]
	\draw[color=red] (-8,6.93) node {{\color{black} $x$}};  
	\draw[color=dg] (8,6.93) node {{\color{black}  $y$}};  
	\draw[color=blue] (0,-3.5) node {{\color{black}  $z$}};  
	\draw (0,13.86) node {{\color{red} $R$}};  
	\draw (-8,0) node {{\color{blue}  $T$}};  
	\draw (8,0) node {{\color{dg} $S$}};  
	\draw[rotate=-30] (-6.93,4) ellipse (2 and 12);  
	\draw[rotate=30] (6.93,4) ellipse (2 and 12);  
	\draw[rotate=90] (0,0) ellipse (2 and 12);  
\end{tikzpicture}
\caption{Triangle query $q_\triangle$}\label{fig:triangleHypergraph}
\end{subfigure}
\begin{subfigure}[b]{.24\linewidth}
	\centering
	\begin{tikzpicture}[ scale=.12]
		\draw (-8,6.93) node {$x$};  
		\draw (8,6.93) node {$y$};  
		\draw (9,10.5) node {$z$};  
		\draw (-4.5,6.93) node {{\color{red} $A$}};  
		\draw (10,14) node {{\color{blue} $C$}};  
		\draw (4.5,6.93) node {{\color{dg} $B$}};  
		\draw (0,13.8) node {{\color{black}  $W$}};  
		\draw[rotate=-30] (-6.93,4) ellipse (2 and 12);  
		\draw[rotate=30] (6.93,4) ellipse (2 and 12);  
		\draw[rotate=90] (14,-6) ellipse (2 and 12);  
	\end{tikzpicture}
\caption{Tripod query $q_\Tri$}\label{fig:tripodHypergraph} 
\end{subfigure}
\begin{subfigure}[b]{.24\linewidth}
	\centering
	\begin{tikzpicture}[ scale=.12]
	\draw (-8,6.93) node {$x$};  
	\draw (8,6.93) node {$y$};  
	\draw (0,-3.5) node {$z$};  
	\draw (-4.5,6.93) node {{\color{red}  $A$}};  
	\draw (0,13.86) node {{\color{red} $R$}};  
	\draw (-8,0) node {{\color{blue}  $T$}};  
	\draw (8,0) node {{\color{dg}  $S$}};  
	\draw[rotate=-30] (-6.93,4) ellipse (2 and 12);  
	\draw[rotate=30] (6.93,4) ellipse (2 and 12);  
	\draw[rotate=90] (0,0) ellipse (2 and 12);  
\end{tikzpicture}	
\caption{Rats query $\rats$}\label{fig:ratsHypergraph}	
\end{subfigure}
\begin{subfigure}[b]{.24\linewidth}
	\centering
	\begin{tikzpicture}[ scale=.35]
	\draw (-4,0) node {{\color{red} $A$}};
	\draw (0,0) node {{\color{dg} $ R$}};  
	\draw (4,0) node {{\color{blue} $ S$}};  
	\draw (-2.8,-1.7) node {{  $x$}};
	\draw (2,-2) node {{ $y$}};
	\draw (2.5,-.45) node {{  $z$}};
	\draw[color=black] (-2,0) ellipse (3.5 and 1.2);
	\draw[color=black] (2,0) ellipse (3.5 and 1.0);
	\draw[color=black] (2,0) ellipse (4 and 1.5);
\end{tikzpicture}
\caption{Example linear query $\lin$}\label{linear fig}
\end{subfigure}
\caption{The hypergraphs of queries $q_\triangle$, $q_\Tri$ and $\lin$.
$\set{R, S,T}$ 
is a triad of $q_\triangle$ and 
$\set{A,B,C}$ is a triad of $q_\Tri$.
Thus $\res(q_\triangle)$ and $\res(q_\Tri)$ are \np-complete.
In contrast, $A$ dominates both $R$ and $T$ in $\rats$ which renders both atoms exogenous, thus ``disarming'' what appears to be a triad.
And $\mathcal{H}(\lin)$ is linear. 
Thus $\res(\rats)$ and $\res(\lin)$ are in \p.
}\label{Fig_HardQueries}
\end{figure*}

\subsection{Query resilience}
  
In this paper, we focus on the problem of resilience, a variant of the problem of deletion propagation focusing on Boolean queries:  
Given $D\models q$, what is the minimum number of endogenous tuples that have to be removed
from $D$ to make the query false?
A large minimum set implies that the query is more ``\emph{resilient}'' and requires
the deletion of more tuples to change the query output.
In order to study the complexity of resilience, we focus on the
decision problem:

\begin{definition}[Resilience Decision]\label{def: resilience}
{Given a query $q$, database $D$, and an integer $k$. We say that $(D, k) \in \res(q)$ if and
only if $D \models q$ and there exists a set $\Gamma$ with at most $k$ endogenous tuples s.t. $D -\Gamma \not\models q$.}
{We define $\rho(D,q)$ as the size of a minimum contingency set for input $D$ and $q$.}
\end{definition}

\noindent 
In other words, $(D,k) \in \res(q)$ means that there is a set of $k$ or fewer endogenous tuples
whose removal makes the query false. 
We refer to such a set of tuples $\Gamma$ as a ``contingency set.''
Observe that, for a fixed $q$, we can talk about data complexity and
$\res(q)\in\np$
 {when} $q$ is computable in PTIME.

A central result of the prior work on resilience~\cite{FreireGIM15} is
that the complexity of resilience of an sj-free CQ can be exactly
characterized via a natural property of its \emph{dual hypergraph}
$\mathcal{H}(q)$. The hypergraph of an sj-free query $q$ is usually defined
with its vertices being the variables of $q$ and the hyperedges being the
atoms~\cite{AbiteboulHV:1995}. The dual hypergraph, $\mathcal{H}(q)$, has
vertex set $V=\{g_1,\ldots,g_m\}$, and each variable $x_i \in \var(q)$
determines the hyperedge consisting of all those atoms in which $x_i$ occurs:
$\;e_i=\{g_j\,|\,x_i\in \var(g_j)\}$. A \emph{path} in the graph is an
alternating sequence of vertices and edges, $g_1,x_1,g_2, x_2, \ldots,
g_{n-1},$ $x_{n-1},g_n$, such that for all $i$, $x_i \in \var(g_i)\cap
\var(g_{i+1})$, i.e., the hyperedge $x_i$ joins vertices $g_i$ and~$g_{i+1}$.
We explicitly list the hyperedges in the path, because more than one hyperedge
may join the same pair of vertices.  {Since we only consider dual
hypergraphs, we use the shorter term ``hypergraph'' from now on.}

\begin{example}[Hypergraphs]\label{ex:hypergraphs}
 {We illustrate the prior results with the following 4 queries and
their hypergraphs shown in} \cref{Fig_HardQueries}:
\begin{align*}
	q_\triangle &\datarule R(x,y),S(y,z),T(z,x)			&\textrm{(Triangle)}\\
	q_\Tri 		&\datarule A(x),B(y),C(z), W(x,y,z)		&\textrm{(Tripod)}	\\
	\rats 		& \datarule R(x,y),A(x),T(z,x),S(y,z) 	&\textrm{(Rats)}	\\
	\lin 	& \datarule A(x),R(x,y, z),S(y,z) 		&\textrm{(Example linear)}	
\end{align*}
\end{example}

In the remainder of this section, we summarize the intuition behind three main
constructs---\emph{triads}, \emph{domination}, and \emph{linear queries}---that lead to the result 
presented in \cref{resilience dichotomy thm}. Then, in \cref{sec:themess} we provide an exposition of how
self-joins alter or completely invalidate these prior constructs.

\subsection{Domination}
{We may mark some relations in an input database as exogenous and, the remaining relations 
are endogenous. However, some relations are ``implicitly'' exogenous. 
For example, the relation} $W$ in $q_T$ is given as endogenous,
but is never needed in minimum contingency sets.
{We next define a syntactic property, called \emph{domination}, that captures when endogenous relations are implicitly exogenous.}

\begin{definition}[Domination]\label{sj-free domination}
If a query $q$ has endogenous atoms $A,B$ such that $\var(A)\!\subset\!\var(B)$, we say that $A$
\emph{dominates} $B$.
\end{definition}	

For example, $A(x)$ dominates $W(x,y,z)$ in $q_T$.  
Whenever a contingency set contains tuples from $W$, 
they can always be replaced with a smaller than, or equal, number of tuples from $A$.

\begin{proposition}[ {Domination for resilience}~\cite{MeliouGMS11}]\label{fact: domination does not change complexity}
Let $q$ be an sj-free CQ and $q'$ the query resulting from labeling some dominated atoms as exogenous.
Then $\res(q) \equiv \res(q')$.
\end{proposition}

When studying resilience, we follow the convention that 
\emph{all dominated atoms are made exogenous}, and we consider that the normal form of a query.  
As we have seen, $A$
dominates $W$ in $q_\Tri$. 
Similarly, the atom $A$ dominates both $R$ and $T$ in $\rats$.
We
thus transform the queries so that the dominated atoms are exogenous.  Exogenous atoms have the superscript ``$\exSymb$''.
\begin{align*}
q_\Tri'	&\datarule A(x),B(y),C(z), W^\exSymb(x,y,z)		\\
\rats' 	& \datarule R^\exSymb(x,y),A(x),T^\exSymb(z,x) ,S(y,z)
\end{align*}

\noindent
\Cref{fact: domination does not change complexity} implies
that 
$\res(\rats) \equiv \res(\rats')$.

\subsection{Triads and hardness}\label{sec: sjfree triad}

We showed in \cite{FreireGIM15} that 
$\res(q_\triangle)$ and
$\res(q_\Tri)$ 
from \cref{ex:hypergraphs}
are \np-complete.
While $q_\triangle$ and $q_\Tri$ appear to be quite different, they
share a key common structural property which alone is responsible for hardness for sj-free CQs.

\begin{definition}[Triad]\label{def: triad}
A \emph{triad} is a set of three endogenous atoms, $\mathcal{T} = \set{S_0,S_1,S_2}$ 
such that for every pair $i,j$, 
there is a path from $S_i$ to $S_j$  {in $\mathcal{H}(q)$} that uses no variable occurring in the other 
atom of $\mathcal{T}$.
\end{definition}

\noindent
Intuitively, a triad is a triple of points with ``robust connectivity.''
Observe that atoms $R,S,T$ form a triad in $q_\triangle$ and atoms $A,B,C$ form a triad in
$q_\Tri$ (see \cref{Fig_HardQueries}).
For example, there is a path from $R$ to $S$ in $q_\triangle$ (across hyperedge $y$) that uses only variables (here $y$) that are not contained in the other atom ($y \not \in \var(T)$).
We showed that triads are responsible for hardness {(see} \Cref{sec: sj-free proofs} {for proof)}:

\begin{lemma}[ {Triads make $\res(q)$ hard}~\cite{FreireGIM15}]\label{hard part dichotomy}
Let $q$ be an sj-free 
CQ where all ``dominated'' atoms are exogenous. If $q$ has a triad, then $\res(q)$ is \np-complete.
\end{lemma}

\subsection{Linear queries}

A query $q$ is \emph{linear} if its atoms can be
arranged in a linear order 
s.t.\ each variable occurs in a contiguous sequence of atoms.
Geometrically, a query is linear if all of the vertices of its hypergraph can be drawn along a straight
line and all of its hyperedges can be drawn as convex regions (thus the variables form intervals on a line of relations).
For example $\lin$ is linear (see \cref{linear fig}).

It was shown in \cite{MeliouGMS11} that 
for any sj-free CQ that is linear, 
$\res(q)$ may be computed in a natural way using network flow.
Thus all such queries are easy.

If all sj-free CQs without a triad were linear, then this would complete the dichotomy theorem for
resilience.  While this is not the case, we completed the proof of \Cref{resilience dichotomy thm},
by showing that \emph{every triad-free sj-free CQ may be transformed to a linear query of equivalent
  resilience}.

\subsection{Dichotomy Theorem}

{Now we can present the full characterization of the complexity of sj-free CQs proved in }~\cite{FreireGIM15}.

\begin{theorem}[ {Dichotomy of resilience for sj-free CQs}~\cite{FreireGIM15}]\label{resilience dichotomy thm}
Let $q$ be an sj-free CQ and let $q'$ be the result of 
making all ``dominated'' atoms exogenous.
If $q'$ has a triad, then $\res(q)$ is \np-complete, otherwise
it is in PTIME.
\end{theorem}

\section{Self-joins change everything}\label{sec:themess}

\begin{figure*}[h]
\begin{subfigure}[b]{.25\linewidth}
	\centering
	\begin{tikzpicture}[ scale=.25]
	\draw (-5,0) node {$R$};
	\draw (0,0) node {$S$};
	\draw (5,0) node {$R$}; 
	\draw (-5,2) node {\color{red}{  $x$}};
	\draw (5,2) node {\color{blue}{  $y$}};
	\draw[color=red] (-2.5,0) ellipse (4 and 1.5);
	\draw[color=blue] (2.5,0) ellipse (4 and 1.5);
	\end{tikzpicture}
\caption{Hypergraph for $\vc$}
\label{fig: vc hypergraph}
\end{subfigure}
~
\begin{subfigure}[b]{.2\linewidth}
	\centering
\begin{tikzpicture}[scale=.35,
				every circle node/.style={fill=white, minimum size=7mm, inner sep=0, draw}]
	\node (1) at (0,0) []  {\color{red}{  $x$}};
	\node (2) at (3,0) []  {\color{blue}{  $y$}};
	
	\path[->, line width=1pt, auto]
		(1) 	edge[color=black] node {$S$}	(2)
		(1) 	edge[color=black, out=120, in=60, distance=1.5cm] node {$R$}	(1)
		(2) 	edge[color=black, out=120, in=60, distance=1.5cm] node {$R$}	(2)
		;

\end{tikzpicture}
\caption{Binary graph for $\vc$}
\end{subfigure}
~
\qquad
~
\begin{subfigure}[b]{.25\linewidth}
	\centering
	\begin{tikzpicture}[ scale=.25]
	\draw (-3,0) node {$R$};
	\draw (3,0) node {$ R$};  
	\draw (1,-1) node {\color{dg}{ $z$}};
	\draw (-5,-1) node {\color{red}{  $x$}};
	\draw (0,1) node {\color{blue}{  $y$}};
	\draw[color=red] (-3,0) ellipse (1.5 and 1.5);
	\draw[color=dg] (3,0) ellipse (1.5 and 1.5);
	\draw[color=blue] (0,0) ellipse (8 and 2);
	\end{tikzpicture}
\caption{Hypergraph for $\chain$}
\end{subfigure}
~
\begin{subfigure}[b]{.2\linewidth}
	\centering
	\begin{tikzpicture}[scale=.35,every circle node/.style={fill=white, minimum size=7mm, inner sep=0, draw}]
	\node (x) at (-3,0) [] {\color{red}{$x$}};
	\node (y) at (0,0) [] {\color{blue}{$y$}};
	\node (z) at (3,0) [] {\color{dg}{$z$}};

	\draw[transform canvas={yshift=+0.0ex},->, line width=1pt](x) to node[above]{$R$} (y);
	\draw[transform canvas={yshift=+0.0ex},->, line width=1pt](y) to node[above]{$R$} (z);
	\end{tikzpicture}
\caption{Binary graph for $\chain$}
\end{subfigure}
~
\caption{\emph{Hypergraphs} only represent which variables occur in a given atom, whereas \emph{binary graphs} represent
containment and position within each atom. Both concepts are illustrated here for two basic hard CQs with self-joins: $\chain$ and $\vc$. }\label{fig: binary vs hypergraph}
\end{figure*}
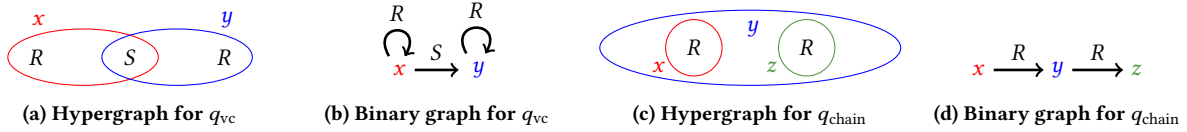

Queries with self-joins are far more complicated than sj-free queries for at least 4 reasons: 
(1) For the sj-free case, triads alone were shown to determine hardness.
Triads need at least 3 existential variables and at least 3 subgoals.
\Cref{sec:newHardQueries} shows that \emph{already 2 atoms or 2 variables} can be enough for hardness;
(2) Linear sj-free queries can be solved using a natural reduction to network flow. For self-join
queries, \emph{linear queries can be hard}.  
Furthermore, \Cref{sec: newEasyQueries} shows that 
we may need \emph{more elaborate reductions to network flow},
even when they are easy.
 {(3) The previous definition of domination does not lead to the desired properties in the presence of self-joins.}
\Cref{sec: domination is broken}  {explains why dominated atoms may still be relevant when computing the minimum contingency set.}
(4) Our previous crucial concept of the dual hypergraph is no longer sufficient to characterize
queries when relations appear multiple times.  The \emph{position at which a variable appears} in a
subgoal may influence the complexity of resilience, including whether an atom has \emph{repeated variables},
e.g., ``$R(x,y), R(y,y)$.''

{In the cases where the variable position is relevant and we are restricted to binary queries, we naturally represent
queries as labeled direct graphs.} This representation captures
all relevant structural information of the binary queries, especially the relative position of variables,
which the hypergraph representation does not reflect.

\begin{definition}[Binary graph]\label{def: binary graph}
Let $q \datarule A_1,$ $\ldots, A_m$ be a binary CQ.
Its \emph{binary graph} has vertex set $V=\var(q)$ and labeled edge sets defined by atoms $ A_1,$ $\ldots, A_m$,
i.e.\ atom $A(x,y)$ translates into 
 {labeled edge~$x \xrightarrow{A} y$.}
For unary atoms, the edge will be a loop.
\end{definition}

\subsection{Basic hard queries: $\vc$ and $\chain$}\label{sec:newHardQueries}

We start by
proving hardness for two queries that will play an important role in our later results.
The first $\vc$ (for ``vertex cover'') has only 2 variables and 3 atoms.
The second $\chain$ (since it ``chains'' two binary relations together) has only 2 atoms and 3 variables:
\begin{align*}
	\vc 	&\datarule R(x), S(x,y), R(y) 	&\textrm{(Vertex cover)}\\
	\chain 	&\datarule R(x,y),R(y,z) 	&\textrm{(Chain query)}
\end{align*}
\Cref{fig: binary vs hypergraph}  {shows graphical representations of both queries
while illustrating the differences between the \emph{dual hypergraph} and the \emph{binary graph} of a binary CQ.}

Recall that 
in the sj-free case, a query needs a triad to be hard and all linear queries are easy.
In particular, an sj-free query must have at least 3 variables and 3 atoms to be hard.

\begin{proposition}[$\vc$]\label{prop:hardness vc}
$\res(\vc)$ is \np-complete.
\end{proposition}

\begin{proposition}[$\chain$]\label{chain is npc} 
	\quad $\res(\chain)$ is NP-complete.
\end{proposition}

\subsection{SJ-Free domination no longer works}\label{sec: domination is broken}

We saw from \cref{fact: domination does not change complexity}
that in sj-free CQs, making all dominated atoms exogenous leaves the query resilience unchanged.
In the presence of self-joins, this is no longer true.

\begin{example}\label{ex: domination does not work}
Query $\rats^{\textrm{sj}_1} \datarule A(x),R(x,y),R(y,z),R(z,x)$ is a self-join variation of $\rats$ with $S,T$ replaced by $R$'s.
{Similar to $\rats$, we have $\var(A) \subseteq \var(R(x,y))$, so $A$ dominates $R$ by} \cref{sj-free domination}. 
 {}{Thus $R$ should become exogenous when searching for the minimal contingency set. }
But this is not the case. Consider the database instance
$$
D = \set{A(1), A(5), R(1,2), R(2,3), R(3,1), R(5,1), R(2,5)}
$$
Our query has 3 witnesses over this database: $(1,2,3)$, $(1,2,5)$, and $(5,1,2)$.
{If $R$ was made exogenous, the only possible minimum contingency set would be $\Gamma = \{A(1), A(5)\}$.
However, if $R$ is considered as endogenous, there is a smaller contingency set, with only $R(1,2)$.}

\end{example}

This example shows that domination as defined in \cref{sj-free domination} no longer implies  {that} a relation can be made exogenous
in the self-join case. This immediately raises the question of whether there is a set of conditions 
 {which implies that}
a relation can be made 
exogenous in the self-join setting, i.e. if there is a self-join version of domination. 
Additionally, does $\rats^{\textrm{sj}_1}$ have a triad?
The answer to both is yes, as we will see in \cref{{sec:domination}} and \cref{sec: sj rats and brats}, respectively.

\subsection{Easy queries that use flow in a trickier way}\label{sec: newEasyQueries}

As mentioned in the discussion of \Cref{resilience dichotomy thm}, resilience for linear sj-free CQ
can be computed directly from network flow.
As we have just seen, in the presence of self-joins, some linear
queries are hard.  For those that are easy,  network flow can still
help us compute resilience, but the arguments become trickier. 

The following queries are two such examples, where modified versions of network flow are used 
to show resilience is easy in these cases.
\begin{align*}
q_{\textrm{conf}}^{AC} &\datarule A(x),R(x,y),R(z,y), C(z)\\ 
q_{\textrm{3perm-R}}^{A} &\datarule A(x), R(x,y), R(y,z), R(z,y) 
\end{align*}

\begin{figure}
\begin{subfigure}[b]{.45\linewidth}
\centering
\begin{tikzpicture}[scale=.35,
			every circle node/.style={fill=white, minimum size=7mm, inner sep=0, draw}]
\node (1) at (0,0) [] {$x$};
\node (2) at (3,0) [] {$y$};
\node (3) at (6,0) [] {$z$};

\draw[transform canvas={yshift=+0.0ex},->, line width=1pt, out=120, in=60, distance=1.5cm](1) to node[above]{$A$} (1);
\draw[transform canvas={yshift=+0.0ex},->, line width=1pt, out=120, in=60, distance=1.5cm](3) to node[above]{$C$} (3);

\path[->, line width=1pt, auto]
	(1) 	edge[color=black] node {$R$}	(2)
	(3) 	edge[color=black] node[above] {$R$}	(2)	
	;
\end{tikzpicture}
\caption{$q_{\textrm{conf}}^{AC}$}
\end{subfigure}
~	
\begin{subfigure}[b]{.45\linewidth}
\centering
\begin{tikzpicture}[scale=.35,every circle node/.style={fill=white, minimum size=7mm, inner sep=0, draw}]
\node (x) at (0,0) [] {$x$};
\node (y) at (3,0) [] {$y$};
\node (z) at (6,0) [] {$z$};

\draw[transform canvas={yshift=+0.5ex},->, line width=1pt](y) to node[above]{$R$} (z);
\draw[transform canvas={yshift=-0.5ex},->, line width=1pt](z) to node[below]{$R$} (y);
\draw[transform canvas={yshift=+0.0ex},->, line width=1pt, out=120, in=60, distance=1.5cm](x) to node[above]{$A$} (x);
\draw[transform canvas={yshift=+0.0ex},->, line width=1pt](x) to node[above]{$R$} (y);
\end{tikzpicture}
\caption{$q_{\textrm{3perm-R}}^{A}$}
\label{fig: easy ks}
\end{subfigure}
\caption{Two example of PTIME queries that require modified version of network flow.
Notice that $q_{\textrm{3-perm-R}}^{A}$ contains the hard query $\chain$ 
and is still in P.}
\end{figure}
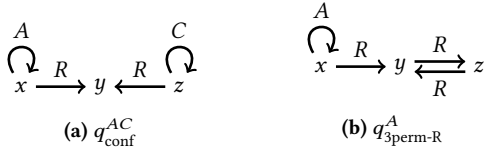

\begin{proposition}[$q_{\textrm{conf}}^{AC}$]\label{2conf is easy}
$\res(q_{\textrm{conf}}^{AC})$ is in \p.
\end{proposition}

\begin{proposition}[$q_{\textrm{3perm-R}}^{A}$]\label{AR perm}
$\res(q_{\textrm{3perm-R}}^{A})$ is in \p.
\end{proposition}

\section{New general observations and plan of attack}\label{sec:general}

We next give 3 new general observations before we describe our plan of attack in the remainder of the paper.

\subsection{Minimal queries}\label{sec:minimization}
Given queries $q_1$ and $q_2$, 
we say that $q_1$ is \emph{contained} in $q_2$ ($q_1 \subseteq q_2$)
if answers to $q_1$ over any database instance $D$ are always a subset of the answers to $q_2$ over $D$.
We say $q_1$ is \emph{equivalent} to $q_2$ ($q_1 \equiv q_2$)
if $q_1 \subseteq q_2$ and $q_2 \subseteq q_1$~\cite{AbiteboulHV:1995}.
We say a conjunctive query $q$ is \emph{minimal} if for every other conjunctive query $q'$ such that $q \equiv q'$, $q'$ has at least as many atoms as $q$.
For every query $q$, there exists a minimal equivalent CQ $q'$ that can be obtained from $q$ by removing zero or more atoms~\cite{DBLP:conf/stoc/ChandraM77}.

From now on, we focus only on \emph{minimal queries}.
This is WLOG, since any non-minimal query can always be minimized as a pre-processing step.
The reason is that our hardness evaluation relies on identifying certain subqueries (or patterns) in
a query that make this query hard. 
However, if a pattern is in a subquery that is removed during minimization, then, this pattern has
no effect on the resilience of the query.

\subsection{Query components}\label{sec:components}

A \emph{connected component} of $q$ (or ``{component}'' in short) is a non-empty subset of atoms that are connected via existential variables.
A query $q$ is \emph{disconnected} if its atoms can be partitioned into two or more
components that do not share any existential variables. For example, 
\begin{align*}
q_\textrm{comp} 	&\datarule A(x), R(x,y), R(z,w), B(w)
\intertext{is disconnected and has
two components:}
q_\textrm{comp}^1 	&\datarule A(x), R(x,y) \\
q_\textrm{comp}^2 	&\datarule R(z,w), B(w).
\end{align*}

The resilience of a query is determined by taking the minimum of the resiliences of each of its components.
In the following, let $\rho(q,D)$ stand for the resilience of query $q$ over database $D$, which is the size of the minimum
contingency set for $(q,D)$.

\begin{lemma}[$\allsj$ Query components]\label{min gamma disconnected}
Let $q \datarule q_1, \ldots, q_k$ be a query that consists of 
	$k$ components $q_i$, $i \in [k]$.
	Then $\rho(q, D) = \min_{i} \rho(q_i, D)$.
\end{lemma}

We can now show that the complexity of a query is determined by the hardest of its components
if the query is minimal:

\begin{lemma}[$\allsj$ Query components complexity]\label{complexity disconnected queries}
Let $q$ be a minimal query that consists of $k$ query components.
$\res(q)$ is \np-complete if there is at least one component $i \in [k]$
for which $\res(q_i)$ is \np-complete. Conversely, if $\res(q_i)$ is in \p\ for all $i$, then 
$\res(q)$ is in~\p.
\end{lemma}

In the remainder of the paper we assume queries to be \emph{connected}.

\subsection{SJ-domination}\label{sec:domination}

As discussed in \Cref{sec: domination is broken}, 
we need to consider the position of the variables in the attribute list of each atom in a sj-query. 
We write $\pos^q_{g}(i) = x$ to express that 
the $i$-th attribute of atom $g$ is variable~$x$ for a query $q$
and omit $q$ when $q$ is clear from the context.

\begin{definition}[ {Domination with Self-Joins}]\label{def:domination}
Let relations $A$ and $B$ be endogenous relations in query $q$. 
We say that $A$ \emph{dominates} $B$ if
there exists a function $$f: [\arity(A)] \to [\arity(B)]$$ such that for each $B$ {atom} $g_B$, there exists an $A$ {atom}
$h_A$ satisfying ${pos}_{h_A}(i) = {pos}_{g_B}(f(i))$, $\forall i \in [\arity(A)]$.
{In other words, each $B$ atom that occurs in $q$ has a corresponding $A$ atom, and each of these pairs will have
matching variables accordingly to function $f$.}
\end{definition}

Notice that when $B$ appears only once, 
the definition of domination is equivalent to the sj-free definition:  $\var(A) \subseteq \var(B)$.

\begin{example}
To illustrate the new self-join domination, consider the following queries:
\begin{align*}
q_1 &\datarule R(x,y),A(y),R(y,z), S(y,z) \\
q_2 &\datarule R(x,y),A(y),R(z,y), S(y,z)
\end{align*}

By following the definition above, $A$ doesn't dominate $R$ in $q_1$ but it does in $q_2$, whereas $S$ is dominated in both queries.
Notice that in~$q_2$, a tuple $R(a,b)$ will always join with tuple $A(b)$ so we can always choose $A(b)$ instead to be in the contingency set.
The same is not true for $q_1$, where a tuple $R(a,b)$ could join with $A(a)$ or $A(b)$.
\end{example}

\begin{proposition}[$\allsj$  {Domination for resilience with Self-Join}]\label{lem: domination sj}
Let $q$ be a CQ and $q'$ the result of labeling some dominated relations exogenous. 
Then $\res(q)\equiv \res(q')$.
\end{proposition}

\subsection{Outline of our plan of attack}

To obtain a dichotomy result for the resilience of binary queries in the presence of a single-self-join, we proceed as follows.
(1)~\Cref{sec:nonLinearQueries} shows that triads in any conjunctive queries with self-joins
still imply hardness 
(\Cref{thm: triads in sj}) and furthermore, when triads are absent, 
the endogenous atoms are linearly connected.  We call such queries \emph{pseudo-linear} (\Cref{thm: no triad means linear}).
We conjecture that pseudo-linear queries may be transformed to linear queries of equivalent
resilience (Conjecture~\ref{conj: no triad means linear}).
In any case, it suffices to study the criteria for hardness of pseudo-linear queries.
(2)~\Cref{sec:paths} generalizes the hardness pattern behind $\vc$ to a more general class of hard ssj binary queries that contain ``paths'' between repeated atoms.
(3)~We then focus on the complexity of the resilience of ssj binary CQs with at most a single repetition of a single relation.  
\Cref{sec:2R} gives a complete characterization of the complexity 
for the cases of 2 occurrences of a repeated relation. This is a dichotomy theorem:
we show that for all such queries, $q$, $\res(q)$ is either $\np$-complete or $\res(q)$ is reducible to
network flow and is thus in P. \Cref{sec:3R} presents the remaining challenges that must be overcome in order to
characterize all queries with 3 occurrences of a repeated relation.  
In \cref{sec:generalization}, we present a ``template'' for hardness proofs that we believe will help us make progress in the general self-join case.

\section{Non-linear Queries: NP-Complete}\label{sec:nonLinearQueries}

In this section we prove that queries containing triads remain hard in the presence of self-joins (\Cref{thm: triads in sj}).
We then show that for any query that does not contain a triad, its endogenous atoms are arranged
linearly.  We call such a query \emph{pseudo-linear}.
Thus, we conclude that either a query
contains a triad in which case its resilience problem is \np-complete, or
it is pseudo-linear.  In the following sections, we can thus safely restrict our
attention to pseudo-linear queries.

\begin{definition}[Self-join variation of a CQ]\label{def: sj variation}
{Let $q$ be a sj-free CQ and let $q^\textrm{sj}$ result from $q$ by replacing some atoms $S_i(\ov{v})$ from $q$
with the atom $R_i(\ov{v})$, where the relation~$R_i$ occurs elsewhere in $q$. We say that $q^\textrm{sj}$
is a \emph{self-join variation} of $q$.} 
\end{definition}

\begin{example}[Self-join variations]
Consider sj-free query $q_\triangle$. The following are all its possible self-join variations:
\begin{align*}
q_\triangle &\datarule R(x,y),S(y,z),T(z,x)			&\textrm{(Triangle)}\\
q^{\sj_1}_\triangle	&\datarule R(x,y),R(y,z),R(z,x)			&\\
q^{\sj_2}_\triangle	&\datarule R(x,y),R(y,z),T(z,x)			&\\
q^{\sj_3}_\triangle	&\datarule R(x,y),S(y,z),R(z,x)			&
\end{align*}
\end{example}

 {}{We first observe that the resilience of self-join variations of a query can only be harder than their sj-free counterpart:}

\begin{lemma}[$\allsj$ SJ Can Only Make Resilience Harder]\label{lem: sj makes harder}
 {}{Let $q$ be an sj-free CQ and let $q^\sj$ be a self-join variation of $q$.
If $q^{\sj}$ is minimal, then $\res(q) \leq \res(q^{\sj})$.  }
\end{lemma}

We need to rely on the fact that $q^{\sj}$ is minimal for the result to hold, as we see in the example below:
\begin{example}\label{counterexample}
Consider query $q\datarule R(x,y), S(z,y), T(z,w), A(x,w)$, and observe that $\res(q)$ is \np-complete
because $q$ contains a triad. A possible self-join variation is $q^{\sj}\datarule R(x,y), R(z,y), R(z,w), R(x,w)$.
Note that $q^{\sj}$ is not minimal, and is equivalent to $R(x,y)$. So $\res(q^{\sj})$ is trivially in \p. 
\end{example}

By \cref{lem: sj makes harder},  {}{the resilience of the self-join variations of $q_\triangle$ are  \np-complete.
Recall from} \cref{def: triad} {that a triad is set of three endogenous atoms, so we can say that the self-join variations of $q_\triangle$
all have triads. However, it does not immediately follow from} \cref{lem: sj makes harder}  {}{that every sj-query with a triad
is hard. The missing cases are when an sj-query includes a triad, but it is not a self-join variation of
an sj-free query with a triad.} We next explore this situation.

\subsection{Self-join variations of  $\rats$ and $\brats$}\label{sec: sj rats and brats}

Recall  {two} important sj-free queries:
\begin{align*}
\rats & \datarule R(x,y),A(x),T(z,x),S(y,z) &\textrm{(Rats)}\\
\brats & \datarule B(y),R(x,y),A(x),T(z,x),S(y,z) &\textrm{(Brats)}
\end{align*}

 {$\res(q_\triangle)$ is \np-complete because it contains the  triad, $R,S,T$.  
However, $\rats$ and $\brats$ are easy because $A$ dominates $R,T$ and $B$ 
dominates $S$ so they only have two endogenous atoms each and thus no triad.

The same doesn't occur with some self-join variations of $\rats$ and $\brats$. 
Below we list two example of variations which contain triads.}
\begin{align*}
\rats^{\textrm{sj}_1} &\datarule R(x,y),A(x),R(y,z),R(z,x)\\ 
\brats^{\textrm{sj}_1} &\datarule B(y),R(x,y),A(x),R(z,x),R(y,z)
\end{align*}

 {In these examples, relation~$R$ is now more robust and not dominated by $A$ or $B$.
Therefore, they still contain triads consisting of their three~$R$-atoms. 
The presence of a triad is a strong indication that these queries are hard but we cannot use} \cref{lem: sj makes harder}
 {to show this because their sj-free counterparts, $\rats$ and $\brats$, are easy. We now proceed to show their complexity is hard.}

\begin{proposition}\label{prop: sj rats and brats are hard}
Let $q$ be a self-join variation of $\rats$ or $\brats$. If~$q$ has a triad, then $\res(q)$ is \np-complete.
\end{proposition}

Using \Cref{prop: sj rats and brats are hard}, we now generalize the fact that triads make sj-free
queries hard (\Cref{hard part dichotomy}) to the same result for general CQs.

\subsection{Triads Make Queries Hard}\label{sec:triadsRemainHard}

\begin{theorem}[$\allsj$ SJ-queries with triads]\label{thm: triads in sj}
If $q$ has a triad, then $\res(q)$ is $\np$-complete.
\end{theorem}

\begin{proof}[ {Proof Sketch}]
We argue that there are only two cases to consider when a sj-query $q$ has a triad. Consider that
$q$ is a self-join variations of a sj-free query $q'$.
The first case is when the resilience of $q'$ is hard, i.e., $\res(q')$ is \np-complete. 
Then, we can use \cref{lem: sj makes harder} to show $\res(q)$ is \np-complete.
The second case is when $\res(q')$ is in \p. Then, we argue that we can show 
a reduction from $\res(q'') \leq \res(q)$, where $q''$ is a self-join variation of either $\rats$ or $\brats$ that
contains a triad. From \cref{prop: sj rats and brats are hard}, $\res(q'')$ is \np-complete, which thus implies that
$\res(q)$ is also \np-complete.
\end{proof}

Thus, if a query contains a triad, it is hard.  
 {In the next section, we discuss queries that do not contain triads and how they
are similar to linear queries, since their endogenous atoms are linearly connected.}

\subsection{No Triad Means Pseudo-Linear}\label{sec:pseudoLinear}

In \cite{FreireGIM15}, we proved that if a sj-free CQ $q$ has no triad, then
$q$ may be transformed to a sj-free CQ query $q'$ which is linear and such that $\res(q) \leq
\res(q')$.  Since linear sj-free CQ's are easy, it follows that $q$ is easy.

This argument no longer works in the presence of self-joins because linear queries can be easy or
hard.  However, we can extend the theorem from \cite{FreireGIM15} to show the
following,

\begin{theorem}[$\allsj$ No Triad Means Pseudo-Linear]\label{thm: no triad means linear} Let $q$ be a CQ
  with no triad.  Then all endogenous atoms in $q$ are connected linearly.
\end{theorem}

{We conjecture that pseudo-linearity is equivalent to linearity when considering resilience. 
What makes a query pseudo-linear, instead of linear or containing a triad, is the presence of some exogenous atoms. 
However, the exogenous atoms of a query are mostly only connecting the endogenous atoms,
and also, if necessary, ensuring that $q$ is a minimal query, so we believe they can be modified to obtain a linear query without
altering the complexity of a query.}

\begin{conjecture}[$\allsj$ No Triad Means Linear]\label{conj: no triad means linear} 
Let $q$ be a CQ
  with no triad.  Then we can transform $q$ to a linear CQ $q'$ with $\res(q) \equiv
  \res(q')$.
\end{conjecture}

\section{Paths are hard}\label{sec:paths}

\Cref{sec:newHardQueries} presented two linear queries that are hard, unlike in the
sj-free case where all linear queries are easy.  {Note that these queries are binary and, in both, only
one relation is part of a self-join. In other words, these are single-self-join binary queries.}

We now identify a pattern characteristic of $\vc$ that we call a \emph{path}.   {The main result of this section is that every ssj binary 
query containing a path is hard. We start by showing the case where the self-join relation is unary.}

\begin{theorem}[$\ssjbin$ Unary path]\label{unary path}
Let $q$ be a minimal ssj-CQ.  If $q$ contains distinct atoms $R(x)$ and $R(y)$, then $\res(q)$ is \np-complete.
\end{theorem}

\begin{proof}[Proof sketch]
Let $R(x)$ and $R(y)$ be the first two occurrences of the relation $R$ in $q$. Since $q$ is
connected, $R(x)$ and $R(y)$ are connected by at least one non-self-join relation, $S$ 
(see \cref{fig: unary path}).
We prove that $\res(\vc) \leq \res(q)$.  Details are in \Cref{sec: proofs}, but it is not hard to see that
any database $D \models \vc$ can be transformed to a database $D' \models q$ that exactly preserves
resilience.  Here $R', S'$ in~$D'$ come from $A$ and $R$ in $D$, and all the other atoms of $q$
(including any additional occurrences of the self-join relation, $R$, to the right of $R(y)$) are
covered by multiple, extra values which complete the joins but are never chosen in minimum
contingency sets. 
Note that this proof doesn't make any assumption about the existence or not of triads.
\end{proof}

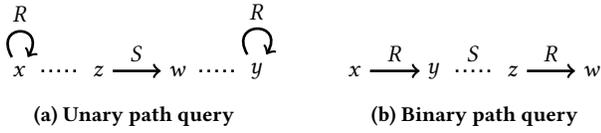
\begin{figure}[t]
\begin{subfigure}[t]{.45\linewidth}
\centering
\begin{tikzpicture}[scale=.35,
				every circle node/.style={fill=white, minimum size=7mm, inner sep=0, draw}]
	\node (1) at (0,0) [] {$x$};
	\node (1a) at (0.5,0) [] {};
	\node (1b) at (2.5,0) [] {};
	\node (2) at (3,0) [] {$z$};
	\node (3) at (6,0) [] {$w$};
	\node (3a) at (6.5,0) [] {};
	\node (3b) at (8.5,0) [] {};
	\node (4) at (9,0) [] {$y$};

	\path[->, line width=1pt, auto]
		(1) 	edge[color=black, out=120, in=60, distance=1.5cm] node {$R$}	(1)
		(4) 	edge[color=black, out=120, in=60, distance=1.5cm] node {$R$}	(4)
		(2) 	edge[color=black] node {$S$}	(3)
		(1a) 	edge[color=black,-, dotted] node {}	(1b)
		(3a) 	edge[color=black, -, dotted] node {}	(3b)	
		;

\end{tikzpicture}
\caption{Unary path query}
\label{fig: unary path}
\end{subfigure}
\qquad
\begin{subfigure}[t]{.45\linewidth}
\centering
\begin{tikzpicture}[scale=.35,
				every circle node/.style={fill=white, minimum size=7mm, inner sep=0, draw}]
	\node (1) at (0,0) [] {$x$};
	\node (2) at (3,0) [] {$y$};
	\node (3) at (3.5,0) [] {};
	\node (4) at (5.5,0) [] {};
	\node (5) at (6,0) [] {$z$};
	\node (6) at (9,0) [] {$w$};

	\path[->, line width=1pt, auto]
		(1) 	edge[color=black] node {$R$}	(2)
		(5) 	edge[color=black] node {$R$}	(6)
		(3) 	edge[color=black, -, dotted] node {$S$}	(4)	
		;

\end{tikzpicture}
\caption{Binary path query}
\label{fig: binary path}
\end{subfigure}
\caption{General structure of path queries.}
\label{fig: paths}
\end{figure}

When the self-join relation is binary, 
if two consecutive atoms, $R(x,y)$, $R(z,w)$, are disjoint, 
then we
call this a \emph{binary path}.  ``Overlapping'' consecutive atoms with shared variables, such as
$R(x,y)$, $R(y,z)$ in 
$\chain$, can also cause hardness and are studied in later sections.

\begin{theorem}[$\ssjbin$ Binary path]\label{binary path}
Let $q$ be a minimal ssj-CQ. If $q$ has distinct consecutive sj atoms $R(x,y), R(z,w)$ with
$\set{x,y}\cap\set{z,w}=\emptyset$,
then $\res(q)$ is $\np$ complete.
\end{theorem}
\begin{proof}[Proof Sketch]
Given $R(x,y), R(z,w)$ as in the statement of the theorem, there must be an atom $S(u,v)$, with
$S\ne R$ on the
path between them, and $u\in \set{x,y}$ and $v \not\in \set{x,y}$.  
Now, as in the proof of \Cref{unary path}, we reduce $\res(\vc)$ to $\res(q)$.  We map any database
$D\models \vc$ to a database $D' \models q$, where  $R'$ contains $\bigset{(a,a)}{A(a) \in~D}$ 
plus other multiple, extra values for any other atoms of the relation $R$ in $q$ to the left
of $R(x,y)$ or the right of $R(z,w)$ and $S' = \bigset{(a,b)}{R(a,b) \in D}$. 
Same as in the unary case, there is no assumption about the linearity of the query.
Details are in \Cref{sec: proofs}.
\end{proof}

Unary and Binary Paths are the simplest of the hard patterns.   By \Cref{unary path} and \Cref{binary
  path}, they always force their queries to be hard.  

 {Since we have established that an sj- query either has a triad or is pseudo-linear} (\cref{thm: no triad means linear})
 {and because we have proved that triads imply hardness} (\cref{thm: triads in sj}), 
 {we can now focus on the pseudo-linear queries}.

In the next sections we study the more subtle pseudo-linear ssj binary queries, which do not contain paths.

\section{Queries with exactly two $R$-atoms}\label{sec:2R}

\begin{figure*}[t]
\centering
\begin{tabular}{  | c | c | c | c |  }

\hline
 Pattern Name	& Binary Graph	& PTIME cases	& NP-hard cases \\ 
\hline
	$\chain$
		&
	\begin{tikzpicture}[scale=.35,
				every circle node/.style={fill=white, minimum size=7mm, inner sep=0, draw}]
	\node (1) at (0,0) [] {$x$};
	\node (2) at (3,0) [] {$y$};
	\node (3) at (6,0) [] {$z$};

	\path[->, line width=1pt, auto]
		(1) 	edge[color=black] node {$R$}	(2)
		(2) 	edge[color=black] node {$R$}	(3)
		;

	\end{tikzpicture}
	
	&
	No PTIME case
	&
	\begin{tikzpicture}[scale=.35,
				every circle node/.style={fill=white, minimum size=7mm, inner sep=0, draw}]
	\node (1) at (0,0) [] {$x$};
	\node (2) at (3,0) [] {$y$};
	\node (3) at (6,0) [] {$z$};
	\node (4) at (9,0) [] {};
	\node (5) at (-3,0) [] {};

	\path[->, line width=1pt, auto]
		(1) 	edge[color=black] node {$R$}	(2)
		(2) 	edge[color=black] node {$R$}	(3)
		(1) 	edge[-, color=black, dotted] node {}	(5)
	(3) 	edge[-, color=black, dotted] node {}	(4)
		
		;
    \path[->, dotted, line width=1pt, color=black, out=120, in=60, distance=1.5cm, auto]
		(1) 	edge[color=black] node {$A$}	(1)
		(2) 	edge[color=black] node {$B$}	(2)
        (3) 	edge[color=black] node {$C$}	(3)
		
		;

	\end{tikzpicture}
\\
\hline
	$\conv$
    &
	\begin{tikzpicture}[scale=.35,
				every circle node/.style={fill=white, minimum size=7mm, inner sep=0, draw}]
	\node (1) at (0,0) [] {$x$};
	\node (2) at (3,0) [] {$y$};
	\node (3) at (6,0) [] {$z$};

	\path[->, line width=1pt, auto]
		(1) 	edge[color=black] node {$R$}	(2)
		(3) 	edge[color=black] node[above] {$R$}	(2)
		;

	\end{tikzpicture}
	&
	\begin{tikzpicture}[scale=.35,
			every circle node/.style={fill=white, minimum size=7mm, inner sep=0, draw}]
\node (1) at (0,0) [] {$x$};
\node (2) at (3,0) [] {$y$};
\node (3) at (6,0) [] {$z$};
\node (4) at (9,0) [] {};
\node (5) at (-3,0) [] {};

\path[->, line width=1pt, auto]
	(1) 	edge[color=black] node {$R$}	(2)
	(3) 	edge[color=black] node[above] {$R$}	(2)	
	(1) 	edge[color=black, out=120, in=60, distance=1.5cm] node {$A$}	(1)
	(2) 	edge[color=black, out=120, in=60, distance=1.5cm, dotted] node {$B$}	(2)
	(3) 	edge[color=black, out=120, in=60, distance=1.5cm] node {$C$}	(3)
	(1) 	edge[-, color=black, dotted] node {}	(5)
	(3) 	edge[-, color=black, dotted] node {}	(4)

	;
\end{tikzpicture}
	&
	\begin{tikzpicture}[scale=.35,
			every circle node/.style={fill=white, size=7mm, inner sep=0, draw}]

    \node (x) at (-9,0) [] {$x$};
    \node (y) at (-6,0) [] {$y$};
    \node (z) at (-3,0) [] {$z$};
    \node (1) at (0,0) [] {};
    \node (2) at (-12,0) [] {};

    \draw[transform canvas={yshift=+0ex},->, line width=1pt](x) to node[above]{$R$} (y);
    \draw[transform canvas={yshift=+0ex},->, line width=1pt](z) to node[above]{$R$} (y);

    \path[-, line width=1pt, auto]
	(x)	edge[dotted, color=black, out=120, in=60, distance=1.5cm] node {$H^{x}$}	(z)
	(2) 	edge[-, color=black, dotted] node {}	(x)
	(1) 	edge[-, color=black, dotted] node {}	(z)
	;
	
\end{tikzpicture}

\\
\hline
	$\perm$	
	&
	
	\begin{tikzpicture}[scale=.35,every circle node/.style={fill=white, minimum size=7mm, inner sep=0, draw}]
\node (x) at (0,0) [] {$x$};
\node (y) at (3,0) [] {$y$};

\draw[transform canvas={yshift=+0.5ex},->, line width=1pt](x) to node[above]{$R$} (y);
\draw[transform canvas={yshift=-0.5ex},->, line width=1pt](y) to node[below]{$R$} (x);
\end{tikzpicture}	
	&
	\begin{tikzpicture}[scale=.35,every circle node/.style={fill=white, minimum size=7mm, inner sep=0, draw}]
\node (x) at (0,0) [] {$x$};
\node (y) at (3,0) [] {$y$};
\node (z) at (-3,0) [] {};

\draw[transform canvas={yshift=+0.0ex},->, line width=1pt, out=120, in=60, distance=1.5cm, dotted](x) to node[above]{$A$} (x);
\draw[transform canvas={yshift=+0.5ex},->, line width=1pt](x) to node[above]{$R$} (y);
\draw[transform canvas={yshift=-0.5ex},->, line width=1pt](y) to node[below]{$R$} (x);
 \path[-, line width=1pt, auto]
	(x) 	edge[-, color=black, dotted] node {}	(z)
	;

\end{tikzpicture}
	&
	\begin{tikzpicture}[scale=.35,every circle node/.style={fill=white, minimum size=7mm, inner sep=0, draw}]
\node (x) at (0,0) [] {$x$};
\node (y) at (3,0) [] {$y$};
\node (z) at (-3,0) [] {};
\node (w) at (6,0) [] {};

\draw[transform canvas={yshift=+0.0ex},->, line width=1pt, out=120, in=60, distance=1.5cm](x) to node[above]{$A$} (x);
\draw[transform canvas={yshift=+0.5ex},->, line width=1pt](x) to node[above]{$R$} (y);
\draw[transform canvas={yshift=-0.5ex},->, line width=1pt](y) to node[below]{$R$} (x);
\draw[transform canvas={yshift=+0.0ex},->, line width=1pt, out=120, in=60, distance=1.5cm](y) to node[above]{$B$} (y);
\path[-, line width=1pt, auto]
	(x) 	edge[-, color=black, dotted] node {}	(z)
	(y) 	edge[-, color=black, dotted] node {}	(w)
	;
\end{tikzpicture}

\\

\hline
	REP
	&
	
	\begin{tikzpicture}[scale=.35,
				every circle node/.style={fill=white, minimum size=7mm, inner sep=0, draw}]
	\node (1) at (0,0) [] {$x$};
	\node (2) at (3,0) [] {$y$};

	\path[->, line width=1pt, auto]
		(1) 	edge[color=black,out=120, in=60, distance=1.5cm] node {$R$}	(1)
		(1) 	edge[color=black] node {$R$}	(2)
		
		;

	\end{tikzpicture}
	&
	\begin{tikzpicture}[scale=.35,
				every circle node/.style={fill=white, minimum size=7mm, inner sep=0, draw}]
	\node (1) at (0,0) [] {$x$};
	\node (2) at (3,0) [] {$y$};
	\node (3) at (-3,0) [] {};
	\node (4) at (6,0) [] {};

	\path[->, line width=1pt, auto]
		(1) 	edge[color=black,out=120, in=60, distance=1.5cm] node {$R$}	(1)
		(1) 	edge[color=black] node {$R$}	(2)
		(2) 	edge[color=black,out=120, in=60, distance=1.5cm] node {$A$}	(2)
		(1) 	edge[-,dotted, color=black] node {}	(3)
		(4) 	edge[-,dotted, color=black] node {}	(2)

		;

	\end{tikzpicture}
	&
	No NP-hard case

\\

\hline
\end{tabular}

\caption{Binary graphs representing all the possible self-join patterns with two $R$-atoms, where $R$-atoms share at least one variable.}
\label{Fig_Overview}
\end{figure*}

In this section we cover the complexity of pseudo-linear ssj binary queries with exactly two atoms referring to the same relation. We will refer to this relation 
as $R$. As always, we assume that our query is minimal and connected,
and from now on also assume that $q$ does not contain a triad or a path as described in \Cref{unary path} and \Cref{binary path}; otherwise we would already know
that $\res(q)$ is $\np$-complete.
Even in this restricted setting, 
we will see that there is a surprisingly rich variety of structures, requiring different strategies
to determine their complexity.

Because there are no paths, $R$ must be a binary relation and the two $R$-atoms must have at least one variable in
common.
\begin{itemize}
\item \emph{Chains} have one common variable and join in different attributes, e.g., $R(x,y), R(y,z)$;
\item \emph{Confluences} have one common variable and join in the same attribute, e.g., $R(x,y), R(z,y)$;
\item \emph{Permutations} share both variables but join in different attributes, e.g., $R(x,y), R(y,x)$.
\item Queries with \emph{repeated variables (REP)} have repeated variables in at least one $R$-atom e.g., $R(x,x), R(x,y), B(y)$
\end{itemize}

\Cref{Fig_Overview} shows the binary graphs for each these patterns, which helps visualize the subtle variations in how 
the $R$-atoms can join. We consider each of these possibilities in turn and characterize their complexity.

\subsection{2-Chains}\label{sec: 2chains}

The \emph{chain query} is the simplest possible minimal sj-query with two atoms and we proved earlier that 
its resilience is \np-complete (\cref{chain is npc}). In this section we prove that the chain structure is quite robust and 
that any of its  {expansions} remains \np-complete. 

We call  {``expansions''} of $\chain$ any query obtained by adding new relations to it, i.e. relations that do not self-join. 
We start by presenting the expansions obtained by adding unary relations and then generalize that to any expansion.

\Cref{fig: chain variations} shows how unary relations can be added to $\chain$. Each one can appear by itself
or combined with others. While the proof involves several subcases, the important take-away is that all 8 of these expansions are hard.

\begin{proposition}[Chains with unary relations]\label{unary chain variation in hard}
Any expansion of $\chain$ with unary relations is \np-complete.
\end{proposition}

\begin{proof}[Proof Sketch]
We prove these expansions are hard by a reduction from 3\sat. The same idea used to prove that
$\res(\chain)$ is hard will work here as long as we adapt the variable and clause gadgets to deal with
the existence of the unary relations. 
\Cref{bchain is npc,achain is npc,acchain is npc} in \cref{sec: proofs} contain the details.
\end{proof}

Now we can generalize this hardness result to any chain expansion using a reduction idea similar to the ones used for the proofs
of \cref{unary path,binary path} for paths.

\begin{proposition}[$\twoR$ Chains]\label{2chains are hard}
If a query $q$ contains a 2-chain as its only self-join, then $\res(q)$ is \np-complete.
\end{proposition}

\begin{figure}
\begin{subfigure}[t]{.45\linewidth}
\centering
\begin{tikzpicture}[scale=.35,
			every circle node/.style={fill=white, minimum size=7mm, inner sep=0, draw}]
\node (1) at (0,0) [] {$x$};
\node (2) at (3,0) [] {$y$};
\node (3) at (6,0) [] {$z$};

\path[->, line width=1pt, auto]
	(1) 	edge[color=black] node {$R$}	(2)
	(2) 	edge[color=black] node {$R$}	(3)	
	(1) 	edge[color=black, out=120, in=60, distance=1.5cm, dotted] node {$A$}	(1)
	(2) 	edge[color=black, out=120, in=60, distance=1.5cm, dotted] node {$B$}	(2)
	(3) 	edge[color=black, out=120, in=60, distance=1.5cm, dotted] node {$C$}	(3)
	;
\end{tikzpicture}
\caption{Expansions of $\chain$.}
\label{fig: chain variations}
\end{subfigure}
~
\begin{subfigure}[t]{.45\linewidth}
\centering
\begin{tikzpicture}[scale=.35,
			every circle node/.style={fill=white, minimum size=7mm, inner sep=0, draw}]
\node (1) at (0,0) [] {$x$};
\node (2) at (3,0) [] {$y$};
\node (3) at (6,0) [] {$z$};

\path[->, line width=1pt, auto]
	(1) 	edge[color=black] node {$R$}	(2)
	(3) 	edge[color=black] node[above] {$R$}	(2)	
	(1) 	edge[color=black, out=120, in=60, distance=1.5cm] node {$A$}	(1)
	(2) 	edge[color=black, out=120, in=60, distance=1.5cm, dotted] node {$B$}	(2)
	(3) 	edge[color=black, out=120, in=60, distance=1.5cm] node {$C$}	(3)
	;
\end{tikzpicture}
\caption{Expansions of $\conv$.}
\label{fig: 2conf}
\end{subfigure}
\caption{Expansions of $\chain$ and $\conv$ with unary relations.}
\end{figure}

\subsection{2-Confluences}\label{sec: 2conf}

\emph{Confluences} are defined by relation $R$ joining only in the same attribute. We refer to this 
pattern as $\conv$ (\cref{fig:  2conf}). 

Note that as a stand-alone query $\conv$ is not minimal, so we need 
other atoms connected to both $x$ and $z$. An example of a minimal query containing a confluence is
$\conv^{AC}$ $\datarule$ $A(x),R(x,y),R(z,y), C(z)$.

We next show that the standard \flow algorithm without any modifications works correctly for linear queries with 
no self-join other than one $2$-confluence, thus generalizing the idea of \cref{2conf is easy}.

\begin{proposition}[$\twoR$ $\conv$]\label{prop: 2conf}
$\res(q)$ for any linear query $q$ with $\conv$ as its only self-join pattern can be solved in \p\ by standard
network flow.
\end{proposition}

In \cref{prop: 2conf} we assume that $q$ is linear, thus guaranteeing that every path
in $q$ from $x$ to $z$ involves the variable $y$, and therefore we are able to create a network flow to solve the problem.
Note that this is not true in general for pseudo-linear queries containing $\conv$.
For example, consider $c\! f_p \datarule R(x,y) H^\exSymb(x,z) R(z,y)$. It is easy to see that $c\! f_p$ is pseudo-linear
but we have $\res(c\! f_p) \equiv \res(\vc)$.   Thus, we cover all possible cases for $\conv$ 
by observing,
	
\begin{proposition}[$\twoR$]\label{2conf hard}
Let $q$ be a pseudo-linear query with $\conv$ as its only self-join pattern.  If $q$ contains an
exogenous path from $x$ to $z$ not involving the variable $y$, then $\res(q)$ is $\np$-complete;
otherwise it is in \p.
\end{proposition}

\subsection{2-Permutations}\label{sec: 2perm}

We call two $R$-atoms sharing both variables a \emph{permutation}.
The smallest pattern that has this property is $R(x,y), R(y,x)$ (\cref{Fig_Overview}).
We show that permutations have both $\np$-complete and PTIME instances.

\introparagraph{Easy permutations}
We start with two easy permutations.
\begin{align*}
q_\textup{perm} &\datarule R(x,y),R(y,x)  \quad & \quad q_\textup{perm}^A &\datarule A(x),R(x,y),R(y,x) 
\end{align*}

\begin{proposition}\label{unbounded perm is easy}
$\res(q_\textup{perm})$ and $\res(q_\textup{perm}^A)$ are in \p.
\end{proposition}

\begin{proof}
Given a database $D_1$ satisfying $q_\textup{perm}$, each tuple that is part of a witness for $D_1, q_\textup{perm}$ is
part of exactly one witness.  Therefore the size of a minimum contingency set for $D_1, q_\textup{perm}$ is
exactly the number of witnesses.

Given a database $D_2$ satisfying $q_\textup{perm}^A$, for each join $(a,b)$, we have 2 possible choices. Either $A(a)$ will be in the min 
$\Gamma$ or either one of $R(a,b)$ and $R(b,a)$ but never both. 
Therefore we can reduce $\res(q_\textup{perm}^A)$ to vertex cover in a bipartite
graph, which is in \p.
\end{proof}

\introparagraph{Hard permutations}
Surprisingly, adding another unary atom to $q_\textup{perm}^A$, thus bounding it on both ends.
leads to a hard query.
\begin{align*}
	q_\textup{perm}^{AB} &\datarule A(x), R(x,y), R(y,x), B(y)
\end{align*}

It is still true
that for any pair $R(a,b), R(b,a)$ participating in a join, a minimum contingency set will only contain one tuple
from the pair. This might lead to the wrong conclusion that network flow could solve this problem. 
We will next show that this is incorrect.

\begin{proposition}\label{perm is npc}
$\res(q_\textup{perm}^{AB} )$ is \np-complete.
\end{proposition}

\introparagraph{The criterion}
The main structural difference between the hard and easy permutations defined above is whether or not there 
are relations that ``bound'' the permutation on both ends, i.e.\ whether there are endogenous relations $S, T$, 
such that $S$ contains variable $x$ but not $y$, and $T$ contains variable
$y$ but not $x$. Thus, the hard permutation, $q_\textrm{perm}^{AB}$, is bound, but the easy ones,
$q_\textrm{perm}, q_\textrm{perm}^A$, are not bound.
Using this characterization, we identify when 2-permutations are hard.

\begin{proposition}[$\twoR$]\label{general perm proof}
Let $q$ be a pseudo-linear query with $R(x,y), R(y,x)$ as its only self-join. If $q$ is bound, then $\res(q)$ is \np-complete;
otherwise, $\res(q)$ is in \p.
\end{proposition}

\subsection{Queries with REP}\label{sec:2-patterns-REP}

We call queries with \emph{repeated variables} (or REP in short) those 
where atoms contain the same variable twice, 
e.g.\ occurrences of $R(x,x)$. Note that this is only relevant for the case where $R$ is part of a self-join, otherwise it could be considered as $R(x)$.

There are only three patterns to consider when we are restricted to two $R$-atoms, either one or both atoms have repeated variables. The following queries
are the smallest examples of this class of queries:
\begin{align*}
z_1 &\datarule R(x,x), S(x,y), R(y,y)   \\
z_2 &\datarule R(x,x), S(x,y), R(y,z)\\
z_3 &\datarule R(x,x), R(x,y), A(y)  
\end{align*}
Notice that queries $z_1$ and $z_2$ satisfy the condition for hardness of binary paths (\cref{binary path}), since 
their set of variables is disjoint. Therefore, we can conclude that $\res(z_1)$ and $\res(z_2)$ are \np-complete, as well as any 
expansion of those queries.
We show that 
any REP queries that contain $z_3$ are in \p.

\begin{proposition}[$\twoR$]\label{easy case REP}
Any pseudo-linear query $q$ with exactly two $R$-atoms that contains $z_3$ is in \p.
\end{proposition}

\subsection{The dichotomy}\label{sec: 2dichotomy}

Combining our results so far, with at most two occurrences of the self-join relation, we have
proved a complete characterization of the complexity of resilience:

\begin{theorem}[$\twoR$ Two-Atom Dichotomy]\label{thm: 2atom dichotomy}
Consider $q$ an ssj-CQ, with at most two occurrences of the self-join relation.
If $q$ has any of the following
\begin{enumerate}
\item triad
\item path
\item chain
\item bounded permutation
\item confluence with exogenous path
\end{enumerate}
then $\res(q)$ is $\np$-complete. 
Otherwise, $\res(q)$ is PTIME via a reduction to network flow. In addition there is a PTIME
algorithm that on input $q$ determines which case occurs.
\end{theorem}

\section{Queries with exactly three $R$-atoms}\label{sec:3R}

In \Cref{thm: 2atom dichotomy} we completely characterized the complexity of resilience of all CQs
with at most one repetition of a single relation, thus extending the dichotomy for sj-free CQs into
the land of self-joins.  

In this section, we present an overview of what can happen when we allow a third $R$-atom to
self-join. Since we only have to consider pseudo-linear queries that do not have a path, all three
$R$-atoms must connect to each other directly or through the third $R$-atom. Even though this 
is still a restrictive setting, we will see that it brings non-trivial complications to the characterization. 
We will present some complexity results; but also some remaining open problems.

\subsection{3-Chains}\label{sec: 3chains}

We obtain a 3-chain by adding an extra $R$-atom to a 2-chain in a way such that the new atom joins
in a different attribute  from the other two. 
\begin{align*}
q_{3\textrm{chain}} \datarule R(x,y), R(y,z), R(z,w)
\end{align*}

Analogous to the 2-chain case, 3-chains are always hard.  In fact this holds for 4-chains, 5-chains, etc.

\begin{proposition}[$\ssjbin$]\label{3chains are hard}
For all $k\geq 2$, if $q$ contains a $k$-chain as its only self-join, then $\res(q)$ is \np-complete.
\end{proposition}

\subsection{3-Confluences}\label{sec: 3conf}

Adding a third $R$-atom to a 2-confluence and making sure that it joins in the same attribute with one 
of the two existing $R$-atoms produces a 3-confluence.
\begin{align*}
q_{3\textrm{conf}} &\datarule R(x,y), R(z,y), R(z,w)
\end{align*}

As in the 2-confluence case, $q_{3\textrm{conf}}$ is not minimal, so other atoms are required to make it minimal.
Here are a few examples of minimal
queries containing $q_{3\textrm{conf}}$.
\begin{align*}
q_{3\textrm{conf}}^{AC} &\datarule A(x), R(x,y), R(z,y), R(z,w), C(w)\\
q_{3\textrm{conf}}^{TS} &\datarule T(x,y)^\exSymb, R(x,y), R(z,y), R(z,w), S(z,w)^\exSymb
\end{align*}

\begin{figure*}
\begin{subfigure}[t]{.3\linewidth}
\centering
\begin{tikzpicture}[scale=.35,
			every circle node/.style={fill=white, size=7mm, inner sep=0, draw}]

\node (x) at (-9,0) [] {$x$};
\node (y) at (-6,0) [] {$y$};
\node (z) at (-3,0) [] {$z$};
\node (w) at (0,0) [] {$w$};

\draw[transform canvas={yshift=+0ex},->, line width=1pt](x) to node[above]{$R$} (y);
\draw[transform canvas={yshift=+0ex},->, line width=1pt](z) to node[above]{$R$} (y);
\draw[transform canvas={yshift=+0ex},->, line width=1pt](z) to node[above]{$R$} (w);

\path[->, line width=1pt, auto]
	(x)	edge[color=black, out=120, in=60, distance=1.5cm] node {$A$}	(x)
	(w)	edge[color=black, out=120, in=60, distance=1.5cm] node {$C$}	(w)
	;
\end{tikzpicture}
\caption{$q_{3\textrm{conf}}^{AC}$}
\label{}
\end{subfigure}
~
\begin{subfigure}[t]{.3\linewidth}
\centering
\begin{tikzpicture}[scale=.35,
			every circle node/.style={fill=white, size=7mm, inner sep=0, draw}]

\node (x) at (-9,0) [] {$x$};
\node (y) at (-6,0) [] {$y$};
\node (z) at (-3,0) [] {$z$};
\node (w) at (0,0) [] {$w$};

\draw[transform canvas={yshift=+0.5ex},->, line width=1pt](x) to node[above]{$R$} (y);
\draw[transform canvas={yshift=-0.5ex},->, line width=1pt](x) to node[below]{$T^\exSymb$} (y);
\draw[transform canvas={yshift=+0ex},->, line width=1pt](z) to node[above]{$R$} (y);
\draw[transform canvas={yshift=+0.5ex},->, line width=1pt](z) to node[above]{$R$} (w);
\draw[transform canvas={yshift=-0.5ex},->, line width=1pt](z) to node[below]{$S^\exSymb$} (w);

\end{tikzpicture}
\caption{$q_{3\textrm{conf}}^{TS}$}
\end{subfigure}
~\begin{subfigure}[t]{.3\linewidth}
\centering
\begin{tikzpicture}[scale=.35,
			every circle node/.style={fill=white, size=7mm, inner sep=0, draw}]

\node (x) at (-9,0) [] {$x$};
\node (y) at (-6,0) [] {$y$};
\node (z) at (-3,0) [] {$z$};
\node (w) at (0,0) [] {$w$};

\draw[transform canvas={yshift=+0ex},->, line width=1pt](x) to node[above]{$R$} (y);
\draw[transform canvas={yshift=+0ex},->, line width=1pt](z) to node[above]{$R$} (y);
\draw[transform canvas={yshift=+0.5ex},->, line width=1pt](z) to node[above]{$R$} (w);
\draw[transform canvas={yshift=-0.5ex},->, line width=1pt](z) to node[below]{$S^\exSymb$} (w);

\path[->, line width=1pt, auto]
	(x)	edge[color=black, out=120, in=60, distance=1.5cm] node {$A$}	(x)
	;
\end{tikzpicture}
\caption{$q_{3\textrm{conf}}^{AS}$}
\end{subfigure}
\caption{Three main queries containing a 3-confluence.}
\label{fig: 3conf variants}
\end{figure*}
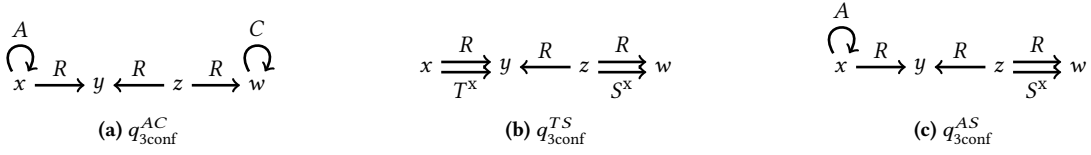

These queries are very similar but one of them is hard, while the other one is easy.

\begin{proposition}\label{cf3.1 is npc}
$\res(q_{3\textrm{conf}}^{AC})$ is \np-complete.
\end{proposition}

\begin{proposition}[$\threeR$]\label{cf3 variations}
Any variation of $q_{3\textrm{conf}}^{AC}$ obtained by including unary relations is \np-complete.
\end{proposition}

\begin{proof}
We define a reduction from Max 2SAT similar to the one used for $q_{3\textrm{conf}}^{AC}$ by adding the appropriate tuples to obtain the same set of joins. The
contingency set doesn't change with the new tuples and therefore the properties of the reduction hold.
\end{proof}

\begin{proposition}\label{cf3.2 is easy}
$\res(q_{3\textrm{conf}}^{TS})$ is in \p.
\end{proposition}

\textbf{Open problem.} There is a third variant of 3-confluences which somewhat mix queries $q_{3\textrm{conf}}^{AC}$ and $q_{3\textrm{conf}}^{TS}$ (\cref{fig: 3conf variants}).
\begin{align*}
q_{3\textrm{conf}}^{AS} &\datarule A(x), R(x,y), R(z,y), R(z,w), S(z,w)^\exSymb
\end{align*}

The complexity of $\res(q_{3\textrm{conf}}^{AS} )$ remains unknown.

\subsection{3-Chain-Confluence}\label{sec: 3cc}

With 3 $R$-atoms, it is possible that different patterns will occur at the same time. This feature of this
case makes it harder to analyze the queries, since the result of these interactions might diverge from
what we expect when we see each pattern in isolation.

In this section we present some queries where a 2-chain and a 2-confluence occur at the same time.
\begin{align*}
q_{\textrm{3cc}}^{AC} &\datarule A(x) R(x,y) R(y,z) R(w,z) C(w) \\
q_{\textrm{3cc}}^{AS} &\datarule A(x) R(x,y) R(y,z) R(w,z) S(w,z)\\
q_{\textrm{3cc}}^{C} &\datarule R(x,y) R(y,z) R(w,z) C(w)
\end{align*}

The resilience of these queries is hard but they require different reductions. If $x$ is bound, then
we can use a reduction from $\res(\chain)$. Otherwise we need a reduction from Max 2SAT.
 
\begin{proposition}\label{bounded 3cc}
$\res(q_{\textrm{3cc}}^{AC})$ and $\res(q_{\textrm{3cc}}^{AS})$ are \np-complete.
\end{proposition}

\begin{proposition}\label{unbounded 3cc}
$\res(q_{\textrm{3cc}}^{C})$ is \np-complete.
\end{proposition}

\textbf{Open Problem.}  In this category of queries with chain and confluence, we don't know the complexity of 
$q_{\textrm{3cc}}^{S} \datarule R(x,y) R(y,z) R(w,z) S(w,z)$.

\subsection{3-Permutation plus R}\label{sec: 3permR}

It is not possible to obtain two permutations in a query with only 3 $R$-atoms. In fact, there are only two ways that a new $R$-atom can 
be connected to a permutation: either by joining with $x$ or $y$, and those are equivalent. 
\begin{align*}
q_{\textrm{3perm-R}} \datarule R(x,y), R(y,z), R(z,y)
\end{align*}

Similar to the $q_{\textrm{3conf}}$ case, $q_{\textrm{3perm-R}}$ is not a minimal query, so additional atoms are necessary. 
We list the main examples of how this query can be made minimal and discuss the complexity of their resilience.

First we start with a query we have already seen and another one that is a slight variation on the first (\cref{fig: easy ks}).
\begin{align*}
q_{\textrm{3perm-R}}^{A} &\datarule A(x) R(x,y) R(y,z) R(z,y)\\
q_{\textrm{3perm-R}}^{S_{wx}} &\datarule S(w,x) R(x,y) R(y,z) R(z,y)
\end{align*}

We proved in \cref{AR perm} that $\res(q_{\textrm{3perm-R}}^{A})$ is in \p\ by using network flow. A similar argument proves that
$\res(q_{\textrm{3perm-R}}^{S_{wx}})$ is also in \p.

\begin{proposition}\label{KS2}
$\res(q_{\textrm{3perm-R}}^{S_{wx}})$ is in \p.
\end{proposition}

The next query we will see is $q_{\textrm{3perm-R}}^{S_{xy}}$. 
Although very similar to $q_{\textrm{3perm-R}}^{A}$ and 
$q_{\textrm{3perm-R}}^{S_{wx}}$, $\res(q_{\textrm{3perm-R}}^{S_{xy}})$ is hard. It is surprising that
such a small difference can already change the complexity of the resilience problem. Moreover, the proof
requires a new reduction instead of a reduction similar to the one used in \cref{perm is npc}.
\begin{align*}
q_{\textrm{3perm-R}}^{S_{xy}} &\datarule S^\exSymb(x,y) R(x,y) R(y,z) R(z,y)
\end{align*}

\begin{proposition}\label{KS0}
$\res(q_{\textrm{3perm-R}}^{S_{xy}})$ is \np-complete.
\end{proposition}

Some other examples of queries that are hard but these are somewhat related to $q_{\textrm{perm}}^{AB}$.
\begin{align*}
q_{\textrm{3perm-R}}^{AC} &\datarule A(x) R(x,y) R(y,z) R(z,y) C(z) \\
q_{\textrm{3perm-R}}^{AB} &\datarule A(x) R(x,y)B(y) R(y,z) R(z,y) \\
q_{\textrm{3perm-R}}^{S_{xy}BC} &\datarule S(x,y) R(x,y) B(y) (R(y,z) R(z,y) C(z)
\end{align*}

\begin{proposition}\label{KS1.5 et al}
$\res(q_{\textrm{3perm-R}}^{AC})$, $\res(q_{\textrm{3perm-R}}^{AB})$ and $\res(q_{\textrm{3perm-R}}^{S_{xy}BC})$ are \np-complete.
\end{proposition}

\textbf{Open Problems.} Despite the similarities with the queries presented in this section, we were not able to 
determine the complexity of the following queries:
\begin{align*}
q_{\textrm{3perm-R}}^{AS_{xy}} &\datarule A(x) S(x,y) R(x,y) R(y,z) R(z,y) \\
q_{\textrm{3perm-R}}^{S_{xy}B} &\datarule S(x,y) R(x,y) B(y) R(y,z) R(z,y) \\
q_{\textrm{3perm-R}}^{S_{xy}C} &\datarule S(x,y) R(x,y) R(y,z) R(z,y) C(z) 
\end{align*}

\subsection{Queries with REP}\label{sec: rep 3 Ratoms}

If all three occurrences of $R$ have repeated variables, then we are in the path case.
\begin{align*}
z_4 &\datarule R(x,x) R(x,y) S(x,y) R(y,y)\\
z_5 &\datarule A(x) R(x,y) R(y,z) R(z,z)
\end{align*}

\begin{proposition}\label{rep 3R}
$\res(z_4)$ and $\res(z_5)$ are \np-complete.
\end{proposition}

\textbf{Open problems.} We don't know the complexity of other queries that fall in this
category of having three $R$-atoms with REP but the following open ones are intriguing.
\begin{align*}
z_6 &\datarule A(x) R(x,y) R(y,y) R(y,z) C(z) \\
z_7 &\datarule A(x) R(x,y) R(y,x) R(y,y)
\end{align*}

Query $z_6$ has a similar structure to $\chain$ but a similar reduction doesn't seem to work.
Similarly, a reduction from $\res(q_{\textrm{perm}}^{AB})$ doesn't work for $z_7$.

\section{Independent Join Paths: a unifying hardness criterion}\label{sec:generalization}

\introparagraph{Motivation}
 {We now define a particular ``template'' for hardness reductions which we call \emph{Independent Join Paths} or IJPs.
The idea is that if we can construct a particular database that fulfills 5 criteria for a query $q$, 
then we can conclude safely that $\res(q)$ is \np-complete.

This recent development is exciting for several reasons:
1)~In our earlier attempts to prove hardness for queries, we amassed a plethora of different hardness proofs, 
with little immediate intuition of how one hardness proof immediately helps facilitate the hardness proof of another query.
Now we expect that the task can be simplified to the task of searching for any particular database that serves as ``proof'' of hardness based on a generalized reduction from Vertex Cover.
2)~We were able to look at our existing hardness proofs and post-hoc identify some part in some gadget that formed an IJP. 
In other words, \emph{IJPs were already present in our hardness proofs} (we give examples in} \cref{app:IJPs}).
Thus IJPs are 
really a \emph{unifying} common denominator for all hard queries known so far.
3)~The search for hardness proofs could now, in theory, be automated. 
While we have not yet explored this idea, we give the intuition in \cref{app:IJPs}.
4)~The hardness based on IJPs is not restricted to the particular fragment of CQs that we have analyzed in this paper; 
rather they are a universal criterion. Even the original criterion of \emph{triads} for sj-free CQs can be subsumed under IJPs.
5)~We conjecture that the inability to form IJPs for those queries that are 
in \PTIME 
can be deduced from the structure of a query, 
and future work will discover the reason.

\begin{figure}[t]
\begin{subfigure}[b]{.27\linewidth}
	\centering
	\includegraphics[scale=0.24]{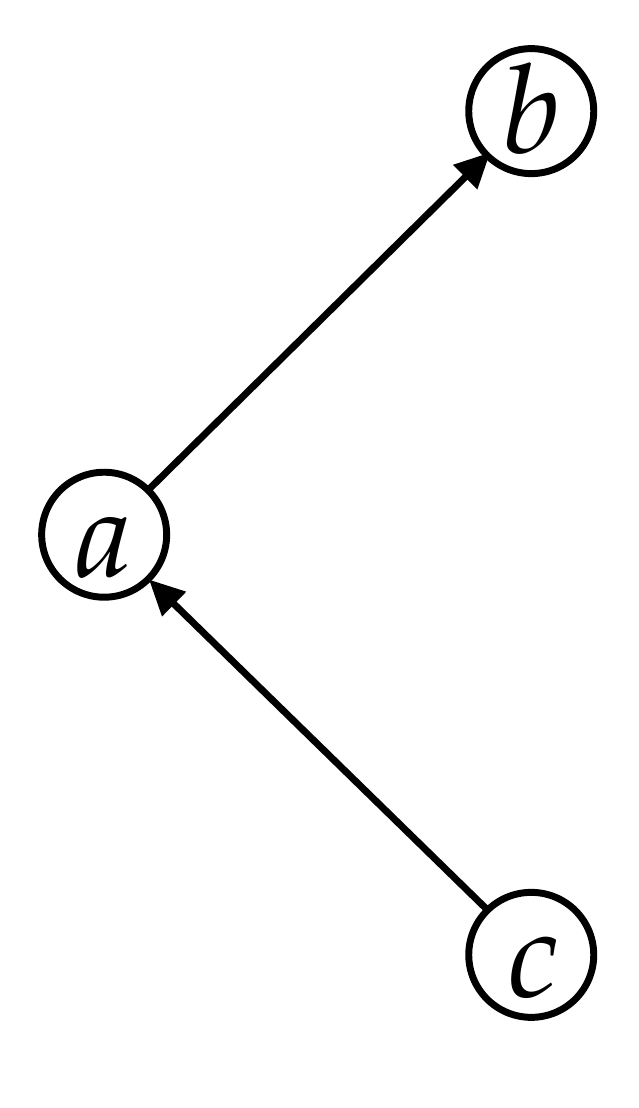}
	\caption{VC($G$)}\label{fig:Fig_VC_intuition_a}
\end{subfigure}
\begin{subfigure}[b]{.27\linewidth}
	\centering
	\includegraphics[scale=0.24]{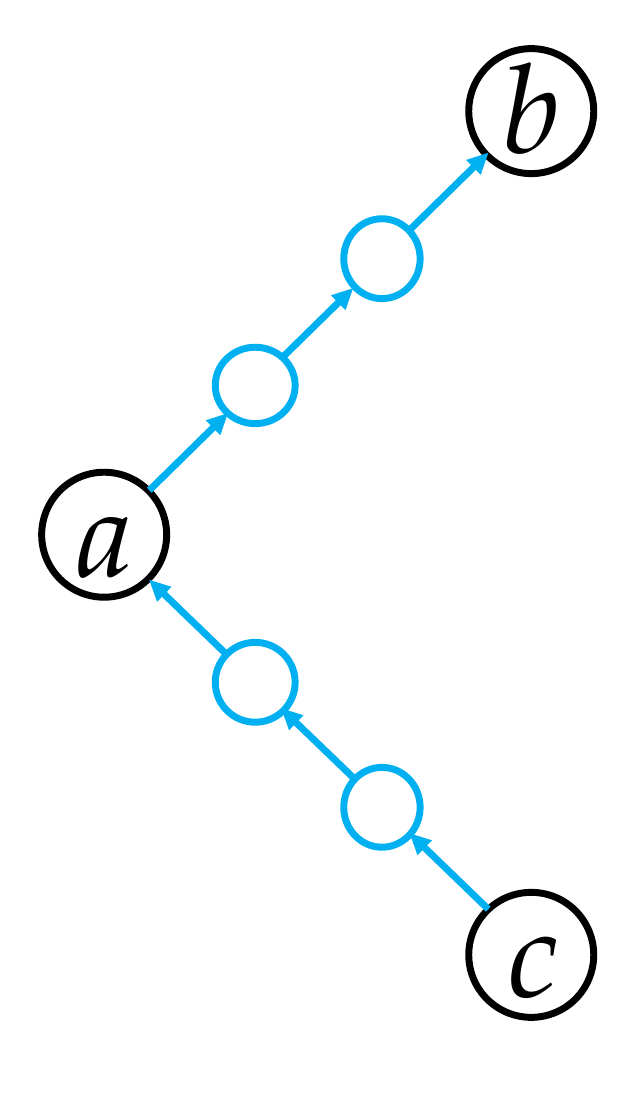}
	\caption{VC($G'$)}\label{fig:Fig_VC_intuition_b}
\end{subfigure}
\begin{subfigure}[b]{.42\linewidth}
	\centering
	\includegraphics[scale=0.24]{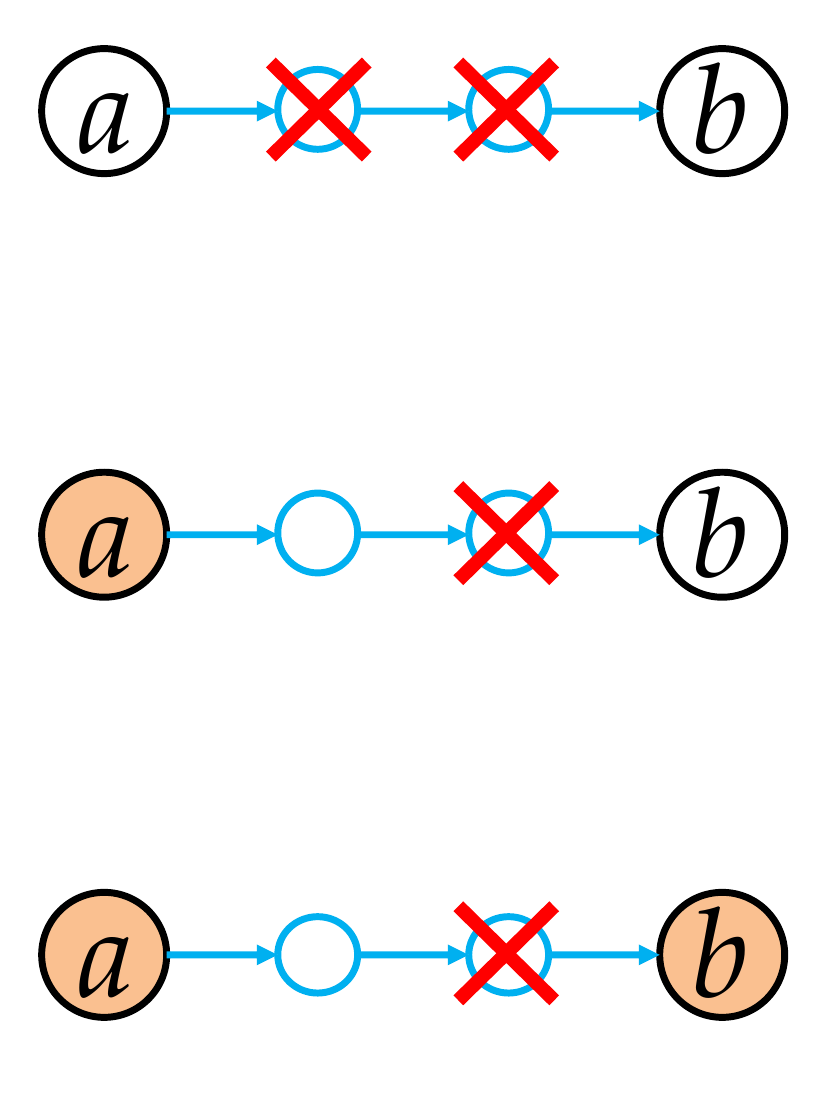}
	\caption{``Or-property'' with $c=1$}\label{fig:Fig_VC_intuition_c}
\end{subfigure}
\caption{Intuition behind IJPs and the ``or-property'': 
if at least one end point of the IJP is removed, 
then the minimal VC for the remaining part of the IJP is reduced by one.}\label{fig:Fig_VC_intuition}
\end{figure}

\introparagraph{The intuition of IJPs}
We have already seen that paths between two subgoals $g_1$ and $g_n$ 
that refer to the same relation are a sufficient condition for hardness under ``certain circumstances''.
Recall our simplest example for a path implying hardness: 
$\vc 	\datarule R(x), S(x,y), R(y)$.
The intuition of our construction is now as follows: 
Take any minimal VC problem for a graph $G(V,E)$
(see \cref{fig:Fig_VC_intuition_a}).
Replace any existing arc with 3 arcs instead to create $G'$
(see \cref{fig:Fig_VC_intuition_b}).
Then $G$ has a VC of size $k$ iff
$G'$ has a VC of size $k+|E|$.
Similarly, replace each arc with 5 instead of 3 arcs, then 
the condition for $G'$ is $k+2|E|$.
The key property we needed for this to work is the fact that 3 arcs form 
a particular path with the following ``OR-property'' (see \cref{fig:Fig_VC_intuition_c}):
As long as at least one end point of the path is removed, 
then the minimal VC is exactly one additional node per path.

\introparagraph{Formalization of IJPs}
We next use this idea to define a particular canonical database instance which we call ``Independent Join Path.''
We conjecture that whenever a query has such a canonical database, then resilience is hard by a proof that generalizes the idea from above.
We give the formal definition here and provide intuition for each of the conditions 
in \cref{app:IJPs}.
In the following, we write $\vec x_{\vec j}$ to denote the subvector of $\vec x$ that retains only the entries indexed by $\vec j$.
For example if $\vec x = (1,2,3,4,5)$ and $\vec j = (2,4,5)$ then $\vec x_{\vec j} = (2,4,5)$

\begin{definition}[Independent Join Path]\label{def:IJP}
A database $D$ forms an \emph{Independent Join Path} for query $q$ if the following conditions hold:
\begin{enumerate}

\item There is a relation $R$ containing at least two tuples 
$R(\vec a)$ and $R(\vec b)$ 
with
$\vec a \not \subseteq \vec b$
and
$\vec b \not \subseteq \vec a$.

\item  In $D$, $R(\vec a)$ and $R(\vec b)$ each participate in exactly one witness
$\vec w_a, \vec w_b$ of $D \models q$.  Both $\vec w_a$ and $\vec w_b$ have exactly $m$ tuples, where $m$
is the number of atoms in $q$.

\item There is no endogenous relation $S$ containing a tuple $S(\vec c)$ 
with $\vec c \subset \vec a$ or $\vec c \subset \vec b$.

\item If there is an exogenous relation $T^\x$ containing a tuple $T^\x(\vec d)$ 
with $\vec d = \vec a_{\vec j}$ for some $\vec j$,
then $T^\x$ also contains $T^\x(\vec e)$ with $\vec e = \vec b_{\vec j}$.

\item Let $c$ be the resilience of $q$ on $D$: $\rho(q,D) = c$. 
Then the resilience is $c-1$ in all 3 cases of removing either 
$R(\vec a)$, or
$R(\vec b)$, or both.
\end{enumerate}
\end{definition}

\begin{conjecture}[IJPs imply hardness]\label{th:ICPs}
	If there is a database $D$ that forms an IJP for a query $q$, 
	then $\res(q)$ is \np-complete.
\end{conjecture}

\introparagraph{The conjecture}
For the fragment of CQs we are considering in this paper, we have been able to simplify some hardness proofs, 
which at times use very different constructions (reductions from VC, 3-SAT, Max 2-SAT),
by looking at our existing hardness proofs and identifying IJPs in our \emph{existing} gadgets. 

We conjecture that the existence of IJPs for a query is also a \emph{necessary} condition for hardness,
that there is an algorithm to verify whether a query can form IJPs or not,
and that the fact that a query cannot form IJPs (such as linear SJ-free CQs) translates immediately into a 
\PTIME algorithm for solving $\res(q)$.

\section{Related work}

In prior work~\cite{FreireGIM15}, we identified the concept of a
triad, a novel structure that allowed us to fully characterize the complexity
of resilience (and consequentially for deletion propagation) for the class of
self-join-free conjunctive queries with potential functional dependencies. Our
work in this paper considers self-joins, which have long-plagued the study of
many problems in database theory; results for such queries have been few and
far between.

\textbf{Deletion propagation and view updates.}
The problem of resilience is a special case of deletion propagation, focusing
on Boolean queries. Deletion propagation generally refers to non-Boolean
queries. Given a non-Boolean query $q$ and database $D$, the typical goal is
to determine the minimum number of tuples that must be removed from $D$, so
that a tuple $\vec t$ is no longer in the query result~\cite{Buneman:2002,
Dayal82} (source side-effects). Variants of deletion propagation consider
side-effects in the query result rather than the source~\cite{KimelfeldVW12,
Kimelfeld12}, and multi-tuple deletions~\cite{Cong12,Kimelfeld:2013}.
Resilience and deletion propagation are special cases of the view update
problem~\cite{Bancilhon81,Cong12,Cosmadakis84,Dayal82,Fagin83,
GottlobPZ1988:ConsistentViews,Keller85}, which consists of finding the set of
operations that should be applied to the database in order to obtain a certain
modification in the view.

\textbf{Causality and explanations.}
Database causality is geared towards providing explanations for query results,
but typically relies on the concept of responsibility~\cite{MeliouGMS11,
MeliouGNS2011}, which is harder than resilience. The idea of interventions
appears in other explanation settings, but often apply to queries instead of
the data~\cite{SudeepaSuciu14,DBLP:journals/pvldb/0002M13,Roy:2015}.
Finally, the problem of explaining \emph{missing} query results \cite{DBLP:conf/sigmod/ChapmanJ09,
DBLP:journals/pvldb/HerschelH10,
DBLP:journals/pvldb/HuangCDN08, 
DBLP:journals/pvldb/HerschelHT09,
DBLP:conf/sigmod/TranC10} is a problem analogous to deletion propagation, but
in this case, we want to add, rather than remove tuples from the view.

\textbf{Provenance and view updates.}
Data provenance studies formalisms that can characterize the 
relation between the input and the output of a given query~\cite{DBLP:conf/icdt/BunemanKT01,
DBLP:journals/ftdb/CheneyCT09, DBLP:journals/tods/CuiWW00,GKT07-semirings}.
``Why-provenance'' is the provenance type most closely related to resilience.
The motivation behind Why-provenance is to find the ``witnesses'' for the query answer, i.e., the tuples or group of tuples in the input 
that can produce the answer. Resilience, searches to find a \emph{minimum} set of input tuples that can make a query false.

\section{Final Remarks}

In this paper, we studied the problem of resilience for conjunctive
queries with self-joins. We identified fundamental query structures that
impact hardness, and proved a complete dichotomy for the restricted class of
single-self-join binary CQs where exactly two atoms can correspond to the same relation.

We also present results towards the for the case of binary CQs with a single self-join relation
that appears in $3$ atoms, and identifies some open problems and
challenges towards completing the dichotomy for this class (\cref{sec:3R}).

Our work also presents a roadmap for tackling the analysis of more extended
query families. \Cref{sec:generalization} provides 
towards 
a possible generalization of our results to all class of self-join queries, by using a unifying 
criterion that we call \emph{Independent Join Paths}.

Overall, our work in this paper contributes important progress in the
theoretical analysis of self-joins, which has long been stalled for many
related problems. We hope that our results, even though they apply to a
restricted class, will provide the foundations to help solve the general case
for CQs with self-joins in the future.

\bibliographystyle{ACM-Reference-Format}
\bibliography{causation}

\appendix
\clearpage
\section{Detailed proofs}\label{sec: proofs}

\subsection{Proofs for \cref{sec:newHardQueries}}

\begin{proof}[Proof of \Cref{prop:hardness vc}]
A database with unary $R$ and binary $S$ is simply a directed graph. 
For a directed graph $G=(V,E)$, we can create a database instance $D_G$
where for each node $v_i \in V$, we add tuple $R(v_i)$ in $D_G$, and for each edge
$(v_i, v_j) \in E$, we add tuple $S(v_i, v_j)$ in $D_G$. 
Furthermore, $D_G \models \vc$ iff graph $G$ has at least one edge.
Note that any vertex cover $C$ of $G$ has a correspondent set of tuples $\Gamma_C$ in $D_G$, 
and it is easy to see that $D_G - \Gamma \not\models \vc$. 

More precisely, $$(G,k) \in VC \Leftrightarrow (D_G,k)\in\res(\vc).$$

Therefore, $\res(\vc)$ is \np-complete.

\end{proof}

\begin{proof}[Proof of \Cref{chain is npc}]
We reduce 3SAT to $\res(\chain)$. 
Let $\psi$ be a 3CNF formula with $n$ variables $x, y, z, \ldots, v_n$ and $m$ clauses 
$C_1, \ldots, C_{m}$.  We map any such $\psi$ to a pair $(D_\psi,k_\psi)$ where $D_\psi$ is a database 
satisfying~$\chain$,  $k_\psi=(2n+5)m$ and 
\[\psi\in 3\sat \quad\Leftrightarrow\quad (D_\psi,k_\psi) \in \res(\chain)\; .\]

\Cref{Fig_red_RxyRyz} 
shows part of $D_\psi$ consisting of the gadgets for $x,y,z, C_1$
where in this example, $C_1 = ( x \lor \ov{y} \lor z)$.  
 {The nodes correspond to tuples in $D_\psi$ and there is a directed edge between any
two nodes those are witnesses for $\chain, D_\psi$.}
The variable gadgets are 
 {cycles of length $2m$}
whose minimum contingency sets are the set of $m$ blue nodes indicating the variable is assigned
true, or the set of $m$ red nodes, indicating the variable is assigned false.  The 9-node clause
gadgets have minimum contingency sets of size 5 when the clause is assigned true, and 6 otherwise.
\end{proof}

\subsection{Proofs for \cref{sec: newEasyQueries}}

\begin{proof}[Proof of \Cref{2conf is easy}]
We first argue that $R$-tuples are not the optimal choice for a contingency set. 
Let $\Gamma$ be a minimum contingency set containing tuple $R(1,2)$. 

\emph{Case 1: } $D$ contains only $A(1)$ or $C(1)$ but not both. 
WLOG, suppose it contains only $A(1)$. We can then obtain a contingency set
$\Gamma' = (\Gamma - R(1,2)) \cup A(1)$ of size $k$.
Similar if it contains only $C(1)$.

\emph{Case 2:} $D$ contains both $A(1)$ and $C(1)$. 
Consider $\Gamma' = (\Gamma \cup A(1)) - R(1,2)$ and $\Gamma'' = (\Gamma \cup C(1)) - R(1,2)$, and suppose
that neither of those is a contingency set. Then we have $A(i),R(i,2),R(1,2),C(1)$ in $D - \Gamma'$ and 
$A(1),R(1,2),R(j,2),C(j)$ in $D - \Gamma''$. 
However, the existence of those witnesses implies that
$D - \Gamma$ has the witness $A(i),R(i,2),R(j,2),C(j)$ contradicting the fact that $\Gamma$ is a contingency set. 
Therefore, at least one of $\Gamma', \Gamma''$ must be a contingency set and we can replace $R(1,2)$ by $A(1)$
or $C(1)$.

Since $R$ can be made exogenous, solving resilience for this query is the same as solving vertex cover
in a bipartite graph, and therefore is in \p.
\end{proof}

\begin{proof}[Proof of \Cref{AR perm}]
For a linear sj-free query, we can represent its resilience problem as a network flow 
making each endogenous tuple an edge of weight 1.  Each flow is a witness and the min-cuts are exactly the
minimum contingency sets (see \cite{MeliouGMS11} for details).
It is not clear what to do with repeated relations because 
there is no obvious way to add to a standard network
flow algorithm an extra constraint that two or more edges represent the same tuple,
and can thus be removed \emph{together} at the reduced cost of only 1.

To handle $q_{\textrm{3perm-R}}^{A}$, consider an input database  {$D$ with}  $A$ and $R$ tuples. 
We refer to $R$-tuples that have an inverse as 2-way tuples, and the ones that don't as 1-way
tuples.  We construct a flow graph by creating 1-weight edges $(a_l, a_r)$ for all tuples $A(a)$, 
and 1-weight edges $(\angle{ab}_l, \angle{ab}_r)$ for pairs $\set{a,b}$of 2-way tuples. 
There are $\infty$-weight edges $(s, a_l)$ for all tuples $A(a)$, where $s$ is the source,
$\infty$-weight edges $(x_r, \angle{uv}_l)$ if and only if $x \in \set{u,v}$ or there is a 1-way tuple $R(x,u)$ or $R(x,v)$,
and  $\infty$-weight edges $(\angle{ab}_r, t)$ for pairs $\set{a,b}$of 2-way tuples, where $t$ is the target.
 {Note that 1-way tuples are never the optimal choice}, since we can always pick an $A$-tuple instead, so they have
infinite weight in the flow graph. Below we refer to the tuple that the edges represent, instead of the edge itself.

We show that from the min-cut, $M$, of the flow graph, we can construct a minimum contingency set,
$\Gamma$, as follows:  $\Gamma$ contains all the $A(a)$'s from $M$. For each 
edge $\set{a,b}\in M$, we add one of $R(a,b)$ or $R(b,a)$ to $\Gamma$ as follows:
If $A(a)\in (D - M)$ but $A(b)\not\in (D - M)$  then we add  $R(a,b)$ to $\Gamma$. Symmetrically,
if $A(b)\in (D - M)$ but $A(a)\not\in (D - M)$  then we add  $R(b,a)$
to $\Gamma$; {otherwise, arbitrarily add one or the other}.

We claim that the resulting $\Gamma$ is a minimum contingency set.
Because it comes from a min-cut, it suffices to show that $\Gamma$ is a
contingency set, i.e., $D-\Gamma \not\models q_{\textrm{3perm-R}}^{A}$.
Suppose for the sake of a contradiction, that $D-\Gamma$ has a wtiness $A(a)$, $R(a,b)$, $R(b,a)$,
$R(a,b)$, i.e., some tuple, $R(a,b)$, occurs twice in the join. This is impossible because since $A(a) \not\in M$, at least one of $R(a,b)$ or
$R(b,a)$ must be in $\Gamma$.

The other possible wtiness is $A(c), R(c,a), R(a,b), R(b,a)$.
Note that if $R(c,a)$ is a 1-way tuple, then this wtiness would be a flow contradicting the fact that
$M$ is a cut.  Thus, $R(c,a)$ is a 2-way tuple.  Since $A(c)\set{c,a}$ can't be a flow, the pair
$\set{c,a}$ must be in $M$.

Since $R(c,a)$ was not chosen in $\Gamma$, it must
be that $A(a)\in (D - M)$.  This means that there is still a flow from $A(a)$ to $\set{a,b}$, so $M$
was not a cut.
\end{proof}

\subsection{Proofs for \cref{sec:components}}

\begin{proof}[Proof of \Cref{min gamma disconnected}]
First observe that disconnected components join as a cross-product,
so for a query to be made false it is enough that at least one of its query
components is made false. Hence, for each query component 
$q_i$, if $D - \Gamma_i \not\models q_i$, then 
$D - \Gamma_i \not\models q$, which then implies $\rho(q, D) = \min_{i} \rho(q_i, D)$.
\end{proof}

\begin{proof}[Proof of \Cref{complexity disconnected queries}]
This is easy to see because the resilience problem for  $q$ consists of the union of the $k$ independent resilience problems for its components.  If $\res(q_i)$ is \np-complete, then we can take a database that has the relevant instance of $q_i$ and all the other components can be extremely resilient, so the minimum contingency sets is always a subset of $q_i$'s component.  Conversely, if each $\res(q_i)$ is in \p, then to solve the minimum contingency set, we find the minimum contingency of each component, and the global minimum is simply the minimum of these minima.
\end{proof}

\subsection{Proofs for \cref{sec:domination}}

\begin{proof}[Proof of \Cref{lem: domination sj}]
We show that tuples from dominated relations don't need to be used in minimum contingency sets.
Assume $q$ is a connected query and let $\Gamma$ be a minimum contingency set of $q$ in $D$.

Suppose that relation $A$ dominates relation $B$ and 
there is some tuple $B(\mathbf{t})$ that is in $\Gamma$. 
Tuple $B(\mathbf{t})$ can participate in joins as one or more of the $B$-atoms in $q$. 
Let's call those atoms $B_i$, for $i \in [k]$. Our definition of domination guarantees that 
there exists an atom $A_j$ for each atom $B_i$ such that the projection of $\mathbf{t}$ 
onto $\var(A_j)$ always produces the same tuple $\mathbf{p}$.
Then we can replace $B(\mathbf{t})$ by $A(\mathbf{p})$ and we remove at least 
as many witnesses if $D \models q$.

As a result we show the complexity of $\res(q)$ is the same if $B$ is made exogenous and therefore $\res(q) \equiv \res(q')$.
\end{proof}

\subsection{Proofs for \cref{sec:nonLinearQueries}}

\begin{proof}[Proof of \Cref{lem: sj makes harder}]
Let $D\models q$ be a database.  We map $D$ to $D'$ by marking all the tuples according to which
variables they refer to in witnesses of~$q$.  For each witness $j$ assigning the variables of $q$ to domain values ($\dom(D)$), 
we add the tuples $T(j(v_1)_{v_1}, \ldots j(v_k)_{v_k})$ to $D'$, where $T(\ov{v})$ occurs in $q^{\sj}$. 
In particular, if $S_i$ was replaced by $R_i$ to obtain $q^{\sj}$, $S_i(j(v_1), \ldots j(v_k))\in D$ results in adding the tuple
$R_i(j(v_1)_{v_1}, \ldots j(v_k)_{v_k})$ to $D'$. 

For example, consider that atom $S(x,y,z)$ was replaced by atom $R(x,y.z)$.
If $S(a,b,c)$ is part of a witness $j$, we have $j(x)=a, j(y)=b, j(z)=c$. Then $R(j(x)_x, j(y)_y, j(z)_z) = R(a_x,b_y, c_z)$ is included in $D'$.

Since the variables mark the tuples in $D'$, the new self-joins have no effect:  if the subscripted
variables are $\ov{v}$ in a tuple of $R_i$ in $D'$, then it came from a tuple of $S_i$ in $D$.  It then
follows that there is a 1:1 correspondence of contingency sets for $(D,q)$ and $(D',q^{\sj})$.  We need
the minimality of $q^{\sj}$, because if there were an assignment where $D-\Gamma' \models q^{\sj}$ when
$D-\Gamma\not\models q$, this would correspond to a reassignment of the variables, $\var(q^{\sj})$ to a
proper subset, so that some $R_i$ would be doing ``double duty''.  This would mean that a proper subset of
$q^{\sj}$ implies $q^{\sj}$, i.e, $q^{\sj}$ is not minimal.

\end{proof}

\subsection{Proofs for \cref{sec: sj rats and brats}}

\begin{proof}[Proof of \Cref{prop: sj rats and brats are hard}]
The proofs essentially follow the same strategy used to reduce 3SAT to $\res(q_\triangle)$ with a few adjustments
to handle the self-joining relation and also the variable order, which is relevant in some
cases. See \cref{prop: sj rats hard} 
and \cref{prop: sj brats hard} for the details.
\end{proof}

\begin{lemma}\label{prop: sj rats hard}
$\res(\rats^{\textrm{sj}_1})$ and $\res(\rats^{\textrm{sj}_2})$ are \np-complete. 
\end{lemma}

\begin{proof}[Proof of \Cref{prop: sj rats hard}]
We first show that $\res(\rats^{\textrm{sj}_1})$ is \np-complete by a reduction from 3SAT, 
similar to the one used to prove $\res(q_\triangle)$ is \np-complete (\cref{thm: hardness of triangle}).

Let $\psi$ be a 3CNF formula with $n$ variables $v_1, \ldots, v_n$ and $m$ clauses 
$C_0, \ldots, C_{m-1}$. Our reduction will map  any such $\psi$ to a pair $(D^1_\psi,k_\psi)$ where $D^1_\psi$ is a database 
satisfying $\rats^{\textrm{sj}_1}$, and $$\psi\in 3\sat \quad\Leftrightarrow\quad (D^1_\psi,k_\psi) \in \res(\rats^{\textrm{sj}_1})$$

In our construction, if $\psi \in 3\sat$, then the size of each minimum contingency set for $\rats^{\textrm{sj}_1}$ in
$D^1_\psi$ will be $k_\psi=6mn$, whereas if $\psi \not\in 3\sat$, then the size of all contingency sets
for $\rats^{\textrm{sj}_1}$ in $D^1_\psi$ will be greater than $k_\psi$.

We construct $D^1_\psi$ by taking $D_\psi$ from the proof of \cref{thm: hardness of triangle}, and
adding the following tuples for each witness $\angle{a,b,c}$ in $D_\psi, q_\triangle$:
\begin{align*}
R &= \set{(a,b),(b,c),(c,a)} \\
A &= \set{(a),(b),(c)}
\end{align*}

Notice that for each witness $\angle{a,b,c}$ in $D_\psi$ we thus create 3 witnesses, $\angle{a,b,c}$,
$\angle{b,c,a}$, $\angle{c,a,b}$ in $D_\psi^1$ but they all use the same $R$-tuples.

We know from \cref{thm: hardness of triangle} that some $R$-tuples participate in 2 witnesses
(triangles) and some only in 1 within a variable gadget. Thus, in $D_\psi^2$ these numbers are 6
witnesses or 3 witnesses.  
Observe that $A$-tuples participate in at most 2 witnesses each, so it is never better to choose an
$A$-tuple instead of an $R$ tuple. Therefore it follows that the same choice of tuples
for the minimum contingency set for $D_\psi, q_\triangle$ will also work for $D^1_\psi, \rats^{\textrm{sj}_1}$ by choosing the corresponding $R$-tuples in $D^1_\psi$
based on the $R,S,T$-tuples chosen from $D_\psi$.

For $\rats^{\textrm{sj}_2}$ the reduction is similar, but the final atom -- $R(x,z)$
instead of $R(z,x)$ --  must be handled.  The solution is that for each witness 
$\angle{a,b,c}$ in $D_\psi, q_\triangle$, we add the following tuples to $D^2_\psi$:
\begin{align*}
R &= \set{(a,b),(b,a),(b,c),(c,b),(c,a),(a,c)} \\
A &= \set{(a),(b),(c)}
\end{align*}

Now, each witness from $D_\psi, q_\triangle$ leads to 6 witnesses in $D^2_\psi, \rats^{\textrm{sj}_2}$ --
the three from the above proof plus their reversals. Thus, the $R$-tuples for solid edges from
\Cref{fig:gadgetgi} are used in 6 witnesses each, whereas $A$-tuples are in at most 4 witnesses each.  
Thus, based on the minimum contingency sets for $D_\psi, q_\triangle$, we create minimum
contingency sets for $D^2_\psi, \rats^{\textrm{sj}_2}$ by including the corresponding $R$-tuples
and their reversals.
\end{proof}

\begin{lemma}\label{prop: sj brats hard}
$\res(\brats^{\textrm{sj}_1})$, $\res(\brats^{\textrm{sj}_2})$ and $\res(\brats^{\textrm{sj}_3})$ are \np-complete.
\end{lemma}

\begin{proof}[Proof of \Cref{prop: sj brats hard}]
The same idea used above to prove that $\res(\rats^{\textrm{sj}_1})$ is hard, works for query $\brats^{\textrm{sj}_1}$. 
When defining $D^1_\psi$ for this case, we just need to add the appropriate $B$-tuples:
\begin{align*}
R &= \set{(a,b),(b,c),(c,a)} \\
A &= \set{(a),(b),(c)}\\
B &= \set{(a),(b),(c)}
\end{align*}

Since $B$-tuples have the same properties as the $A$-tuples, they are never better choices than
$R$-tuples and we can obtain a minimum contingency set 
with only $R$-tuples, as we saw in \cref{prop: sj rats hard} above. Similar reduction thus follow for $\res(\brats^{\textrm{sj}_2})$ and $\res(\brats^{\textrm{sj}_3})$.
\end{proof}

\subsection{Proofs for \cref{sec:triadsRemainHard}}

\begin{proof}[Proof of \Cref{thm: triads in sj}]
This mostly follows from the fact that triads make sj-free queries hard and adding self-joins to a
hard query keeps it hard (\cref{hard part dichotomy}, \cref{lem: sj makes harder}).

The case we haven't covered yet is where the triad in $q$ involves self-join relations which would
be dominated and thus exogenous in the corresponding sj-free query.  Examples are self-join variations
of $\rats$ and $\brats$ which are hard even though -- because of domination -- their sj-free cases
are easy (\cref{prop: sj rats and brats are hard}).

We now follow and extend the proof of \cref{hard part dichotomy} when $q$ has a triad, ${\mathcal T}
= (S_0, S_1, S_2)$, even though if ${\mathcal T}$ did not include a self join, one or more of its
members would be dominated.  In Case 1, $\var(S_i)$, $i= 1, 2, 3$, are pairwise disjoint.  Here the
reduction from $\res(q_\triangle)$ to $\res(q)$ goes through exactly as in the proof of \cref{hard part dichotomy}.
We can choose a single relevant variable for each $S_i$, so no domination is
possible.  Any minimum contingency set consists of elements of $S_0 (\angle{ab})$, $S_1
(\angle{bc})$ or $S_2 (\angle{ca})$, and the reduction from $\res(q_{\triangle})$ goes through.

In Case 2, where $\var(S_i)$ are not pairwise disjoint, we have to consider a partition of the
variables into 7 pieces (Eqn.\ \ref{variable partition eq} from the proof of \cref{hard part
  dichotomy}). As argued there, there is still a 1:1 
correspondence between witnesses of $(D,q_{\triangle})$ and witnesses of $(D',q)$.

If there are no (endogenous) relations containing just the $a$, $b$ or $c$ variables, then the reduction from
$\res(q_{\triangle})$ goes through.  If there is a relation containing just $a$, then we instead use
the same reduction but 
from the appropriate self-join variation of $\rats$.  If there are relations containing just $a$ and
$b$ but not $c$, then we get a reduction from 
the appropriate self-join variation of $\brats$. If there are relations for $a$, $b$ and $c$, 
then these form an sj-free triad and thus we already know that $\res(q)$ is hard.
\end{proof}

\subsection{Proofs for \cref{sec:pseudoLinear}}

\begin{proof}[Proof of \Cref{thm: no triad means linear}]
We are given $q$, a CQ with no triad.  Let $n$ be the number of groups of endogenous atoms in $q$,
where we put two atoms in the same group iff they contain exactly the same variables, so $A(x,y)$
and $R(y,x)$ belong in the same group, but $B(x)$ and $R(x,z)$ do not. We refer to the groups of 
endogenous atoms as $G_1, G_2, \ldots, G_n$.

Since $q$ is connected but has no triad, for any pair $G_i, G_j$, either these atoms are connected directly in $\mathcal{H}(q)$,
or they are connected via at least another group $G_k$, but both cases cannot occur. If they are connected directly, then they must appear
consecutive in an order of the endogenous atoms. Otherwise, $G_k$ must be placed between them. Note that in the latter case,  
removing the variables of $G_k$ separates the atoms of $q$ into two connected components, one containing $G_i$ and the other
containing $G_j$, so we call $G_k$ the separator of $G_i, G_j$.

Now, for any set $A,B,C$ of endogenous atoms from different groups, when $A$ and $B$ are already placed
along the line, say with $B$ to the right of $A$, then it is easy to see where $C$ must go.  
If $A$ is the separator, $C$ goes to the left of $A$, if $B$ is the separator, $C$ goes
to the right of $B$ and if $C$ is the separator, then it goes between $A$ and $B$, and that's
what guarantees the endogenous atoms are linearly connected. Looking at \Cref{easyWalkFig}, 
we see that the endogenous atoms of $q$ are arranged linearly.

\end{proof}

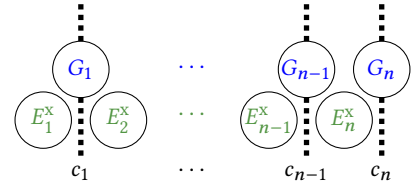
\begin{figure}
\begin{center}
\begin{tikzpicture}[ scale=.5]
\draw (3,.65)  circle [radius=.75] node {{\color{dg} $E^\exSymb_1$}};  
\draw (4,-.75) node {{\color{black} $c_1$}};  
\draw (5,.65)  circle [radius=.75]  node {{\color{dg} $E^\exSymb_2$}};  
\draw (7,.75) node {{\color{dg} $\cdots$}};  
\draw (7,-.75) node {{\color{black} $\cdots$}};  
\draw (7,2) node {{\color{blue} $\cdots$}};  
\draw (9,.65)  circle [radius=.75]  node {{\color{dg} $E^\exSymb_{n-1}$}};  
\draw (10,-.75) node {{\color{black} $c_{n-1}$}};  
\draw (11,.65)   circle [radius=.75] node {{\color{dg} $E^\exSymb_{n}$}};  
\draw (12,-.75) node {{\color{black} $c_{n}$}};  
\draw (4,2) circle [radius=.75] node {{\color{blue} $G_1$}};
\draw (10,2) circle [radius=.75] node {{\color{blue} $G_{n-1}$}};
\draw (12,2) circle [radius=.75] node {{\color{blue} $G_{n}$}};
 \draw[line width=2pt,dotted] (4,2.75) -- (4,3.75); 
 \draw[line width=2pt,dotted] (4,1.25) -- (4,-.3); 
 \draw[line width=2pt,dotted] (10,2.75) -- (10,3.75); 
 \draw[line width=2pt,dotted] (10,1.25) -- (10,-.3); 
 \draw[line width=2pt,dotted] (12,2.75) -- (12,3.75); 
 \draw[line width=2pt,dotted] (12,1.25) -- (12,-.3); 

\end{tikzpicture}
\end{center}
\caption{A walk along the endogenous atoms.  The cut $c_i$ results from removing all the variables
  (edges) from group $G_i$.}
\label{easyWalkFig}
\end{figure}

 \subsection{Proofs for \cref{sec:paths}}

\begin{proof}[Proof of \Cref{unary path} (Unary Path)]
We define a reduction from $\res(\vc)$.
Given a database $D$ we want to define a database $D'$ such that
\begin{equation}\label{unaryPathReductionEq}
(D,k) \in \res(\vc) \qLra (D',k) \in \res(q) 
\end{equation}

We can assume that $A(x)$ and $A(y)$ are consecutive occurrences of $A$ so let $p$ be a subquery of
$q$ consisting of a path from $A(x)$ to $A(y)$ with no intervening occurrences of $A$.
Thus, $q = q_\ell A(x)pA(y)q_r$.  
Since $A$ is the only sj relation, the relations that occur in $p$ occur only in $p$.

For each atom $R_i(v_1, v_2)$ occurring in $p$, we define 
\[
R_i' = \bigset{ (t(v_1,a,b), t(v_2,a,b))}{D \models \vc(a,b)}
\]
where
\[
t(v,a,b) \;\eqdef\;
\begin{cases}
a   & \textrm{if } v = x \\
b   & \textrm{if } v = y \\
\angle{ab}_v   & \textrm{otherwise}
\end{cases}
\]

In other words, $x$ maps to $a$, $y$ maps to $b$, and any other variable $v$ maps to $\angle{ab}_v$.
Thus, we have made a faithful copy of $D$ capturing~$\vc$.  For the other atoms, $S_j(v_1,v_2)$,
not in $p$,  let 
\[
S_j' = \bigset{ (m(v_1,a,b), m(v_2,a,b))}{D \models \vc(a,b)}
\]

where $m(v,a,b)$ matches with $t(v,a,b)$ as well as with a set of $n$ new values, where $n =
\abs{\dom(D)}$.  It follows that there is always a minimum contingency sets for $(D',q)$ with only $A$-tuples,
in particular, the sets $\bigset{A(a)}{V(a)\in \Gamma}$ for $\Gamma$ any minimum contingency set for $(D,\vc)$.
\end{proof}

\begin{proof}[Proof of \Cref{binary path} (Binary Path)]
Similar to the unary case, we define a reduction from $\res(\vc)$.
Given a database $D$ we want to define a database $D'$ such that
\begin{equation}
(D,k) \in \res(\vc) \qLra (D',k) \in \res(q) 
\end{equation}

Consider $q = q_\ell R(x,y)pR(z,w)q_r$, and that $p$ is a subquery of
$q$ consisting of a path from $R(x,y)$ to $R(z,w)$ with no intervening occurrences of $R$.
By assumption, there is no path of just $R$'s from $R(x,y)$ to $R(z,w)$, so we may assume that
$R(x,y)$ and $R(z,w)$ have such an $R$-free path, $p$, between them.

In order to define the reduction, we define an equivalence relation, $\equiv$, on the variables
occurring in $q$, namely $u\equiv v$ iff $q$ has an $R$-path from $u$ to $v$, i.e., there is a path
of $R$-atoms occuring in $q$ that takes us from $u$ to $v$.  (For example, for the query
$R(x,y),S(u,z), R(z,w), Q(w,x), R(x,v)$, the equivalence classes of~$\equiv$ are $\set{x,y,v},
\set{z,w}, \set{u}$.) Note that by assumption, for the equivalence relation defined by $q$, $x \not\equiv z$.

For any atom $S_i(v_1, v_2)$ occurring in $R(x,y)pR(z,w)$, we define 
\[
S_i' = \bigset{ (t'(v_1,a,b), t'(v_2,a,b))}{D \models \vc(a,b)}
\]
where
\[
t'(v,a,b) \;\eqdef\;
\begin{cases}
a   & \textrm{if } v \equiv  x \\
b   & \textrm{if } v \equiv z \\
\angle{ab}_{v}   & \textrm{otherwise}
\end{cases}
\]

Additionally, for atoms $T_j(v_1,v_2)$ occurring in $q_l, q_r$,  let 
\[
T_j' = \bigset{ (m(v_1,a,b), m(v_2,a,b))}{D \models \vc(a,b)}
\]
where $m(v,a,b)$ matches with $t'(v,a,b)$ as well as with a set of $n$ new values, where $n =
\abs{\dom(D)}$.

We have that all $R$-tuples in $D'$ will have the same value as first and second attributes, so $R$ can be seen as
corresponding to relation $A$ in $D$. Similar to the unary case, we have made a copy of $D$ capturing $\vc$ and 
there is always a minimum contingency sets for $(D',q)$ with only $R$-tuples,
in particular, the sets $\bigset{R(a,a)}{V(a)\in \Gamma}$ for $\Gamma$ any minimum contingency set for $(D,\vc)$.
\end{proof}

\subsection{Proofs for \cref{sec: 2chains}}

These are the expansions of $\chain$ with unary relations: 
\begin{align*}
\achain &\datarule A(x),R(x,y),R(y,z) &\quad\Cref{achain is npc}\\
\bchain &\datarule R(x,y),B(y),R(y,z) &\quad\Cref{bchain is npc}\\
\cchain &\datarule R(x,y),R(y,z),C(z) &\quad\Cref{achain is npc}\\
\abchain &\datarule A(x),R(x,y),B(y),R(y,z)&\quad\Cref{achain is npc}\\
\bcchain &\datarule R(x,y),B(y),R(y,z),C(z)&\quad\Cref{achain is npc}\\
\acchain &\datarule A(x),R(x,y),R(y,z),C(z)&\quad\Cref{acchain is npc}\\
\abcchain &\datarule A(x),R(x,y),B(y),R(y,z),C(z)&\quad\Cref{acchain is npc} 
\end{align*}

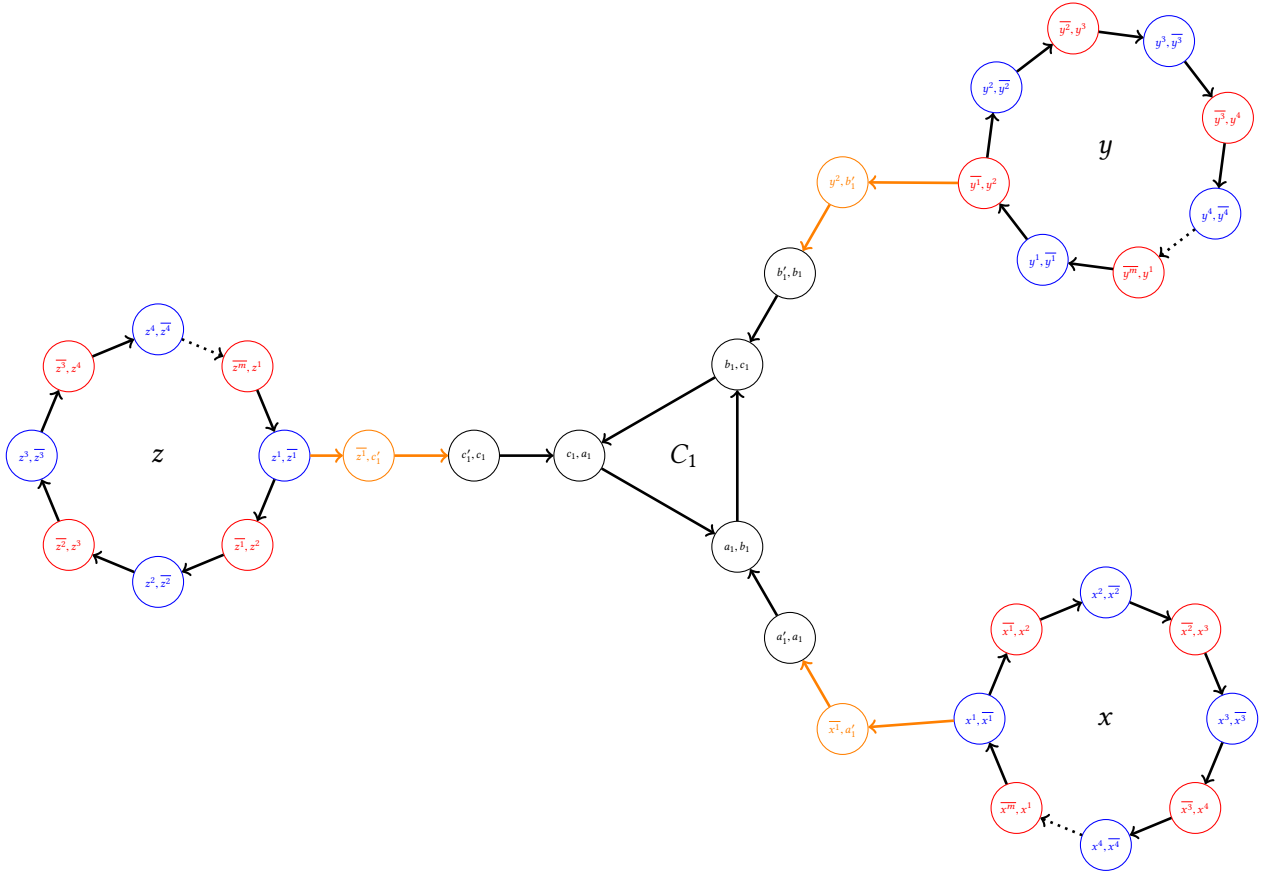
\begin{figure*}[h]
\centering
\begin{tikzpicture}[scale=1.4,
			    every circle node/.style={fill=white, minimum size=3cc, inner sep=0, scale=.5, draw}]


\def \radius {1cm}

\begin{scope}[rotate=-90]

\begin{scope}[rotate=30]
\node (a1) at ({360/3 * (0)}:\radius) [circle] {$a_{1},b_{1}$};
\node (b1) at ({360/3 * (1)}:\radius) [circle] {$b_{1},c_{1}$};
\node (c1) at ({360/3 * (2)}:\radius) [circle] {$c_{1},a_{1}$};
\node (a2) at ({360/3 * (0)}:{\radius+1cm}) [circle] {$a_{1}',a_{1}$};
\node (b2) at ({360/3 * (1)}:{\radius+1cm}) [circle] {$b_{1}',b_{1}$};
\node (c2) at ({360/3 * (2)}:{\radius+1cm}) [circle] {$c_{1}',c_{1}$};
\node (a3) at ({360/3 * (0)}:{\radius+2cm}) [circle, color=orange] {$\ov{x^{1}},a_{1}'$};
\node (b3) at ({360/3 * (1)}:{\radius+2cm}) [circle, color=orange] {$y^{2}, b_{1}'$};
\node (c3) at ({360/3 * (2)}:{\radius+2cm}) [circle, color=orange] {$\ov{z^{1}},c_{1}'$};
\end{scope}

\def \radius {1.2cm}
\def \n {8}

\begin{scope}[xshift=2.5cm, yshift=4cm, rotate=-90]
\node (11) at ({360/\n * (0)}:\radius)  [circle, color=blue] {$x^{1},\ov{x^{1}}$};
\node (12) at ({360/\n * (-1)}:\radius)  [circle, color=red] {$\ov{x^{1}},x^{2}$};
\node (13) at ({360/\n * (-2)}:\radius) [circle, color=blue] {$x^{2},\ov{x^{2}}$};
\node (14) at ({360/\n * (-3)}:\radius) [circle, color=red] {$\ov{x^{2}},x^{3}$};
\node (15) at ({360/\n * (-4)}:\radius) [circle, color=blue] {$x^{3},\ov{x^{3}}$};
\node (16) at ({360/\n * (-5)}:\radius) [circle, color=red] {$\ov{x^{3}},x^{4}$};
\node (17) at ({360/\n * (-6)}:\radius) [circle, color=blue] {$x^{4},\ov{x^{4}}$};
\node (18) at ({360/\n * (-7)}:\radius) [circle, color=red] {$\ov{x^{m}},x^{1}$};
\end{scope}

\begin{scope}[xshift=-2.9cm,yshift=4cm, rotate=-30]
\node (21) at ({360/\n * (0)}:\radius)  [circle, color=blue] {$y^{1},\ov{y^{1}}$};
\node (22) at ({360/\n * (-1)}:\radius)  [circle, color=red] {$\ov{y^{1}},y^{2}$};
\node (23) at ({360/\n * (-2)}:\radius) [circle, color=blue] {$y^{2},\ov{y^{2}}$};
\node (24) at ({360/\n * (-3)}:\radius) [circle, color=red] {$\ov{y^{2}},y^{3}$};
\node (25) at ({360/\n * (-4)}:\radius) [circle, color=blue] {$y^{3},\ov{y^{3}}$};
\node (26) at ({360/\n * (-5)}:\radius) [circle, color=red] {$\ov{y^{3}},y^{4}$};
\node (27) at ({360/\n * (-6)}:\radius) [circle, color=blue] {$y^{4},\ov{y^{4}}$};
\node (28) at ({360/\n * (-7)}:\radius) [circle, color=red] {$\ov{y^{m}},y^{1}$};
\end{scope}

\begin{scope}[xshift=0,yshift=-5cm,rotate=90]
\node (31) at ({360/\n * (0)}:\radius)  [circle, color=blue] {$z^{1},\ov{z^{1}}$};
\node (32) at ({360/\n * (-1)}:\radius)  [circle, color=red] {$\ov{z^{1}},z^{2}$};
\node (33) at ({360/\n * (-2)}:\radius) [circle, color=blue] {$z^{2},\ov{z^{2}}$};
\node (34) at ({360/\n * (-3)}:\radius) [circle, color=red] {$\ov{z^{2}},z^{3}$};
\node (35) at ({360/\n * (-4)}:\radius) [circle, color=blue] {$z^{3},\ov{z^{3}}$};
\node (36) at ({360/\n * (-5)}:\radius) [circle, color=red] {$\ov{z^{3}},z^{4}$};
\node (37) at ({360/\n * (-6)}:\radius) [circle, color=blue] {$z^{4},\ov{z^{4}}$};
\node (38) at ({360/\n * (-7)}:\radius) [circle, color=red] {$\ov{z^{m}},z^{1}$};
\end{scope}

\end{scope}

\path[->, line width=1pt, auto] 
	(a2) 	edge[color=black] 	node {} (a1) 
	(a1) 	edge[color=black] 	node {} (b1)
	(b1) 	edge[color=black] 	node {} (c1)
	(c1) 	edge[color=black] 	node {} (a1)
	(b2) 	edge[color=black] 	node {} (b1)
	(c2) 	edge[color=black]	node {} (c1)
	(a3) 	edge[color=orange] 	node {} (a2)
	(b3) 	edge[color=orange] 	node {} (b2)
	(c3) 	edge[color=orange] 	node {} (c2);

\path[->, line width=1pt, auto] 
	(11) 	edge[color=orange] 	node {} (a3)
	(22) 	edge[color=orange] 	node {} (b3)
	(31) 	edge[color=orange] 	node {} (c3);

\node (v3) at (-5,0) {\Large $z$};
\node (v2) at (4,2.9) {\Large $y$};
\node (v1) at (4,-2.5) {\Large $x$};
\node (c1) at (0,0) {\Large $C_1$};

\path[->, line width=1pt, auto] 
	(11) 	edge[color=black] 	node {} (12) 
	(12) 	edge[color=black] 	node {} (13)
	(13) 	edge[color=black] 	node {} (14)
	(14) 	edge[color=black] 	node {} (15)
	(15) 	edge[color=black] 	node {} (16)
	(16) 	edge[color=black]	node {} (17)
	(18) 	edge[color=black] 	node {} (11);
	
\path[->, line width=1pt, dotted]
	(17) 	edge[color=black] 	node {} (18);

\path[->, line width=1pt, auto] 
	(21) 	edge[color=black] 	node {} (22) 
	(22) 	edge[color=black] 	node {} (23)
	(23) 	edge[color=black] 	node {} (24)
	(24) 	edge[color=black] 	node {} (25)
	(25) 	edge[color=black] 	node {} (26)
	(26) 	edge[color=black]	node {} (27)
	(28) 	edge[color=black] 	node {} (21);
	
\path[->, line width=1pt, dotted]
	(27) 	edge[color=black] 	node {} (28);

\path[->, line width=1pt, auto] 
	(31) 	edge[color=black] 	node {} (32) 
	(32) 	edge[color=black] 	node {} (33)
	(33) 	edge[color=black] 	node {} (34)
	(34) 	edge[color=black] 	node {} (35)
	(35) 	edge[color=black] 	node {} (36)
	(36) 	edge[color=black]	node {} (37)
	(38) 	edge[color=black] 	node {} (31);
	
\path[->, line width=1pt, dotted]
	(37) 	edge[color=black] 	node {} (38);

\end{tikzpicture}
\caption{Part of $D_\psi$ for variables $x,y,z$ and clause $C_1 = (x \lor \ov{y} \lor   z)$ in the proof of  \Cref{chain is npc}.
  In the variable gadgets, blue nodes represent \true and red nodes, \false.}
\label{Fig_red_RxyRyz}
\end{figure*}

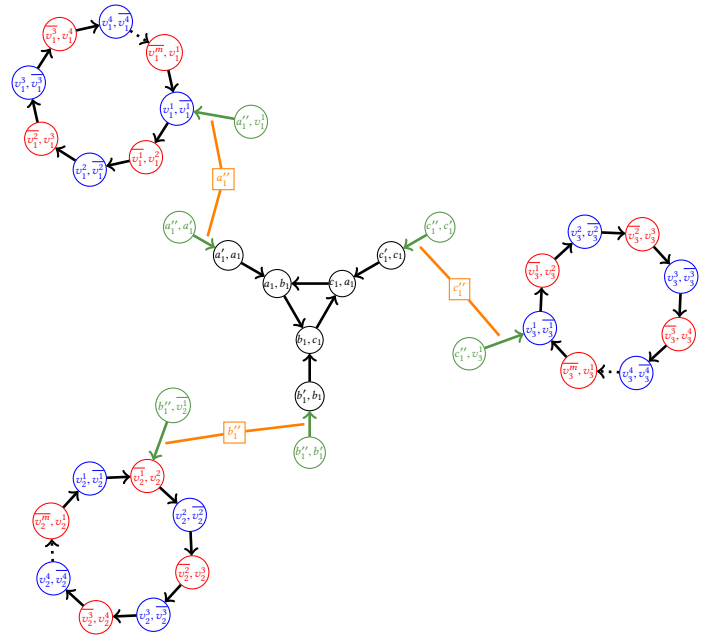
\begin{figure*}[t]
\centering

\begin{tikzpicture}[scale=.5,
			    every circle node/.style={fill=white, minimum size=7mm, inner sep=0, scale=.5, draw}]


\def \radius {1cm}

\begin{scope}[rotate=90]

\begin{scope}[rotate=60]
\node (a1) at ({360/3 * (0)}:\radius) [circle] {$a_{1},b_{1}$};
\node (b1) at ({360/3 * (1)}:\radius) [circle] {$b_{1},c_{1}$};
\node (c1) at ({360/3 * (2)}:\radius) [circle] {$c_{1},a_{1}$};
\node (a2) at ({360/3 * (0)}:{\radius+1.5cm}) [circle] {$a_{1}',a_{1}$};
\node (b2) at ({360/3 * (1)}:{\radius+1.5cm}) [circle] {$b_{1}',b_{1}$};
\node (c2) at ({360/3 * (2)}:{\radius+1.5cm}) [circle] {$c_{1}',c_{1}$};
\node (a3) at ({360/3 * (0)}:{\radius+3cm}) [circle, color=dg] {$a_{1}'',a_{1}'$};
\node (b3) at ({360/3 * (1)}:{\radius+3cm}) [circle, color=dg] {$b_{1}'', b_{1}'$};
\node (c3) at ({360/3 * (2)}:{\radius+3cm}) [circle, color=dg] {$c_{1}'',c_{1}'$};

\node (a4) at ({-25}:{\radius+3cm}) [draw, scale=.5, color=orange] {$a_{1}''$};
\node (b4) at ({-270}:{\radius+3cm}) [draw,scale=.5,  color=orange] {$b_{1}''$};
\node (c4) at ({-145}:{\radius+3cm}) [draw,scale=.5,  color=orange] {$c_{1}''$};

\end{scope}

\def \radius {2cm}
\def \n {8}

\begin{scope}[xshift=5.5cm, yshift=5.5cm, rotate=-100]
\node (11) at ({360/\n * (0)}:\radius)  [circle, color=blue] {$v^{1}_{1},\ov{v^{1}_{1}}$};
\node (12) at ({360/\n * (-1)}:\radius)  [circle, color=red] {$\ov{v^{1}_{1}},v^{2}_{1}$};
\node (13) at ({360/\n * (-2)}:\radius) [circle, color=blue] {$v^{2}_{1},\ov{v^{2}_{1}}$};
\node (14) at ({360/\n * (-3)}:\radius) [circle, color=red] {$\ov{v^{2}_{1}},v^{3}_{1}$};
\node (15) at ({360/\n * (-4)}:\radius) [circle, color=blue] {$v^{3}_{1},\ov{v^{3}_{1}}$};
\node (16) at ({360/\n * (-5)}:\radius) [circle, color=red] {$\ov{v^{3}_{1}},v^{4}_{1}$};
\node (17) at ({360/\n * (-6)}:\radius) [circle, color=blue] {$v^{4}_{1},\ov{v^{4}_{1}}$};
\node (18) at ({360/\n * (-7)}:\radius) [circle, color=red] {$\ov{v^{m}_{1}},v^{1}_{1}$};

\node (19) at ({360/\n * (0)}:\radius+2cm)  [circle, color=dg] {$a_{1}'',v^{1}_{1}$};
\end{scope}

\begin{scope}[xshift=-6.5cm,yshift=5cm, rotate=25]
\node (21) at ({360/\n * (0)}:\radius)  [circle, color=blue] {$v^{1}_{2},\ov{v^{1}_{2}}$};
\node (22) at ({360/\n * (-1)}:\radius)  [circle, color=red] {$\ov{v^{1}_{2}},v^{2}_{2}$};
\node (23) at ({360/\n * (-2)}:\radius) [circle, color=blue] {$v^{2}_{2},\ov{v^{2}_{2}}$};
\node (24) at ({360/\n * (-3)}:\radius) [circle, color=red] {$\ov{v^{2}_{2}},v^{3}_{2}$};
\node (25) at ({360/\n * (-4)}:\radius) [circle, color=blue] {$v^{3}_{2},\ov{v^{3}_{2}}$};
\node (26) at ({360/\n * (-5)}:\radius) [circle, color=red] {$\ov{v^{3}_{2}},v^{4}_{2}$};
\node (27) at ({360/\n * (-6)}:\radius) [circle, color=blue] {$v^{4}_{2},\ov{v^{4}_{2}}$};
\node (28) at ({360/\n * (-7)}:\radius) [circle, color=red] {$\ov{v^{m}_{2}},v^{1}_{2}$};

\node (29) at ({360/\n * (-1)}:\radius+2cm)  [circle, color=dg] {$b_{1}'',\ov{v^{1}_{2}}$};
\end{scope}

\begin{scope}[xshift=0,yshift=-8cm,rotate=110]
\node (31) at ({360/\n * (0)}:\radius)  [circle, color=blue] {$v^{1}_{3},\ov{v^{1}_{3}}$};
\node (32) at ({360/\n * (-1)}:\radius)  [circle, color=red] {$\ov{v^{1}_{3}},v^{2}_{3}$};
\node (33) at ({360/\n * (-2)}:\radius) [circle, color=blue] {$v^{2}_{3},\ov{v^{2}_{3}}$};
\node (34) at ({360/\n * (-3)}:\radius) [circle, color=red] {$\ov{v^{2}_{3}},v^{3}_{3}$};
\node (35) at ({360/\n * (-4)}:\radius) [circle, color=blue] {$v^{3}_{3},\ov{v^{3}_{3}}$};
\node (36) at ({360/\n * (-5)}:\radius) [circle, color=red] {$\ov{v^{3}_{3}},v^{4}_{3}$};
\node (37) at ({360/\n * (-6)}:\radius) [circle, color=blue] {$v^{4}_{3},\ov{v^{4}_{3}}$};
\node (38) at ({360/\n * (-7)}:\radius) [circle, color=red] {$\ov{v^{m}_{3}},v^{1}_{3}$};

\node (39) at ({360/\n * (0)}:\radius+2cm)  [circle, color=dg] {$c_{1}'',v^{1}_{3}$};
\end{scope}

\path[->, line width=1pt, auto] 
	(a2) 	edge[color=black] 	node {} (a1) 
	(a1) 	edge[color=black] 	node {} (b1)
	(b1) 	edge[color=black] 	node {} (c1)
	(c1) 	edge[color=black] 	node {} (a1)
	(b2) 	edge[color=black] 	node {} (b1)
	(c2) 	edge[color=black]	node {} (c1);

\path[->, line width=1pt, color=dg, inner sep=0, auto] 
	(19) 	edge 	node (av) {} (11)
	(29) 	edge 	node (bv) {} (22)
	(39) 	edge 	node (cv) {} (31)
	(a3) 	edge 	node (ag) {} (a2)
	(b3) 	edge 	node (bg) {} (b2)
	(c3) 	edge 	node (cg) {} (c2);
	
\path[line width=1pt, color=orange] 
	(a4) edge node {} (av)
	(b4) edge node {} (bv)
	(c4) edge node {} (cv)
	(a4) edge node {} (ag)
	(b4) edge node {} (bg)
	(c4) edge node {} (cg);

\path[->, line width=1pt, auto] 
	(11) 	edge[color=black] 	node {} (12) 
	(12) 	edge[color=black] 	node {} (13)
	(13) 	edge[color=black] 	node {} (14)
	(14) 	edge[color=black] 	node {} (15)
	(15) 	edge[color=black] 	node {} (16)
	(16) 	edge[color=black]	node {} (17)
	(17) 	edge[color=black, dotted] 	node {} (18)
	(18) 	edge[color=black] 	node {} (11);

\path[->, line width=1pt, auto] 
	(21) 	edge[color=black] 	node {} (22) 
	(22) 	edge[color=black] 	node {} (23)
	(23) 	edge[color=black] 	node {} (24)
	(24) 	edge[color=black] 	node {} (25)
	(25) 	edge[color=black] 	node {} (26)
	(26) 	edge[color=black]	node {} (27)
	(27) 	edge[color=black ,dotted] 	node {} (28)
	(28) 	edge[color=black] 	node {} (21);
	
\path[->, line width=1pt, auto] 
	(31) 	edge[color=black] 	node {} (32) 
	(32) 	edge[color=black] 	node {} (33)
	(33) 	edge[color=black] 	node {} (34)
	(34) 	edge[color=black] 	node {} (35)
	(35) 	edge[color=black] 	node {} (36)
	(36) 	edge[color=black]	node {} (37)
	(37) 	edge[color=black, dotted] 	node {} (38)
	(38) 	edge[color=black] 	node {} (31);

\end{scope}

\end{tikzpicture}

\caption{Excerpt from the construct showing the gadget for clause $C_1 = (v_1 \vee \bar v_2 \vee v_3)$. We omit the $A$-tuples that participate
in only one witness, since they shall never be chose for a minimum contingency set, as well as green nodes.}
\label{Fig_red_AxRxyRyz}
\end{figure*}

We next show all of them are hard queries.

\begin{lemma}\label{bchain is npc}
$\res(\bchain)$ is \np-complete.
\end{lemma}

\begin{proof}[Proof of \Cref{bchain is npc}]
For this case we are going to use almost the same reduction as the one used for $\res(\chain)$, just with the added $B$-tuples. 
Then we argue that there is always a min $\Gamma$ that only uses $R$-tuples.

Let $\psi$ be a 3CNF formula with $n$ variables $v_1, \ldots, v_n$ and $m$ clauses 
$C_1, \ldots, C_{m}$. Our reduction will map  any such $\psi$ to a pair $(D_\psi,k_\psi)$ where $D_\psi$ is a database 
satisfying $\bchain$, and 
$$\psi\in 3\sat \quad\Leftrightarrow\quad (D_\psi,k_\psi) \in \res(\bchain)$$

In our construction, if $\psi \in 3\sat$, then the size of each minimum contingency set for 
$\bchain$ in $D_\psi$ will be $k_\psi=(n+5)m$, whereas if $\psi \not\in 3\sat$, then the size of all contingency sets
for $\bchain$ in $D_\psi$ will be greater than $k_\psi$.

First, include in $D_{\psi}$ all the same $R$-tuples included in the proof of \cref{chain is npc}. In addition to that add the following
$B$-tuples:

\begin{enumerate}

\item Variable gadget: 
For each variable $v_i$ and each $j \in [m]$ insert the following two tuples into the database:
$B(v_i^j)$ and $B(\ov{v_{i}^{j}})$.

\item Clause gadget:
For each clause $j \in [m]$ insert the following 6 tuples into the database:
$B(a_j)$, $B(b_j)$, $B(c_j)$,
$B(a_j')$, $B(b_j')$, $B(c_j')$.

\end{enumerate}

By adding those tuples, we obtain the same structure and witnesses of the reduction for $\res(\chain)$.
Now suppose that $t = B(d)$ is in a minimum contingency set $\Gamma$. 
If $d = v_{i}^{j}$ (or $\ov{v_{i}^{j}}$) for some $i,j$,
we know that $t$ must join with $t' = R(\ov{v_{i}^{j-1}}, v_{i}^{j})$ (or $R(\ov{v_{i}^{j}}, v_{i}^{j})$) by our construction. 
Thus, we can exchange $t$ for $t'$ and obtain contingency set $\Gamma'$. 
Similar, if $d \in \set{a,b,c,a',b',c'}$, then $t$ must join with tuple $R(d, *)$, since there is only tuple of that kind for each possible
value of $d$.

This shows that there is a minimum contingency set for $D_{\psi}$ without $B$-tuples, and the properties of 
the reduction in \cref{chain is npc} also hold in this case.
\end{proof}

\begin{lemma}\label{achain is npc}
$\res(\achain)$, $\res(\cchain)$, $\res(\abchain)$ and $\res(\bcchain)$ are \np-complete.
\end{lemma}

\begin{proof}[Proof of \Cref{achain is npc}]
We again define a reduction from 3\sat, using gadgets similar to the one in \cref{chain is npc}.
The variable gadget remains such that a minimum cover will choose either blue nodes 
(variable is set to \true), or red nodes (variable is set to \false).
The clause gadget (black nodes) is chosen as to enforce a clause: if one or more of the outermost black nodes are chosen, 
then the minimum cover is 5, otherwise 6.

We next reduce 3\sat\ to $\res(\achain)$.
Let $\psi$ be a 3CNF formula with $n$ variables $v_1, \ldots, v_n$ and $m$ clauses 
$C_1, \ldots, C_{m}$.  
Our reduction will map  any such $\psi$ to a pair $(D_\psi,k_\psi)$ where $D_\psi$ is a database 
satisfying $\achain$, and 
$$\psi\in 3\sat \quad\Leftrightarrow\quad (D_\psi,k_\psi) \in \res(\achain)$$

In our construction, if $\psi \in 3\sat$, then the size of each minimum contingency set for 
$\achain$ in $D_\psi$ will be $k_\psi=(n+5)m$, whereas if $\psi \not\in 3\sat$, then the size of all contingency sets
for $\achain$ in $D_\psi$ will be greater than $k_\psi$.

\begin{enumerate}

\item Variable gadget: 
For each variable $v_i$ and each $j \in [m]$ insert the following tuples into the database:
$R(v_i^j,\ov{v_{i}^{j}})$, 
$R(\ov{v_{i}^{j}}, v_{i}^{j+1})$ and
$A(v_{i}^{j})$, 
$A(\ov{v_{i}^{j}})$. If $j+1 > m$, then make the superscript 1.
The resulting witnesses between the tuples form a cycle of length~$2m$. The minimum contingency sets are to either choose all 
tuples $R(v_i^j,\ov{v_{i}^{j}})$ representing a variable to have assignment \true,
or all tuples $R(\ov{v_{i}^{j}}, v_{i}^{j+1})$ representing a variable to have assignment \false. Note that any $A$-tuple only joins once,
therefore it is better to choose an $R$-tuple, since all of these join at least twice.

\item Clause gadget:
For each clause $j \in [m]$ insert the following tuples into the database:
$R(a_j,b_j)$, $R(b_j,c_j)$, $R(c_j,a_j)$,
$R(a_j',a_j)$, $R(b_j',b_j)$, $R(c_j',c_j)$,
$A(a_j)$, $A(b_j)$, $A(c_j)$,
$A(a_j')$, $A(b_j')$, $A(c_j')$.
The resulting witnesses form a triangle. If either of the $R(*',*)$ is removed, then the remaining witnesses can be destroyed 
by choosing only 2 or more tuples, otherwise we need 3. Similar to the variable gadget, $A$-tuples are not an optimal 
choice because they only participate in one witness each.

\item Connecting the gadgets:
For each variable $i$ that appears in clause $j$ at position 1, add the following tuples: 
$R(a_{j}'',a_{j}')$ and $A(a_{j}'')$. If $v_{i}$ appears as positive add tuple $R(a_{j}'', v_{i}^{j})$, 
if it appear as negative add tuple $R(a_{j}'', \ov{v_{i}^{j}})$.
Analogously use $b_j', b_{j}''$ or $c_{j}', c_j''$ instead of $a_j', a_{j}''$ for positions 2 and 3 instead of position 1.
\end{enumerate}

Observe that if the clause is not satisfied, then we need to choose the $A$-tuples (orange squares 
in \cref{Fig_red_AxRxyRyz}), and not choose the outer black nodes ($R$-tuples) in the clause gadget, resulting in choosing 6 tuples
in total in order to delete all the witnesses, otherwise we just need 5 tuples. 

The reduction for $\abchain$ is very similar to the one presented above. 
First, use the same $D_{\psi}$ just adding the appropriate $B$-tuples, i.e., $B$-tuples that preserve the witnesses.

Now note that for any $t = B(d) \in D_{\psi}$, there is only one $R$-tuple such that $t' = R(d,*)$, therefore $t$ must join with $t'$.
Therefore, any occurrence of $B$-tuple in a contingency set can be exchanged by its correspondent $R$-tuple, and we are guaranteed
this reduction has the same properties as the one for $\achain$.
\end{proof}

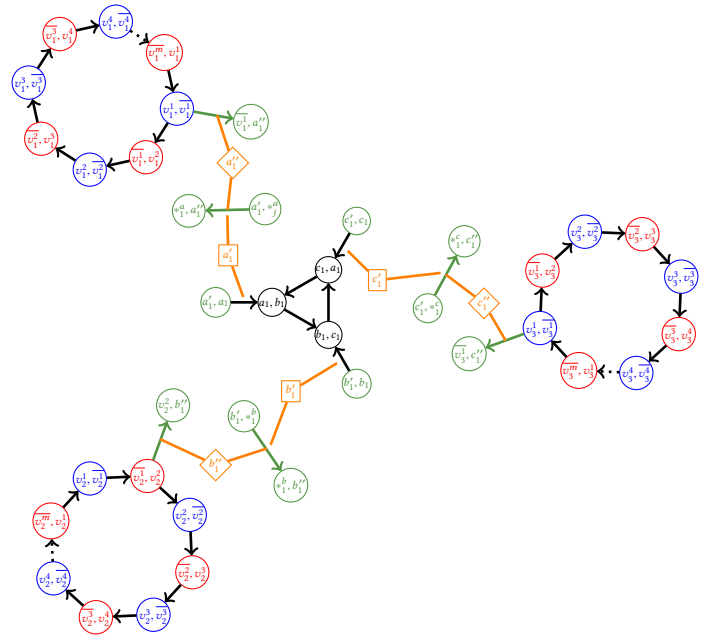
\begin{figure*}[t]
\centering

\begin{tikzpicture}[scale=.5,
			    every circle node/.style={fill=white, minimum size=7mm, inner sep=0, scale=.5, draw}]


\def \radius {1cm}

\begin{scope}[rotate=90]
\begin{scope}[rotate=90]
\node (a1) at ({360/3 * (0)}:\radius) [circle] {$a_{1},b_{1}$};
\node (b1) at ({360/3 * (1)}:\radius) [circle] {$b_{1},c_{1}$};
\node (c1) at ({360/3 * (2)}:\radius) [circle] {$c_{1},a_{1}$};
\node (a2) at ({360/3 * (0)}:{\radius+1.5cm}) [circle, color=dg] {$a_{1}',a_{1}$};
\node (b2) at ({360/3 * (1)}:{\radius+1.5cm}) [circle, color=dg] {$b_{1}',b_{1}$};
\node (c2) at ({360/3 * (2)}:{\radius+1.5cm}) [circle, color=dg] {$c_{1}',c_{1}$};
\end{scope}

\begin{scope}[rotate=60]
\begin{scope}[xshift=0cm, yshift=-.5cm]
\node (a3) at ({360/3 * (0)}:{\radius+3cm}) [circle, color=dg] {$*_{1}^{a},a_{1}''$};
\node (a5) at ({-25}:{\radius+1.5cm}) [circle, color=dg] {$a_{1}', *_{j}^{a}$};
\node (a4) at ({-25}:{\radius+3cm}) [diamond, draw, inner sep=1, color=orange, scale=.5] {$a_{1}''$};
\end{scope}
\node (a6) at ({0}:{\radius+1.5cm}) [draw, color=orange, scale=.5] {$a_{1}'$};
\begin{scope}[xshift=0cm, yshift=1cm]
\node (b3) at ({360/3 * (1)}:{\radius+3cm}) [circle, color=dg] {$*_{1}^{b}, b_{1}''$};
\node (b5) at ({-270}:{\radius+1.5cm}) [circle, color=dg] {$b_{1}', *_{1}^{b}$};
\node (b4) at ({-270}:{\radius+3cm}) [diamond, draw, inner sep=1, color=orange, scale=.5] {$b_{1}''$};
\node (b6) at ({-240}:{\radius+.5cm}) [draw, color=orange, scale=.5] {$b_{1}'$};
\end{scope}
\begin{scope}[xshift=-.75cm, yshift=0cm]
\node (c3) at ({360/3 * (2)}:{\radius+3cm}) [circle, color=dg] {$*_{1}^{c},c_{1}''$};
\node (c5) at ({-145}:{\radius+1.5cm}) [circle, color=dg] {$c_{1}', *_{1}^{c}$};
\node (c4) at ({-145}:{\radius+3cm}) [diamond, draw, inner sep=1, color=orange, scale=.5] {$c_{1}''$};
\node (c6) at ({-110}:{\radius+.5cm}) [draw, color=orange, scale=.5] {$c_{1}'$};
\end{scope}

\end{scope}

\def \radius {2cm}
\def \n {8}

\begin{scope}[xshift=5.5cm, yshift=5.5cm, rotate=-100]
\node (11) at ({360/\n * (0)}:\radius)  [circle, color=blue] {$v^{1}_{1},\ov{v^{1}_{1}}$};
\node (12) at ({360/\n * (-1)}:\radius)  [circle, color=red] {$\ov{v^{1}_{1}},v^{2}_{1}$};
\node (13) at ({360/\n * (-2)}:\radius) [circle, color=blue] {$v^{2}_{1},\ov{v^{2}_{1}}$};
\node (14) at ({360/\n * (-3)}:\radius) [circle, color=red] {$\ov{v^{2}_{1}},v^{3}_{1}$};
\node (15) at ({360/\n * (-4)}:\radius) [circle, color=blue] {$v^{3}_{1},\ov{v^{3}_{1}}$};
\node (16) at ({360/\n * (-5)}:\radius) [circle, color=red] {$\ov{v^{3}_{1}},v^{4}_{1}$};
\node (17) at ({360/\n * (-6)}:\radius) [circle, color=blue] {$v^{4}_{1},\ov{v^{4}_{1}}$};
\node (18) at ({360/\n * (-7)}:\radius) [circle, color=red] {$\ov{v^{m}_{1}},v^{1}_{1}$};

\node (19) at ({360/\n * (0)}:\radius+2cm)  [circle, color=dg] {$\ov{v^{1}_{1}},a_{1}''$};
\end{scope}

\begin{scope}[xshift=-6.5cm,yshift=5cm, rotate=25]
\node (21) at ({360/\n * (0)}:\radius)  [circle, color=blue] {$v^{1}_{2},\ov{v^{1}_{2}}$};
\node (22) at ({360/\n * (-1)}:\radius)  [circle, color=red] {$\ov{v^{1}_{2}},v^{2}_{2}$};
\node (23) at ({360/\n * (-2)}:\radius) [circle, color=blue] {$v^{2}_{2},\ov{v^{2}_{2}}$};
\node (24) at ({360/\n * (-3)}:\radius) [circle, color=red] {$\ov{v^{2}_{2}},v^{3}_{2}$};
\node (25) at ({360/\n * (-4)}:\radius) [circle, color=blue] {$v^{3}_{2},\ov{v^{3}_{2}}$};
\node (26) at ({360/\n * (-5)}:\radius) [circle, color=red] {$\ov{v^{3}_{2}},v^{4}_{2}$};
\node (27) at ({360/\n * (-6)}:\radius) [circle, color=blue] {$v^{4}_{2},\ov{v^{4}_{2}}$};
\node (28) at ({360/\n * (-7)}:\radius) [circle, color=red] {$\ov{v^{m}_{2}},v^{1}_{2}$};

\node (29) at ({360/\n * (-1)}:\radius+2cm)  [circle, color=dg] {$v^{2}_{2}, b_{1}''$};
\end{scope}

\begin{scope}[xshift=0,yshift=-8cm,rotate=110]
\node (31) at ({360/\n * (0)}:\radius)  [circle, color=blue] {$v^{1}_{3},\ov{v^{1}_{3}}$};
\node (32) at ({360/\n * (-1)}:\radius)  [circle, color=red] {$\ov{v^{1}_{3}},v^{2}_{3}$};
\node (33) at ({360/\n * (-2)}:\radius) [circle, color=blue] {$v^{2}_{3},\ov{v^{2}_{3}}$};
\node (34) at ({360/\n * (-3)}:\radius) [circle, color=red] {$\ov{v^{2}_{3}},v^{3}_{3}$};
\node (35) at ({360/\n * (-4)}:\radius) [circle, color=blue] {$v^{3}_{3},\ov{v^{3}_{3}}$};
\node (36) at ({360/\n * (-5)}:\radius) [circle, color=red] {$\ov{v^{3}_{3}},v^{4}_{3}$};
\node (37) at ({360/\n * (-6)}:\radius) [circle, color=blue] {$v^{4}_{3},\ov{v^{4}_{3}}$};
\node (38) at ({360/\n * (-7)}:\radius) [circle, color=red] {$\ov{v^{m}_{3}},v^{1}_{3}$};

\node (39) at ({360/\n * (0)}:\radius+2cm)  [circle, color=dg] {$\ov{v^{1}_{3}},c_{1}''$};
\end{scope}

\path[->, line width=1pt, inner sep=0, auto] 
	(a2) 	edge[color=black] 	node (ag2) {} (a1) 
	(a1) 	edge[color=black] 	node {} (b1)
	(b1) 	edge[color=black] 	node {} (c1)
	(c1) 	edge[color=black] 	node {} (a1)
	(b2) 	edge[color=black] 	node (bg2) {} (b1)
	(c2) 	edge[color=black]	node (cg2) {} (c1);

\path[->, line width=1pt, color=orange, inner sep=0, auto] 
	(11) 	edge [color=dg] 	node (av) {} (19)
	(22) 	edge [color=dg]     	node (bv) {} (29)
	(31) 	edge [color=dg]	 	node (cv) {} (39)
	(a5) 	edge [color=dg] 	node (ag1) {} (a3)
	(b5) 	edge [color=dg] 	node (bg1) {} (b3)
	(c5) 	edge [color=dg] 	node (cg1) {} (c3);
	
\path[line width=1pt, color=orange, inner sep=0, auto] 
	(a4) edge node {} (av)
	(b4) edge node {} (bv)
	(c4) edge node {} (cv)
	(a4) edge node {} (ag1)
	(b4) edge node {} (bg1)
	(c4) edge node {} (cg1)
	(b6) edge node {} (bg1)
	(b6) edge node {} (bg2)
	(a6) edge node {} (ag1)
	(a6) edge node {} (ag2)
	(c6) edge node {} (cg1)
	(c6) edge node {} (cg2);

\path[->, line width=1pt, auto] 
	(11) 	edge[color=black] 	node {} (12) 
	(12) 	edge[color=black] 	node {} (13)
	(13) 	edge[color=black] 	node {} (14)
	(14) 	edge[color=black] 	node {} (15)
	(15) 	edge[color=black] 	node {} (16)
	(16) 	edge[color=black]	node {} (17)
	(17) 	edge[color=black, dotted] 	node {} (18)
	(18) 	edge[color=black] 	node {} (11);

\path[->, line width=1pt, auto] 
	(21) 	edge[color=black] 	node {} (22) 
	(22) 	edge[color=black] 	node {} (23)
	(23) 	edge[color=black] 	node {} (24)
	(24) 	edge[color=black] 	node {} (25)
	(25) 	edge[color=black] 	node {} (26)
	(26) 	edge[color=black]	node {} (27)
	(27) 	edge[color=black, dotted] 	node {} (28)
	(28) 	edge[color=black] 	node {} (21);

\path[->, line width=1pt, auto] 
	(31) 	edge[color=black] 	node {} (32) 
	(32) 	edge[color=black] 	node {} (33)
	(33) 	edge[color=black] 	node {} (34)
	(34) 	edge[color=black] 	node {} (35)
	(35) 	edge[color=black] 	node {} (36)
	(36) 	edge[color=black]	node {} (37)
	(37) 	edge[color=black, dotted] 	node {} (38)	
	(38) 	edge[color=black] 	node {} (31);
	
\end{scope}

\end{tikzpicture}

\caption{Excerpt from the construct showing the gadget for clause $C_1 = (v_1 \vee \bar v_2 \vee v_3)$. We omit the $A$-tuples and $C$-tuples
that would not be chosen for a minimum contingency set, as well as green nodes.}
\label{Fig_red_AxCzRxyRyz}
\end{figure*}

\begin{lemma}\label{acchain is npc}
$\res(\acchain)$ and $\res(\abcchain)$ are \np-complete.
\end{lemma}

\begin{proof}[Proof of \Cref{acchain is npc}]
We define a reduction from 3\sat.
As in the previous cases, the variable gadget remains such that a minimum cover will choose either blue nodes 
(variable is set to \true), or red nodes (variable is set to \false).
The clause gadget (center black nodes) is chosen as to enforce a clause: if one or more of the outermost joins (black edges) are deleted 
by choosing the corresponding $A$-tuple (orange square), then the minimum cover for the black subgraph is 2, otherwise 3.

We next reduce 3\sat\ to $\res(\acchain)$.
Let $\psi$ be a 3CNF formula with $n$ variables $v_1, \ldots, v_n$ and $m$ clauses 
$C_1, \ldots, C_{m}$.  
Our reduction will map  any such $\psi$ to a pair $(D_\psi,k_\psi)$ where $D_\psi$ is a database 
satisfying $\chain$, and 
$$\psi\in 3\sat \quad\Leftrightarrow\quad (D_\psi,k_\psi) \in \res(\acchain)$$

In our construction, if $\psi \in 3\sat$, then the size of each minimum contingency set for 
$\acchain$ in $D_\psi$ will be $k_\psi=(n+5)m$, whereas if $\psi \not\in 3\sat$, then the size of all contingency sets
for $\acchain$ in $D_\psi$ will be greater than $k_\psi$.

\begin{enumerate}

\item Variable gadget: 
For each variable $v_i$ and each $j \in [m]$ insert the following tuples into the database:
$R(v_i^j,\ov{v_{i}^{j}})$, $R(\ov{v_{i}^{j}}, v_{i}^{j+1})$ and
$A(v_{i}^{j})$, $A(\ov{v_{i}^{j}})$ and
$C(v_{i}^{j})$, $C(\ov{v_{i}^{j}})$. 
If $j+1 > m$, then make the superscript 1.
The resulting witnesses between the tuples form a cycle of length $2m$. The minimum contingency sets are to either choose all 
tuples $R(v_i^j,\ov{v_{i}^{j}})$ representing a variable to have assignment \true,
or all tuples $R(\ov{v_{i}^{j}}, v_{i}^{j+1})$ representing a variable to have assignment \false. If we only consider those tuples, note
that $A$- and $C$-tuples participate in only one witness, so the optimal choice is to delete $R$-tuples. 

\item Clause gadget:
For each clause $j \in [m]$ insert the following tuples into the database:
$R(a_j,b_j)$, $R(b_j,c_j)$, $R(c_j,a_j)$,
$R(a_j',a_j)$, $R(b_j',b_j)$, $R(c_j',c_j)$,
$A(a_j)$, $A(b_j)$, $A(c_j)$,
$A(a_j')$, $A(b_j')$, $A(c_j')$,
$C(a_j)$, $C(b_j)$, $C(c_j)$.
The resulting witnesses form a triangle. If either of the $A(*')$ is removed, then the remaining witnesses can be destroyed 
by choosing only 2 or more tuples, otherwise we need 3. We later argue that these tuples only need be $R$-tuples.

\item Connecting the gadgets:
For each variable $i$ that appears in clause $j$ at position 1, add the following tuples: 
$R(a_{j}',*_{j}^{a}), R(*_{j}^{a}, a_{j}'')$ and $C(a_{j}'')$. If $v_{i}$ appears as positive add tuple $R(\ov{v_{i}^{j}}, a_{j}'')$, 
if it appear as negative add tuple $R(v_{i}^{j}, a_{j}'')$.
Analogously use $b_j', b_{j}''$ or $c_{j}', c_j''$ instead of $a_j', a_{j}''$ for positions 2 and 3 instead of position 1.
\end{enumerate}

With our gadget, if the clause cannot be satisfied, then we need to choose all the $C$-tuples (orange diamonds on \cref{Fig_red_AxCzRxyRyz}),
since we can delete two witnesses by doing deleting each. In that case, in order to delete the remaining witnesses we need to delete 3 tuples, namely
the 3 black nodes in the triangle, resulting on the total deletion of 6 tuples.  

We now need to argue that, besides the tuples depicted in \cref{Fig_red_AxCzRxyRyz}, we don't need other $A$- or $C$-tuples for a minimum
contingency set. Assume there is a tuple $t = A(d)$ in a min $\Gamma$. Given that $d \notin \set{a_{j}',b_{j}',c_{j}'}$, our construction guarantees
there is only one $R$-tuple such that $t' = R(d,-)$, therefore we can have $\Gamma' = \Gamma - t + t'$, and $\Gamma'$ is also a minimum contingency
set. Similarly, if there is a tuple $t = C(d)$ in $\Gamma$, and assuming $d \notin \set{a_{j}'',b_{j}'',c_{j}''}$, there is only one $R$-tuple
$t' = R(-, d)$, and therefore the same follows.

For $\abcchain$ use almost the same construction as above. We just add the appropriate
$B$-tuples and show that there is a minimum contingency set that does not contain those.

Consider $D_{\psi}$ as initially defined for $\acchain$. Now we include the appropriate $B$-tuples:
\begin{enumerate}

\item Variable gadget: 
For each variable $v_i$ and each $j \in [m]$ insert the following tuples into the database:
$B(v_{i}^{j})$, $B(\ov{v_{i}^{j}})$

\item Clause gadget:
For each clause $j \in [m]$ insert the following tuples into the database:
$B(a_j)$, $B(b_j)$, $B(c_j)$.

\item Connecting the gadgets:
For each variable $i$ that appears in clause $j$ at position 1, add tuple $B(*_{j}^{a})$.
Analogously $B(*_{j}^{b})$ and $B(*_{j}^{c})$ for positions 2 and 3, respectively.
\end{enumerate}

By adding those $B$-tuples we obtain the same witnesses we saw in the reduction for $\res(\acchain)$. With this construction we guarantee
that for any tuple $t = B(d)$, there is either only one tuple $R(d, -)$ or only one tuple $R(-, d)$, which means we can always
choose one of those $R$-tuples instead and obtain another minimum contingency set without $B$-tuples. 
\end{proof}

\begin{proof}[Proof of \cref{2chains are hard}]
Suppose that $R(x,y), R(y,z)$ are the unique $R$-atoms in $q$.  Assume first that there are no unary
atoms $A(x), B(y), C(z)$. We 
define a reduction from $\res(\chain)$ to $\res(q)$ as follows:

Consider a database $D$ with $D \models \chain$ and we may assume that there are no loops $R(a,a) \in D$, since those 
would have to be in any~$\Gamma$. We define a new database $D'$ such that for each atom $S_i(v_1, v_2)$ or $A(v)$ 
occurring in $q$, we define 
\begin{align*}
S_i &= \bigset{ (t(v_1,a,b,c), t(v_2,a,b,c))}{D \models \chain(a,b,c)}\\
A &= \bigset{(t(v,a,b,c))}{D \models \chain(a,b,c)}
\end{align*}
where
\[
t(v,a,b) \;\eqdef\;
\begin{cases}
a   & \textrm{if } v = x \\
b   & \textrm{if } v = y \\
c   & \textrm{if } v = z \\
\angle{abc}_v   & \textrm{otherwise.}
\end{cases}
\]

Now we want to show
$$(D,k) \in \res(\chain) \Leftrightarrow (D',k) \in \res(q)$$

Notice that this mapping from $D$ to $D'$ preserves the witnesses in~$D,\chain$. Moreover, there are no new witnesses created
where variables $x,y,z$ are mapped to values that did not correspond to witnesses before.
Since $q$ is pseudo-linear, no endogenous atom of $q$ contains both $x$ and $z$. 
Therefore, any minimum contingency set for $D,\chain$ is also a minimum contingency set for $D',q$.
This completes our reduction.

Now, if any subset of unary relations $A(x), B(y), C(z)$ does appear in $q$, then we define a
reduction from the appropriate  
unary expansion of $\chain$. The same mapping used above to define $D'$ from $D$ preserves all
minimum contingency sets, as desired. 
\end{proof}

\subsection{Proofs for \cref{sec: 2conf}}

\begin{proof}[Proof of \Cref{prop: 2conf}]
For $q \datarule q_\ell,R(x,y),q_m,R(z,y), q_r$, let $D$ be any database satisfying $q$ and let
$j$ be a witness of $D$ satisfying $q$. Note that if $y$ occurs in $q_\ell$ then, by linearity, it must
be as an atom  $F(x,y)$ immediately to the left of $R(x,y)$.  Furthermore, any such atom may be
considered exogenous because it is never better to choose $F(a,b)$ over~$R(a,b)$.  Furthermore, if
$x$ occurs in $q_m$, then it would be via an atom $F(x,y)$ immediately to the right of $R(x,y)$.  If
so, we can assume it is immediately to the left of $R(x,y)$.  In particular, we may assume that
neither $x$ nor $z$ occurs in $q_m$.

We can write $j = A(a,b)R(a,b)B(b)R(c,b)C(b,c)$ where $A(a,b)$, $B(b)$, $C(b,c)$
stand for the atoms of $q_\ell(a,b), q_m(b), q_r(b,c)$, respectively.

Let $N_D$ be a network flow for $D,q$ ignoring the fact that $q$ has a self-join.  Thus $N_D$ has
duplicates edges for its $R$-tuples, i.e., for each $R(a,b) \in D$ 
there are two edges, $R_\ell(a,b),R_r(a,b)$ in $N_D$.  Assume that each edge corresponding to an
endogenous, resp.\  exogenous tuple has weight 1, resp.\  $\infty$.

Let $M$ be a min cut for $N_D$.  Let $\Gamma_M$ be the corresponding set of atoms of $D$, where any
edges $R_\ell(a,b), R_r(a,b)$ are replaced by the atom $R(a,b)$.  Observe that since there is no flow
through $N_D - M$, $\Gamma_M$ is a contingency set for $(D,q)$.

We claim that in fact $\Gamma_M$ is a minimum contingency set for $(D, q)$.
The key idea is the following:

\begin{lemma}\label{minimal cut lemma}
Let $M$ be a minimal cut of $N_D$.  Then $M$ does not include more than one instance of
any $R$ tuple.
\end{lemma}

\begin{proof}
Suppose to the contrary, that $M$ is a minimal cut for $N_D$ and  contains both $R_\ell(a,b)$ and
$R_r(a,b)$.  Since $M$ is minimal, it follows that $N_D - (M - \set{R_\ell(a,b)})$ and $N_D - (M -
\set{R_r(a,b)})$ both contain flows:

$f_1 = A(a,b) R_\ell(a,b)B(b) R(c,b) C(b,c)$ \quad \text{and}

$f_2 = A(a',b) R(a',b)B(b) R_r(a,b) C(b,a)$. \quad
But then $N_d -M$ contains the flow

$f = A(a',b) R(a',b)B(b) R(c,b) C(b,c)$, \quad
contradicting the fact that $M$ is a cut. See \cref{fig: flow 2conf} for a depiction in
the graph.
\end{proof}

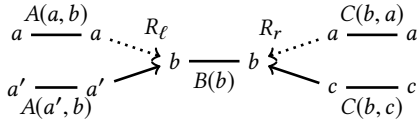
\begin{figure}[h]
\begin{tikzpicture}[scale=.35,
			every circle node/.style={fill=white, size=7mm, inner sep=0, draw}]

\node (zero) at (-16,0){};

\node (x1) at (-12,0) [] {$a$};
\node (x2) at (-12,-2) [] {$a'$};

\node (x1') at (-9,0) [] {$a$};
\node (x2') at (-9,-2) [] {$a'$};

\node (y1) at (-6,-1) [] {$b$};
\node (y2) at (-3,-1) [] {$b$};

\node (z1') at (0,0) [] {$a$};
\node (z2') at (0,-2) [] {$c$};

\node (z1) at (3,0) [] {$a$};
\node (z2) at (3,-2) [] {$c$};

\path[->, line width=1pt, auto]
	(x1')	edge[color=black,  dotted] node {$R_\ell$}(y1)
	(x2')	edge[color=black] (y1)
	(z1')	edge[color=black, dotted] node [swap] {$R_r$}(y2)
	(z2')	edge[color=black] (y2)

	;

\path[-, line width=1pt, auto]
	(x1)	edge[] node {$A(a,b)$} (x1')
	(x2)	edge[] node [swap]{$A(a',b)$} (x2')
	(y1)	edge[] node [swap]{$B(b)$} (y2)
	(z1)	edge[] node [above] {$C(b,a)$} (z1')
	(z2)	edge[] node {$C(b,c)$} (z2')
	
;

\end{tikzpicture}
\caption{Graph depicting flow $f_1, f_2, f$ described in \cref{minimal cut lemma}. Flow $f$ is represent by bold edges.}
\label{fig: flow 2conf}
\end{figure}

Now, let $\Gamma$ be any contingency set.  We claim  that $\Gamma$ is the same size as some
cut of $N_D$.  To see this, let us first let $S$ be the result of replacing each atom
$R(a,b)\in \Gamma$ with both possible edges, $R_\ell(a,b),R_r(a,b)$ in $N_D$.  Since $\Gamma$ is a
contingency set, it follows that $S$ is a cut of $N_D$.  Now, let $S'$ be a minimal
subset of $S$ that is still a cut, where some of the extra $R$-edges, i.e., either $R_\ell(a,b)$
or $R_r(a,b)$  have been removed.

By the proof of \Cref{minimal cut lemma}, we know that $S'$ has only one edge for each atom
$R(a,b)\in \Gamma$.  Thus, $\abs{S'} = \abs{\Gamma}$ as claimed.  It follows that the size of a min
cut of $N_D$ is the same as the size of a minimum contingency set for $(D,q)$.
\end{proof}

\subsection{Proofs for \cref{sec: 2perm}}

\begin{figure*}[t]
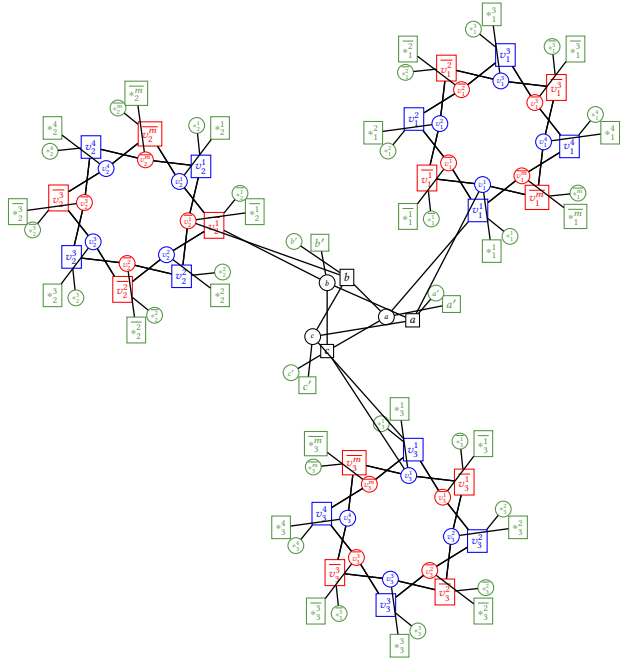

\centering

\include{figs/AxByRxyRyx_gadget}

\caption{
$q_\textup{perm}^{AB}$:  Gadgets for variables $v_1, v_2, v_3$ and clause $C_1 = (v_1 \vee \ov{v_2}
  \vee v_3)$ in proof of  \Cref{perm is npc}.
Circles represent $A$-tuples
and squares $B$-tuples. 
$R$-tuples are the edges between circles and squares.}
\label{Fig_red_AxByRxyRyx}
\end{figure*}

\begin{proof}[Proof of \Cref{perm is npc}]
We define a reduction from 3SAT to $\res(q_\textup{perm}^{AB})$, see \Cref{Fig_red_AxByRxyRyx}.
Similar to the previous cases, we want to create variable gadgets such that a minimum cover will choose either blue nodes 
(variable is set to \true), or red nodes (variable is set to \false), and a clause gadget (black nodes) such that 
if the clause is satisfied, then the minimum cover is 5, otherwise 6.

Let $\psi$ be a 3CNF formula with $n$ variables $v_1, \ldots, v_n$ and $m$ clauses $C_1, \ldots, C_{m}$.  
Our reduction will map  any such $\psi$ to a pair $(D_\psi,k_\psi)$ where $D_\psi$ is a database 
satisfying $\perm$, and 
$$\psi\in 3\sat \quad\Leftrightarrow\quad (D_\psi,k_\psi) \in \res(q_\textup{perm}^{AB})$$

In our construction, if $\psi \in 3\sat$, then the size of each minimum contingency set for 
$q_\textup{perm}^{AB}$ in $D_\psi$ will be $k_\psi=(3n+5)m$, whereas if $\psi \not\in 3\sat$, then the size of all contingency sets
for $q_\textup{perm}^{AB}$ in $D_\psi$ will be greater than $k_\psi$.

\begin{enumerate}

\item Variable gadget: 
For each variable $v_i$ and each $j \in [m]$ insert the following tuples into the database:
$A(v_{i}^{j})$, $B(v_{i}^{j})$, 
$A(\ov{v_{i}^{j}})$, $B(\ov{v_{i}^{j}})$ and
$R(v_{i}^{j}, \ov{v_{i}^{j}})$, $R(\ov{v_{i}^{j}}, v_{i}^{j})$,
$R(v_{i}^{j+1}, \ov{v_{i}^{j}})$, $R(\ov{v_{i}^{j}}, v_{i}^{j+1})$. If $j+1 > m$, then make the superscript 1.

We want to join those tuples such that the minimum contingency sets are to either choose all 
tuples $A(v_i^j), B(v_{i}^{j})$ representing a variable to have assignment \true,
or all tuples $A(\ov{v_i^j}), B(\ov{v_{i}^{j}})$ representing a variable to have assignment \false, plus some $R$-tuples.
To obtain that property, we need the following additional tuples:
$A(*_{i}^{j})$, $B(*_{i}^{j})$, $A(\ov{*_{i}^{j}})$, $B(\ov{*_{i}^{j}})$ and
$R(*_{i}^{j}, v_{i}^{j})$, $R(v_{i}^{j}, *_{i}^{j})$,
$R(\ov{*_{i}^{j}}, \ov{v_{i}^{j}})$, $R(\ov{v_{i}^{j}}, \ov{*_{i}^{j}})$.

With this construction we guarantee that we can ``cover'' the variable gadget by choosing either all positive $A,B$-tuples plus
the $m$ tuples $R(\ov{*_{i}^{j}}, \ov{v_{i}^{j}})$, 
or all negative $A,B$-tuples plus the $m$ tuples $R(*_{i}^{j}, v_{i}^{j})$.
In both cases, we choose $3m$ tuples.

\item Clause gadget:
For each clause $j \in [m]$ insert the following tuples into the database:
$A(a_j)$, $B(a_j)$, $A(b_j)$, $B(b_j)$, $A(c_j)$, $B(c_j)$,
$R(a_{j}, b_{j})$, $R(b_{j}, a_{j})$, $R(b_{j}, c_{j})$, $R(c_{j}, b_{j})$, $R(c_{j}, a_{j})$, $R(a_{j}, c_{j})$ and
$A(a_j')$, $B(a_j')$, $A(b_j')$, $B(b_j')$, $A(c_j')$, $B(c_j')$,
$R(a_{j}, a_{j}')$, $R(a_{j}', a_{j})$, $R(b_{j}, b_{j}')$, $R(b_{j}', b_{j})$, $R(c_{j}, c_{j}')$, $R(c_{j}', c_{j})$ and

For this gadget, we have 3 options to choose only 5 tuples in order to delete all the witnesses. For example:
$A(a_j)$, $B(a_j)$, $A(b_j)$, $B(b_j), R(c_{j}, c_{j}')$.

\item Connecting the gadgets:
For each variable $i$ that appears in clause $j$ at position 1, add the following tuples: 
$R(v_{i}^{j},a_{j}), R(a_{j}, v_{i}^{j})$ if $v_{i}$ appears as positive, and 
$R(\ov{v_{i}^{j}},a_{j}), R(a_{j}, \ov{v_{i}^{j}})$ if it appear as negative.
Analogously use $b_{j}$ or $c_{j}$ instead of $a_{j}$ for positions 2 and 3 instead of position 1.
\end{enumerate}

After connecting the variable gadgets with the clause gadgets, the witnesses are formed such that if a clause cannot be
satisfied, then we need to pick all $A$- and $B$-tuples from the clause gadget (the black triangle), totaling 6 tuples. 
Otherwise, we can  delete all witnesses by picking 5 tuples, namely 2 pairs of $A,B$-tuples and one $R$-tuple.
\end{proof}

\begin{proof}[Proof of \cref{general perm proof}]
There are 2 cases.

\textbf{Case 1:} $q$ is not bound. We can write  $q =  q_\ell(x), G(x,y)$ where $q_\ell(x)$ 
does not contain the variable $y$. $G(x,y)$ includes $R(x,y), R(y,x)$ and may
include exogenous atoms containing the variable  $y$.  Think of $G(x,y)$ as the rightmost group in
\Cref{easyWalkFig}.

For any database, $D \models q$, $\res(D,q)$ is equivalent to the following Network Flow.
As usual, each endogenus atom from the pseudo-linear $q_\ell(x)$ becomes a 1-weight edge and each
exogenus atom is an $\infty$-weight edge.  Whenever $\set{R(c,d),R(d,c)} \subseteq D$, we add $\infty$-weight
edges from the rightmost output of $q_\ell(c)$ and $q_\ell(d)$ to $\set{c,d}$  and a 1-weight edge from
$\set{c,d}$ to the terminal node, $t$.

\textbf{Case 2:} $q$ is bound. We can write  $q =  q_\ell(x), G(x,y), q_r(y)$ where $G(x,y)$
includes $R(x,y), R(y,x)$ and may include an essentially exogenous atom $D(x,y)$ if that occurs in
$q$.  The relevant issues are that removing $G(x,y)$ separates $q_\ell(x)$ from $q_r(y)$ and these
contain at least one endogenous atom each.  

We define a reduction from $\res(q_\textrm{perm}^{AB})$ to $\res(q)$.
We say that variable $z \ \text{isLike}\ x$, if $z$ occurs in $q_\ell(x)$.  Otherwise,  $z$ $\text{isLike}\ y$.

Now consider a database $D$ with $D \models q_\textrm{perm}^{AB}$. We define a new database $D'$ such that for each atom $S_i(v_1, v_2)$ or $A(v)$ 
occurring in $q$, we define 
\begin{align*}
S_i' &= \bigset{ (t(v_1,a,b),t(v_2,a,b))}{D \models  q_\textrm{perm}^{AB}(a,b)}\\
A' &= \bigset{(t(v,a,b))}{D \models  q_\textrm{perm}^{AB}(a,b)}
\end{align*}
where
\[
t(v,a,b) \;\eqdef\;
\begin{cases}
a   & \textrm{if } v \ \text{isLike}\ x  \\
b   & \textrm{if } v \ \text{isLike}\ y \\
\end{cases}
\]

It is clear that the witnesses and minimum contingency sets  of $D \models q_\textrm{perm}^{AB}$ are
exactly preserved in $D'\models q$.
\end{proof}

\subsection{Proof for \cref{sec:2-patterns-REP}}

\begin{proof}[Proof of \cref{easy case REP}]
First consider $q = z_3$. Given a database $D$ such that $D \models z_3$, {witnesses can be of two forms:} 
\begin{align*} 
(a,a,a) &= \set{R(a,a),A(a)}\\
(a,a,b) &= \set{R(a,a), R(a,b), A(b)}
\end{align*}

From that, we can conclude that no tuple $R(a,b)$ with $a\neq b$ needs to be in a contingency set, 
since we can choose either $R(a,a)$ or $A(b)$ instead. Thus, we can construct a network flow
that doesn't include tuples $R(a,b)$ and solve resilience for $z_3$. 
Note that when we consider any expansion of $z_3$ that is pseudo-linear, we always have that 
$R(a,b)$ with $a\neq b$ is not needed in a minimum contingency set. This property together with the assumption that 
query $q$ is pseudo-linear, allows for a construction of a network flow to solve resilience. Therefore, $\res(q)$ is in \p.
\end{proof}

\subsection{Proof for \cref{sec: 2dichotomy}}

\begin{proof}[Proof of \cref{thm: 2atom dichotomy}]
If $q$ has a triad, then $\res(q)$ is \np-complete by \cref{thm: triads in sj}. By \cref{thm: no triad means linear}, we only need to consider 
the cases where $q$ is pseudo-linear. 

In this case, if $q$ has a path (\cref{unary path}, \cref{binary path}), then $\res(q)$ is \np-complete.
Paths cover all the queries where $R$-atoms do not share a variable, including cases with variable repetition. It remains to characterize
the complexity of the queries where $R$-atoms share at least one variable. Note that chain, permutation, and confluence are the only three 
possible patterns for a query with exactly two $R$-atoms and no variable repetition.

If $q$ has a chain , then $\res(q)$ is np-complete (\cref{2chains are hard}). If $q$ has a permutation, then $\res(q)$ is \np-complete when
the permutation is bounded, and it is in \p, when the permutation is unbounded (\cref{general perm proof}). These are the only two possible ways
a permutation can occur. If $q$ has a confluence, then $\res(q)$ is \np-complete when there is an exogenous path, and it is in \p\ otherwise (\cref{2conf hard}).

Now we only have left the case where $q$ has variable repetition and the $R$-atoms share a variable, which implies $\res(q)$ is in \p\ (\cref{easy case REP}).

Since we have exhausted all the cases to consider, we show that there is a dichotomy for the class of ssj binary queries with only two $R$-atoms.
\end{proof}

\subsection{Proofs for \cref{sec: 3chains}}

\begin{proof}[Proof of \cref{3chains are hard}]

We define a reduction from $\res(\chain)$ to $\res(q)$, using a strategy similar to the proof of
 \cref{unary path}.
\end{proof}

\subsection{Proofs for \cref{sec: 3conf}}

\begin{proof}[Proof of \cref{cf3.1 is npc}]
We reduce Max 2-SAT to $\res(q_{3\textrm{conf}}^{AC})$.  Given a 2CNF formula, $\phi$, with $n$ variables and $m$
clauses, and a number $r<m$, we produce a database, $D$, and bound $k$, such that $\phi$ has an
assignment satisfying at least $r$ clauses iff $(D,k) \in \res(q_{3\textrm{conf}}^{AC})$.  The construction is drawn
in \Cref{Fig_cf3}.  A sample variable gadget for variable $x$ is shown. The two minimum contingency sets
consist of $2s$  $x$ nodes, plus 2 helper nodes in the two crossover gadgets or $2s$  $\ov{x}$
nodes, plus 2 helper nodes, corresponding to variable $x$ being true or false, respectively.
The reason for the crossover is so that each variable can be instantiated via diamonds and hexagons
corresponding to the atoms $A,C$, respectively.

The clause gadgets for clauses of size 1 and size 2 are also drawn.  Clauses of size 1 need no nodes
chosen when they are true and one node otherwise.  Clauses of size 2 need 1 node chosen when they
are true and 2 when they are false.  Let $d$ be the number of clauses of size 2 in $\phi$.
Saying that at least $r$ clauses of $\phi$ are true means that at most $m-r$ clauses are false.
Thus, the size of the minimum contingency set is $k = n(2s+2) + d + m -r$.
\end{proof}

\begin{proof}[Proof of \cref{cf3.2 is easy}]
First observe that any contingency set contains only $R$-tuples, since $S,T$ are dominated and therefore exogenous. For any tuple $R(a,b) \in D$, if $S(a,b), T(a,b) \in D$, then $R(a,b)$ must be in all contingency sets, since those 3 tuples form a witness. Let $\Gamma_{TS}$ be the set of all such tuples. We then proceed to create a \flow with tuples $D' = D - \Gamma_{TS}$ and we claim that $\Gamma = \Gamma_{TS} \cup C$ is a min contingency set for $(q_{3\textrm{conf}}^{TS}, D)$, where $C$ is a min cut found by \flow. 

Let $C$ be a min cut and suppose there is a $\Gamma'$ such that $D' -\Gamma' \not\models q_{3\textrm{conf}}^{TS}$ and $|C| > |\Gamma'|$. That implies that there are at least 2 witnesses that can be broken by deleting one tuple but the min cut chose to delete 2 edges. Consider the tuple $R(a,b)$ and these witnesses to be 
\begin{align*}
T(a,b)R(a,b)R(1,b)R(1,2)S(1,2)\\
T(3,b)R(3,b)R(a,b)R(a,4)S(a,4)
\end{align*}

Note that with this set of tuples we also have witness $$T(3,b)R(3,b)R(1,b)R(1,2)S(1,2)$$ which cannot be deleted by deleting $R(a,b)$, contradicting the assumption that it was possible.
\end{proof}

\begin{figure*}
\begin{center}
\begin{tikzpicture}[scale=.45,
		every diamond node/.style={draw=black, inner sep=0pt, minimum size=2cc, line
                  width=1pt,  fill=red!60},
			every ellipse node/.style={draw=black, inner sep=0pt, minimum size=1.5cc, line
                          width=1pt,  fill=blue!40},
			every regular polygon node/.style={draw=black, inner sep=0pt, minimum
                          size=2cc, line width=1pt,  fill=green!40}]
\node[diamond] (x1) at (-6,0) {$\lnot x_1$};
\node[diamond] (x2) at (-6,-3) {$\lnot x_2$};
\node[diamond] (xs) at (-6,-7) {$\lnot x_s$};
\node[ellipse] (x11) at(-2,0) { $\lnot x_1,5$};
\node[ellipse] (x21) at(-2,-3) { $\lnot x_2,5$};
\node[ellipse] (xs1) at(-2,-7) { $\lnot x_s,5$};
\node[regular polygon,regular polygon sides=6] (55) at (10,0) {5};
\node at (-6,-4.75) {$\vdots$};
\node at (-2,-4.75) {$\vdots$};
\node[diamond] (5) at (6,-7) {$5$};
\node[ellipse] (5a) at (6,-3.5) {$5,a$};
\node[ellipse] (ba) at (8,-3.5) {$b,a$};
\node[ellipse] (b5) at (10,-3.5) {$b,5$};
\node[regular polygon,regular polygon sides=6] (y1) at (22,0) {$\ov{y_1}$};
\node[regular polygon,regular polygon sides=6] (y2) at (22,-3) {$\ov{y_2}$};
\node[regular polygon,regular polygon sides=6] (ys) at (22,-7) {$\ov{y_s}$};
\node[ellipse] (y11) at(18,0) { $5,\ov{y_1}$};
\node[ellipse] (y21) at(18,-3) { $5,\ov{y_2}$};
\node[ellipse] (ys1) at(18,-7) { $5,\ov{y_s}$};
\node at (22,-4.75) {$\vdots$};
\node at (18,-4.75) {$\vdots$};
\node at (8,-9)  { Gadget for clause $(\lnot x \lor \lnot y)$};
\foreach \from/\to in
  {x1/x11,x2/x21,xs/xs1,x11/55,x21/55,xs1/55,5/5a,5a/ba,ba/b5,b5/55,
   5/y11,5/y21,5/ys1,y11/y1,y21/y2,ys1/ys}
\draw[->,line width=1pt] (\from) -- (\to);
\end{tikzpicture}

\vspace*{.3in}

\begin{tikzpicture}[scale=.5,
		every diamond node/.style={draw=black, inner sep=0pt, minimum size=2cc, line
                  width=1pt,  fill=red!60},
			every ellipse node/.style={draw=black, inner sep=0pt, minimum size=2cc, line
                          width=1pt,  fill=blue!40},
			every regular polygon node/.style={draw=black, inner sep=0pt, minimum
                          size=2cc, line width=1pt,  fill=green!40}]
\node[diamond] (x1) at (-6,0) {$x_1$};
\node[diamond] (x2) at (-6,-3) {$x_2$};
\node[diamond] (xs) at (-6,-7) {$x_s$};
\node[ellipse] (x11) at(-2,0) { $x_1,1$};
\node[ellipse] (x21) at(-2,-3) { $x_2,1$};
\node[ellipse] (xs1) at(-2,-7) { $x_s,1$};
\node[regular polygon,regular polygon sides=6] (11) at (2,-3) {1};
\node at (-6,-4.75) {$\vdots$};
\node at (-2,-4.75) {$\vdots$};
\node at (-2,-9)  {Gadget for clause $(x)$};
\foreach \from/\to in
  {x1/x11,x2/x21,xs/xs1,x11/11,x21/11,xs1/11}
\draw[->,line width=1pt] (\from) -- (\to);
\end{tikzpicture}
\hspace*{.5in}
\begin{tikzpicture}[scale=.45,
		every diamond node/.style={draw=black, inner sep=0pt, minimum size=2cc, line
                  width=1pt,  fill=red!60},
			every ellipse node/.style={draw=black, inner sep=0pt, minimum size=2cc, line
                          width=1pt,  fill=blue!40},
			every regular polygon node/.style={draw=black, inner sep=0pt, minimum
                          size=2cc, line width=1pt,  fill=green!40}]
\node[diamond] (x1) at (-4,0) {$x_1$};
\node[ellipse] (x1x1b) at(0,0) { $x_1,\ov{x_1}$};
\node[regular polygon,regular polygon sides=6] (x1b) at (4,0) {$\ov{x_1}$};
\node[ellipse] (x2x1b) at(0,-3) { $x_2,\ov{x_1}$};
\node[diamond] (x2) at (-4,-6) {$x_2$};

\node[regular polygon,regular polygon sides=6] (x2b) at (4,-6) {$\ov{x_2}$};
\node[ellipse] (x2x2b) at(0,-6) { $x_2,\ov{x_2}$};
\node[diamond] (xs) at (-4,-12) {$x_{s-1}$};

\node[regular polygon,regular polygon sides=6] (xsb) at (4,-12) {$\ov{x_{s-1}}$};
\node[ellipse] (xsxsb) at(0,-12) { $x_{s-1},\ov{x_{s-1}}$};
\node at (-4,-8.75) {$\vdots$};
\node at (4,-8.75) {$\vdots$};
\node at (0,-8.75) {$\vdots$};
\node (vdots) at (-1,-9.75) { };
\node at (0,-15)  {Top half of Variable Gadget};
\foreach \from/\to in
  {x1/x1x1b,x1x1b/x1b,x2/x2x1b,x2x1b/x1b,x2/x2x2b,x2x2b/x2b,
   xs/xsxsb,xs/vdots,xsxsb/xsb}
\draw[->,line width=1pt] (\from) -- (\to);
\end{tikzpicture}

\vspace*{.3in}

\begin{tikzpicture}[scale=.45,
		every diamond node/.style={draw=black, inner sep=0pt, minimum size=2cc, line
                  width=1pt}, 
			every ellipse node/.style={draw=black, inner sep=0pt, minimum size=2cc, line
                          width=1pt}, 
			every regular polygon node/.style={draw=black, inner sep=0pt, minimum
                          size=2cc, line width=1pt}]       
\node[diamond,fill=red!60] (x1) at (-5.5,0) {$x_s$};
\node[ellipse,fill=blue!40] (e) at(-3,0) { $x_s,e$};
\node[ellipse,fill=blue!40] (x1x1b) at(0,0) { $e,e'$};
\node[ellipse,fill=blue!40] (e') at(3,0) { $e',\ov{x_s}$};
\node[regular polygon,regular polygon sides=6,fill=green!40] (x1b) at (5.5,0) {$\ov{x_s}$};
\node[ellipse,fill=green!40] (x2x1b) at(0,-3) { $x_s,x_{s+1}$};

\node[diamond,fill=blue!40] (b) at (-11,-6) {$b$};
\node[ellipse,fill=blue!40] (b') at (-8.5,-6) {$b,b'$};
\node[ellipse,fill=blue!40] (d') at (8.6,-6) {$d',d$};
\node[regular polygon,regular polygon sides=6,fill=blue!40] (d) at (11,-6) {$d$};

\node[diamond,fill=blue!40] (a) at (-11,-3) {$a$};
\node[ellipse,fill=blue!40] (a') at (-8.5,-3) {$a,a'$};
\node[ellipse,fill=blue!40] (a1) at(-5,-3) { $x_s,a'$};
\node[ellipse,fill=blue!40] (b1) at(-5,-6) { $\ov{x_{s+1}},b'$};
\node[ellipse,fill=blue!40] (1c) at(5,-3) { $c',x_{s+1}$};
\node[ellipse,fill=blue!40] (1d) at(5.5,-6) { $d',\ov{x_{s}}$};
\node[ellipse,fill=blue!40] (c') at (8.5,-3) {$c',c$};
\node[regular polygon,regular polygon sides=6,fill=blue!40] (c) at (11,-3) {$c$};

\node[ellipse,fill=red!60] (x2x2b) at(0,-6) { $\ov{x_{s+1}},\ov{x_s}$};
\node[diamond,fill=green!40] (xm) at (-7,-9) {$\ov{x_{s+1}}$};

\node[regular polygon,regular polygon sides=6,fill=red!60] (xmb) at (7,-9) {${x_{s+1}}$};
\node[ellipse,fill=blue!40] (f) at(-3.7,-9) { $\ov{x_{s+1}},f$};
\node[ellipse,fill=blue!40] (xmxmb) at(-0.1,-9) { $f',f$};
\node[ellipse,fill=blue!40] (f') at(3.5,-9) { $f',x_{s+1}$};
\node at (0,-12)  {Middle Crossover Part of Variable Gadget};
\foreach \from/\to in
  {x1/e,e/x1x1b,x1x1b/e',e'/x1b,x2x1b/xmb,x1/x2x1b,
   xm/f,f/xmxmb,xmxmb/f',f'/xmb,xm/x2x2b,x2x2b/x1b,
  a/a',a'/a1,a1/x2x1b,x2x1b/1c,1c/c',c'/c,b/b',b'/b1,b1/x2x2b,x2x2b/1d,1d/d',d'/d}
\draw[->,line width=1pt] (\from) -- (\to);
\end{tikzpicture}
\end{center}

\caption{Reduction Gadgets for proof of \Cref{cf3.1 is npc}: diamonds represent $A$, ellipses, $R$, and
  hexagons, $C$. In the variable gadgets, the minimum contingency sets choose all red vertices and
  no green, or all green vertices and no red.}
\label{Fig_cf3}
\end{figure*}

\subsection{Proofs for \cref{sec: 3cc}}

\begin{proof}[Proof of \cref{bounded 3cc}]
Reduction from $\res(\chain)$.
\end{proof}

\begin{proof}[Proof of \cref{unbounded 3cc}]
Reduction from Max 2SAT, similar to the one used for $q_{\textrm{3conf}}^{AC}$.
\end{proof}

\subsection{Proofs for \cref{sec: 3permR}}

\begin{proof}[Proof of \cref{KS2}]
This is similar to \cref{AR perm}.  The difference is that while $A(a)$ ``dominates'' the
1-way tuple $R(a,b)$ in $q_{\textrm{3perm-R}}^{A}$, it is not the case that $S(e_1,a)$ would dominate $R(a,b)$ because
there might be many $e_i$'s such that $S(e_i,a)\in D$, in which case it might be advantageous to
choose one $R(a,b)$ instead of many $S(e_i,a)$'s.

We thus modify the flow graph to include all the $S(e,a)$ edges at cost 1 each on the left, all the
$\set{a,b}$ pairs at cost 1 each on the right.  We include $\infty$-weight edges from any $S(e,a)$
to $\set{a,b}$ plus cost 1 edges from $S(e,a)$ to $\set{b,c}$ for any 1-way edges $R(a,b)$.

Let $M$ be a min-cost flow and form $\Gamma$ by including all the $S(e,a)$'s and 1-way $R(a,b)$'s
from $M$ together with one of $R(a,b)$ or $R(b,a)$ whenever $\set{a,b}\in M$.  Similar to
\cref{AR perm}, the rule for which to choose is that if some $S(e,a) \in (D-M)$ but no
$S(f,b) \in (D-M)$, then add $R(a,b)$ to $\Gamma$.  Symmetrically, if $S(e,b) \in (D-M)$ but no
$S(f,a) \in (D-M)$, then add $R(b,a)$ to $\Gamma$; otherwise, arbitrarily add one or the other.

The same argument as in \cref{AR perm} shows that the resulting $\Gamma$ is a minimum
contingency set.
\end{proof}

\begin{proof} [Proof of \cref{KS0}] 
We reduce 3SAT to $\res(q_{\textrm{3perm-R}}^{S_{xy}})$.   
The idea for the variable gadgets is that for a database that
contains the tuples  $T_{x_i} = \{S(x_i,\ov{x_i})$, $R(x_i,\ov{x_i})$, $S(\ov{x_i},{x_i})$,
$R(\ov{x_i},{x_i})\}$, 
we must choose exactly one $R(x_i,\ov{x_i})$ or $R(\ov{x_i},{x_i})$, the first of  which will
correspond to the assignment $x$ to 1, and the second of which, to 0.  In full detail, the $x$
gadget consists of a chain of these choices, i.e., the union of $T_{x_i}$, $i = 1 \ldots, m$, together
with all the tuples $R(x_i,x_{i+1})$, $R(x_{i+1},x_{i})$, $R(\ov{x_i},\ov{x_{i+1}})$,
$R(\ov{x_{i+1}},\ov{x_{i}})$.  
For a minimum contingency over this gadget we may choose 
all of the $R(x_i,x_{i+1})$ and $R(x_i,\ov{x_i})$ edges (corresponding to $x$ gets 1), 
or all the $R(\ov{x_i},\ov{x_{i+1}})$ and $R(\ov{x_i},{x_i})$ edges (corresponding to $x$ gets 0).  

The clause gadget is similar.  If $C_i$ is $(x \lor \ov{y} \lor z)$,
then the clause can eliminate two, but not all three pointers to the edges $\set{x_i, x_{i+1}}$,
$\set{\ov{y_i}, \ov{y_{i+1}}}$, $\set{z_i, z_{i+1}}$ after removing 8 tuples.
To simplify the explanation, let $P(a,b) = \set{R(a,b),R(b,a)}$ and $F(a,b) = P(a,b)\cup \set{S(a,b),S(b,a)}$ for elements $a,b\in D$.
The $C_i$ clause gadget contains the union of the following sets of tuples:
$F(a_i,b_i)$, $F(b_i,c_i)$, $F(c_i,a_i)$, $F(a_i,x_i)$, $F(b_i,\ov{y_i})$, $F(c_i,z_i)$, $P(a_i,a_i')$, $P(b_i,b_i')$, $P(c_i,c_i')$.
The idea is that for each full pair, $F(e,f)$, exactly one of $R(e,f)$ or $R(f,e)$ must be chosen in
the minimum contingency set $\Gamma$.
$C_i$ is designed so that a contingency set of size 8 exists
iff at least one pair from $P(x_i,x_{i+1})$, $P(\ov{y_i},\ov{y_{i+1}})$, $P(z_i,z_{i+1})$ has been
previously chosen, i.e., iff the clause $C_i$ is true.
\end{proof}

\begin{proof}[Proof of \cref{KS1.5 et al}]
We reduce $\res(q_\textrm{perm}^{AB})$ to $\res(q_\textrm{3perm-R}^{AC})$.  Given a database $D\models q_\textrm{perm}^{AB}$, construct
$D'\models q_\textrm{3perm-R}^{AC}$ as 
\begin{align*}
A' &:= \bigset{a'}{A(a) \in D}\\
R' &:= R \cup \bigset{(a',a)}{A(a) \in D}.
\end{align*}
It then follows, that it is always at least as good to put $A(a')$ into $\Gamma$, rather than
$R(a',a)$.  Thus, the minimum contingency sets for $(D',q_\textrm{3perm-R}^{AC})$ correspond exactly to the minimum
contingency sets for $(D,q_\textrm{perm}^{AB})$.

For $\res(q_\textrm{3perm-R}^{AB})$, Even though $q_\textrm{perm}^{AB} \rightarrow q_\textrm{3perm-R}^{AB}$, there is no obvious reduction between $\res(q_\textrm{perm}^{AB})$
 and $\res(q_\textrm{3perm-R}^{AB})$.  However, the same reduction from 3SAT to $\res(q_\textrm{perm}^{AB})$ in \cref{perm is npc} 
 also works for $\res(q_\textrm{3perm-R}^{AB})$.

For $\res(q_\textrm{3perm-R}^{S_{xy}BC})$, we can define a reduction from $\res(q_\textrm{perm}^{AB})$.
\end{proof}

\subsection{Proofs for \cref{sec: rep 3 Ratoms}}

\begin{proof}[Proof of \cref{rep 3R}]
For $\res(z_4)$, a reduction from $\res(\vc)$ is enough. Note that tuples $R(a,b)$ with $a\neq b$ do not need to be in a contingency set.

For $\res(z_5)$, a reduction from Max 2SAT, similar to the one used in \cref{cf3.1 is npc}, can be used to show \np-hardness.
\end{proof}

\section{Relevant proofs from sj-free case}\label{sec: sj-free proofs}

\begin{proof}[Proof of \Cref{fact: domination does not change complexity}]
{Let $\Gamma$ be a minimum contingency set of $q$ in $D$.  
Suppose that atom $A$ dominates atom $B$ but there is some tuple $B(\vec t)\in \Gamma$. 
Let $\vec p$ be the projection of $\vec t$ onto $\var(A)$. 
Then we
can replace $B(\vec t)$ by $A(\vec p)$ 
and we remove at least as many witnesses that $D\models q$.  It
follows, as desired, that the complexity of $\res(q)$ is unchanged if $B$ is exogenous, i.e.,  $\res(q) \equiv \res(q')$.}
\end{proof}

\begin{proposition}[Triangle $q_\triangle$ is hard]\label{thm: hardness of triangle}
$\res(q_\triangle)$ is \NP-complete.
\end{proposition}

\begin{proof}[Proof of \cref{thm: hardness of triangle}]
We reduce 3SAT to $\res(q_\triangle)$.  It will then follow that $\res(q_\triangle)$ is NP complete.
Let $\psi$ be a 3CNF formula with $n$ variables $v_1, \ldots, v_n$ and $m$ clauses 
$C_0, \ldots, C_{m-1}$.  
Our reduction will map  any such $\psi$ to a pair $(D_\psi,k_\psi)$ where $D_\psi$ is a database 
satisfying
$q_\triangle$, and 
\begin{equation}\label{hard rats reduction}
\psi\in 3\sat \qLra (D_\psi,k_\psi) \in \res(q)
\end{equation}
In our construction, if $\psi \in 3\sat$, then the size of each minimum contingency set for $q_\triangle$ in
$D_\psi$ will be $k_\psi=6mn$, whereas if $\psi \not\in 3\sat$, then the size of all contingency sets
for $q_\triangle$ in $D_\psi$ will be greater than $k_\psi$.

Note $D_\psi \models q_\triangle$ iff it contains three pairs $R(a,b)$, $S(b,c)$, $T(c,a)$.
We visualize $R(a,b)$ as a red edge, $S(b,c)$ as a green edge and $T(c,a)$ as a
blue edge.
Thus each witness $(a,b,c)$ that $D_\psi\models q_\triangle$ is an RGB triangle.
(Notice that the edge direction $a \rightarrow b$ drawn in \Cref{fig:gadgetgi} corresponds to the
variable order in $R$, and analogously for $S$ and $T$.)
The job of a contingency set for
$q_\triangle$ is to remove all RGB triangles.

$D_\psi$ contains one circular gadget $G_i$ for each variable $v_i$.  
The circle consists of $12m$ solid edges, half of
them marked $v_i$ and the other half marked $\ov{v_i}$ (see \autoref{fig:gi segment} and
\autoref{fig:gadget}).  Note that there are $12m$ RGB triangles and they can be minimally broken by
choosing the $6m$ $v_i$ edges or the $6m$ $\ov{v_i}$ edges. Any other way would require more
edges removed.
Thus, each minimum contingency set for $D_\psi$ corresponds to a truth assignment to the
variables of $\psi$. And there will be
a minimum contingency set of size $k_\psi = 6mn$ iff $\psi \in 3\sat$.

\begin{figure}
\begin{subfigure}[b]{\linewidth}
\begin{center}
\begin{tikzpicture}[ scale=.35]
\draw (2,2) circle [radius=1.0] node {{\color{black} $a^i_1$}};
\draw (6,2) circle [radius=1.0] node {{\color{black} $b^i_1$}};
\draw (10,2) circle [radius=1.0] node {{\color{black} $c^i_1$}};
\draw (14,2) circle [radius=1.0] node {{\color{black} $a^i_2$}};
\draw (18,2) circle [radius=1.0] node {{\color{black} $b^i_2$}};
\draw (22,2) circle [radius=1.0] node {{\color{black} $c^i_2$}};
 \draw[->,line width=2pt,color=red] (3,2) -- (5,2)  node[pos=.5,above]{$v_i$}; 
 \draw[->,line width=2pt,color=dg] (7,2) -- (9,2)  node[pos=.5,above]{$\ov{v_i}$}; 
 \draw[->,line width=2pt,dotted,color=blue] (10,1) .. controls(6,-1) .. (2,1)  node[pos=.5,below]{ }; 
 \draw[->,line width=2pt,dotted,color=red] (14,1) .. controls(10,-1) .. (6,1)  node[pos=.5,below]{ }; 
 \draw[->,line width=2pt,dotted,color=dg] (18,1) .. controls(14,-1) .. (10,1)  node[pos=.5,below]{ }; 
 \draw[->,line width=2pt,dotted,color=blue] (22,1) .. controls(18,-1) .. (14,1)  node[pos=.5,below]{ }; 
\draw (6,-.6) node {{\color{black}  \includegraphics[scale=0.025]{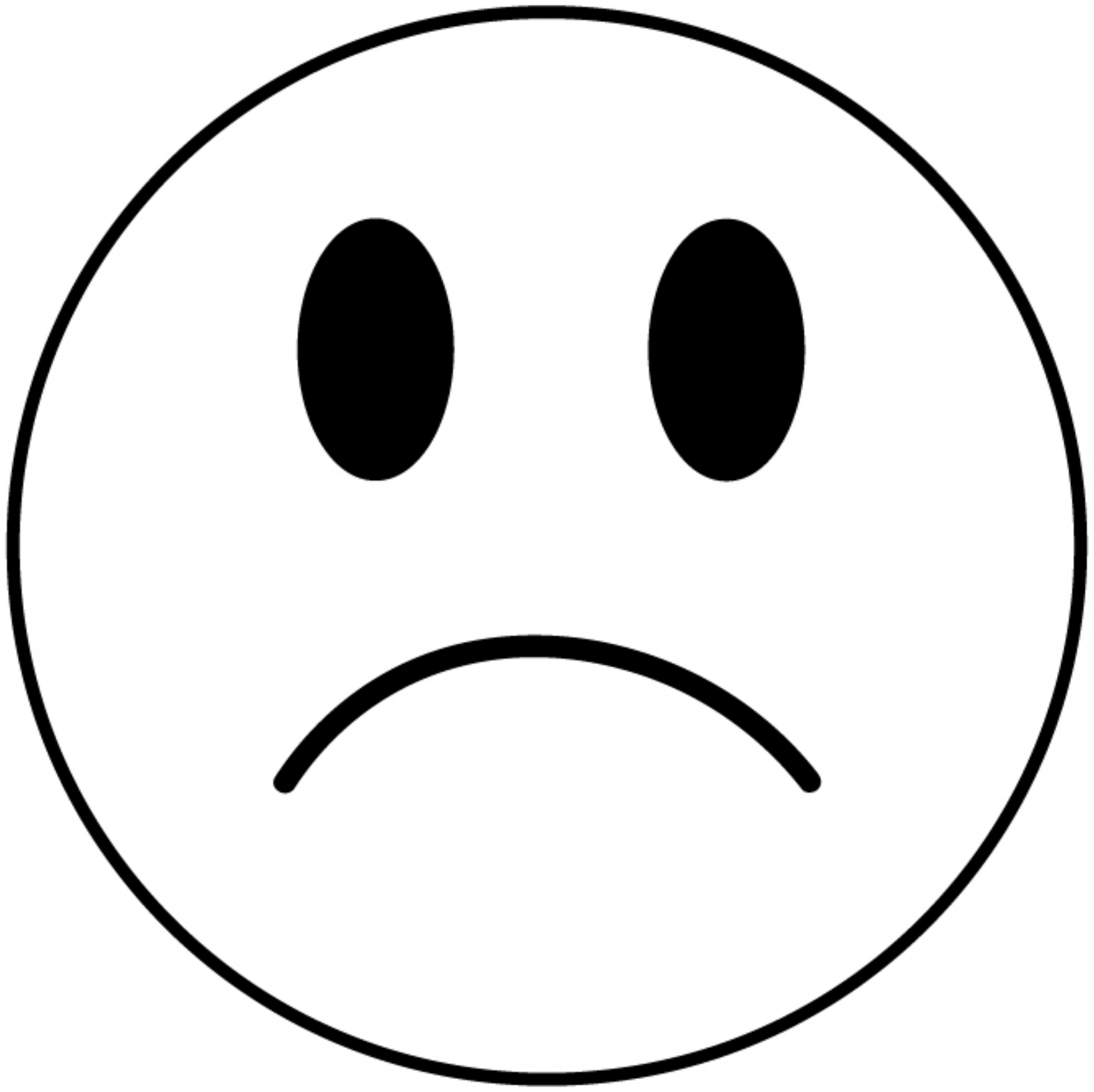} }};
\draw (10,-.6) node {{\color{black}  \includegraphics[scale=0.025]{frown} }};
\draw (14,-.6) node {{\color{black}  \includegraphics[scale=0.025]{frown} }};
\draw (18,-.6) node {{\color{black}  \includegraphics[scale=0.025]{frown} }};
 \draw[->,line width=2pt,color=blue] (11,2) -- (13,2)  node[pos=.5,above]{${v_i}$}; 
 \draw[->,line width=2pt,color=red] (15,2) -- (17,2)  node[pos=.5,above]{$\ov{v_i}$}; a2 b2
 \draw[->,line width=2pt,color=dg] (19,2) -- (21,2)  node[pos=.5,above]{${v_i}$}; 
 \draw[->,line width=2pt,color=blue] (23,2) -- (25,2)  node[pos=.5,above]{$\ov{v_i}$}; 
\end{tikzpicture}
\end{center}
\caption{A six-node segment of the gadget $G_i$. A minimum contingency set chooses either all the solid lines marked $v_i$, or all the solid lines marked $\ov{v_i}$.  The dotted lines are sad because each of them is only part of one single RGB triangle, thus they are never chosen.  
}
\label{fig:gi segment}
\end{subfigure}

\vspace*{.1in}
\begin{subfigure}[b]{\linewidth}
\begin{tikzpicture}[ scale=.2]
\draw (2,2) circle [radius=1.0] node {{\color{red} }};
\draw (6,2) circle [radius=1.0] node {{\color{blue}}};
\draw (10,2) circle [radius=1.0] node {{\color{dg} }};
\draw (14,2) circle [radius=1.0] node {{\color{red}}};
\draw (18,2) circle [radius=1.0] node {{\color{blue}}};
\draw (22,2) circle [radius=1.0] node {{\color{dg} }};
 \draw[->,line width=2pt,color=red] (3,2) -- (5,2)  node[pos=.5,above]{$v_i$}; 
 \draw[->,line width=2pt,color=dg] (7,2) -- (9,2)  node[pos=.5,above]{$\ov{v_i}$}; 
 \draw[->,line width=2pt,color=blue] (11,2) -- (13,2)  node[pos=.5,above]{${v_i}$}; 
 \draw[->,line width=2pt,color=red] (15,2) -- (17,2)  node[pos=.5,above]{$\ov{v_i}$}; a2 b2
 \draw[->,line width=2pt,color=dg] (19,2) -- (21,2)  node[pos=.5,above]{${v_i}$}; 
 \draw[->,line width=2pt,color=blue] (22.5,1.5) -- (24,0)  node[pos=.8,above]{$\ov{v_i}$}; 
 \draw[->,line width=2pt,dotted,color=blue] (10,1) .. controls(6,-1) .. (2,1)  node[pos=.5,below]{ }; 
 \draw[->,line width=2pt,dotted,color=red] (14,1) .. controls(10,-1) .. (6,1)  node[pos=.5,below]{ }; 
 \draw[->,line width=2pt,dotted,color=dg] (18,1) .. controls(14,-1) .. (10,1)  node[pos=.5,below]{ }; 
 \draw[->,line width=2pt,dotted,color=blue] (22,1) .. controls(18,-1) .. (14,1)  node[pos=.5,below]{ }; 
\draw (24.5,-.5) circle [radius=1.0] node {{\color{red} }};
\draw (27,-3) circle [radius=1.0] node {{\color{blue}}};
\draw (29.5,-5.5) circle [radius=1.0] node {{\color{dg} }};
\draw (32,-8) circle [radius=1.0] node {{\color{red}}};
\draw (34.5,-10.5) circle [radius=1.0] node {{\color{blue}}};
\draw (37,-13) circle [radius=1.0] node {{\color{dg} }};
 \draw[->,line width=2pt,color=red] (25,-1) -- (26.5,-2.5)  node[pos=.8,above]{$v_i$}; 
 \draw[->,line width=2pt,color=dg] (27.5,-3.5) -- (29,-5)  node[pos=.8,above]{$\ov{v_i}$}; 
 \draw[->,line width=2pt,color=blue] (30,-6) -- (31.5,-7.5)  node[pos=.8,above]{${v_i}$}; 
 \draw[->,line width=2pt,color=red] (32.5,-8.5) -- (34,-10)  node[pos=.8,above]{$\ov{v_i}$}; a2 b2
 \draw[->,line width=2pt,color=dg] (35,-11) -- (36.5,-12.5)  node[pos=.8,above]{${v_i}$}; 
 \draw[->,line width=2pt,color=blue] (37,-14) -- (37,-17)  node[pos=.5,right]{$\ov{v_i}$}; 
\draw (37,-18.5) node {{\color{black} $\vdots$ }};
\draw[->,line width=2pt,dotted,color=red] (23.5,-.5) .. controls(20.75,-1) .. (18,1)  ; 
\draw[->,line width=2pt,dotted,color=dg] (26,-3) .. controls(24,-3) .. (22,1)  ; 
\draw[->,line width=2pt,dotted,color=blue] (29,-6) .. controls(25,-5) .. (24,-1)  ; 
\draw[->,line width=2pt,dotted,color=red] (31.5,-8.5) .. controls(27.5,-7.5) .. (26.5,-3.5)  ; 
\draw[->,line width=2pt,dotted,color=dg] (34,-11) .. controls(30,-10) .. (29,-6)  ; 
\draw[->,line width=2pt,dotted,color=blue] (36.5,-13.5) .. controls(32.5,-12.5) .. (31.5,-8.5)  ; 
\draw[->,line width=2pt,color=blue] (0,0) -- (1.5,1.5)  node[pos=.8,left]{$\ov{v_i}$}; 
\draw (-1.2,-.8) node {{\color{black} ${\mathstrut^{.^{.^{.^{}}}}}$ }};
\draw (13,-3.5) node {{\color{black}  \includegraphics[scale=0.05]{frown} }};
\draw (13,-6.75) node {{\color{black}  1 }};
\draw (25.5,-9) node {{\color{black} \includegraphics[scale=0.05]{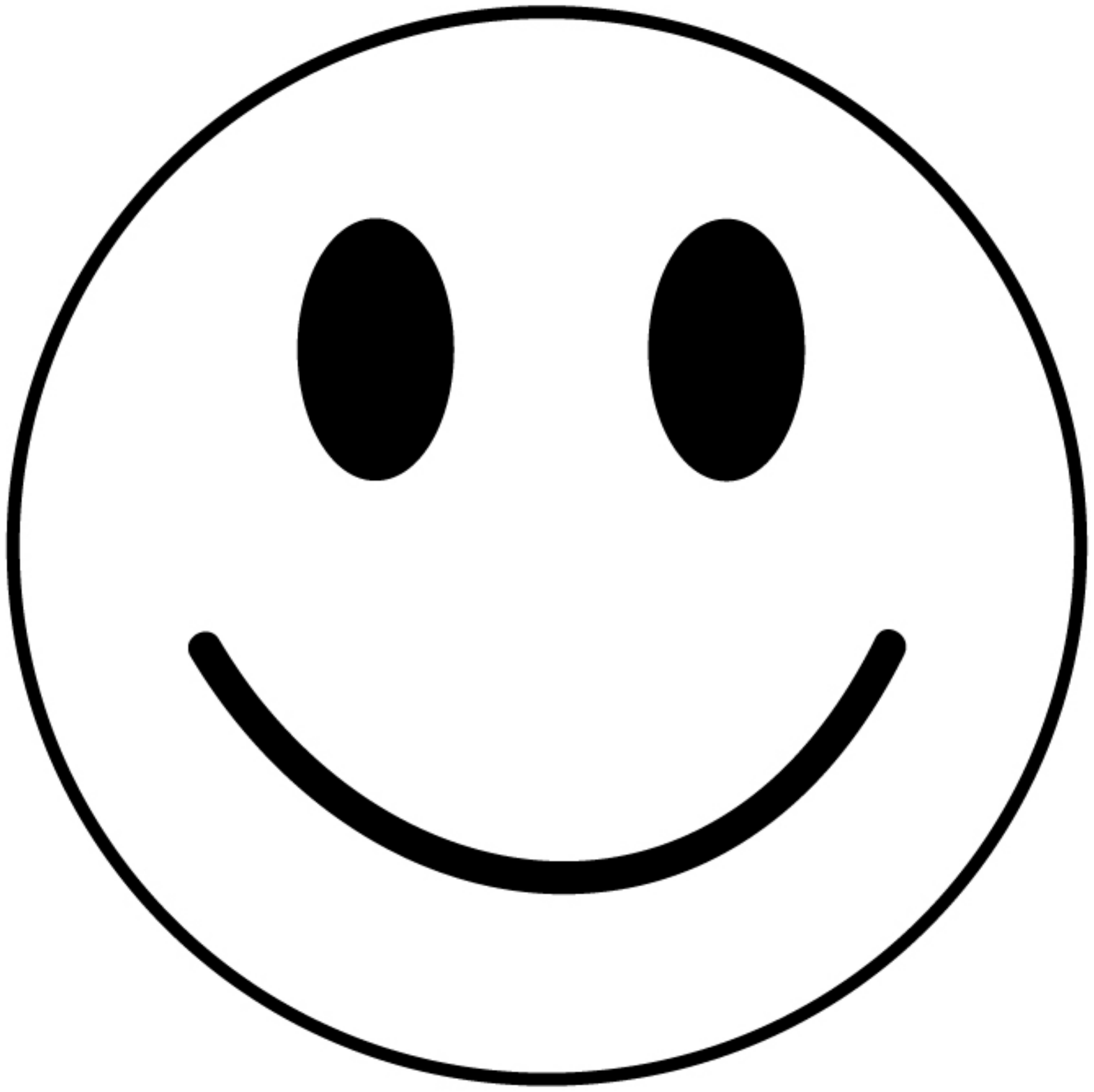} }};
\draw (25.5,-12.25) node {{\color{black}  2 }};
\draw (1,-6) node {{\color{black} \includegraphics[scale=0.05]{smile} }};
\draw (1,-9.25) node {{\color{black} $2m$ }};
\end{tikzpicture}
\caption{Gadget $G_i$ is a cycle containing $2m$ six-node segments with
$12m$ RGB triangles.  They can all be
  eliminated by removing the $6m$ edges marked $v_i$ or the $6m$ edges marked $\ov{v_i}$.  
The odd segments are sad because they are never used for connecting different gadgets (corresponding to clauses that use several variables);
they only separate the even ones, thus preventing spurious triangles.
}\label{fig:gadget}
\end{subfigure}

\vspace*{.1in}

\begin{subfigure}[b]{\linewidth}
 \begin{center}
 \begin{tikzpicture}[ scale=.9]
\fill[black!25] (16,3.5) ellipse (3.0 and 0.5);
\draw[rotate=58,fill=black!25] (9.75,-16.1) ellipse (2.5 and 0.6);
\draw[rotate=-58,fill=black!25] (7.3,11.25) ellipse (2.5 and 0.6);
\draw (16,3.5) node {{\color{black}  $G_1$ }};
\draw (18.75,-.25) node {{\color{black}  $G_2$ }};
\draw (13.5,-.5) node {{\color{black}  $G_3$ }};
 \draw (14,2) circle [radius=1.0] node {{\color{black} $ $ }};
 \draw (13.7,1.7) node {{\color{black}$a^3_{4j+2}$ }};
 \draw (14.4,2.4) node {{\color{black}$a^1_{4j+1}$ }};
 \draw (18,2) circle [radius=1.0] node {{\color{black} }};
 \draw (17.7,2.3)  node {{\color{black} $b^1_{4j+1}$}};
 \draw (18.3,1.6) node {{\color{black} $b^2_{4j+1}$}};
 \draw (16,-.67) circle [radius=1.0] node {{\color{black}$c^3_{4j+1}\hspace*{.07in}c^2_{4j+1}$}};
  \draw[<-,line width=2pt,color=red] (11.5,2) -- (13,2) ; 
  \draw[->,line width=2pt,color=red] (15,2) -- (17,2)  node[pos=.5,above]{${v_1}$}; a to  b
  \draw[<-,line width=2pt,color=red] (19,2) -- (20.5,2)  ; 
  \draw[->,line width=2pt,color=dg] (17.5,1.1) -- (16.7,0) node[pos=.6,right]{$\;\ov{v_2}$} ; 
  \draw[<-,line width=2pt,color=dg] (15.5,-1.57) -- (14.7,-2.67) ; 
  \draw[<-,line width=2pt,color=dg] (19.5,3.77) -- (18.7,2.67) ; 
  \draw[<-,line width=2pt,color=blue] (14.5,1.1) -- (15.3,0) node[pos=.4,right]{$\;{v_3}$} ; 
  \draw[->,line width=2pt,color=blue] (12.5,3.77) -- (13.3,2.67) ; 
  \draw[->,line width=2pt,color=blue] (16.5,-1.57) -- (17.3,-2.67) ; 
 \end{tikzpicture}
 \end{center}
\caption{For clause $C_j=(v_1 \lor \ov{v_2} \lor v_3)$,  we identify vertices 
$b^1_{4j+1}\in G_1$ with $b^2_{4j+1}\in G_2$; 
$c^2_{4j+1}\in G_2$ with $c^3_{4j+1}\in G_3$ and 
$a^3_{4j+2}\in G_3$ with  $a^1_{4j+1}\in G_1$. 
This RGB triangle will be deleted iff the chosen variable
  assignment satisfies $C_j.$
}
\label{fig:gadgetgi identification}
\end{subfigure}

\caption{Gadget construction for hardness proof for $q_\triangle$.
}
\label{fig:gadgetgi}
\end{figure}

We complete the construction of $D_\psi$ by adding one RGB triangle
for each clause $C_j$.  For example, suppose
$C_j = v_1 \lor \ov{v_2} \lor v_3$.
The RGB triangle we add consists of a red edge marked $v_1$, a green edge marked $\ov{v_2}$ and
a blue edge marked $v_3$ (see \autoref{fig:gadgetgi identification}).  Note that if the chosen assignment satisfies $C_j$,  then all $v_1$ edges
are removed, or all $\ov{v_2}$ edges are removed, or all $v_3$ edges are removed.  Thus the $C_j$
triangle is automatically removed.  

How do we create $C_j$'s RGB triangle?  Remember that we have chosen $G_i$ to contain 2 segments for
each clause.  We use the $j$th  odd-numbered segment of $G_i$ to produce the $v_i$ or
$\ov{v_i}$ used in the clause-$j$ triangle.  The even numbered segments are not used:  they serve as
buffers to prevent spurious RGB triangles from being created (In \autoref{fig:gadget} we 
mark these even segments with frowns:  they are sad because they are never used).

More precisely, the  red
$v_1$-edge from $G_1$ is $(a^1_{4j+1},b^1_{4j+1})$, 
the green $\ov{v_2}$-edge from $G_2$ is $(b^2_{4j+1},c^2_{4j+1})$,
and the blue $v_3$-edge from $G_3$ is $(c^3_{4j+1},a^3_{4j+2})$ 
(see \autoref{fig:gadgetgi identification}).  

Now to make this an RGB triangle in $D_\psi$, we identify the two $a$-vertices, the two $b$ vertices and the
two $c$ vertices.  In other words, $G_1$'s $a$-vertex $a^1_{4j+1}$ is equal to $G_3$'s $a$-vertex
$a^3_{4j}$, i.e., they are the same element of the domain of $D_\psi$.
We have thus constructed $C_j$'s RGB triangle  (see \autoref{fig:gadgetgi identification}).

The key idea is that these
identifications can only create this single new RGB triangle because there is no other way to get back to
$G_1$ from $G_2$ in two steps.  All other identifications involve
different segments and so are at least six steps away. Recall that this is the reason why the odd-numbered
segments in the $G_i$'s are not used: this ensures that no additional RGB triangles
are created.

Thus, as desired, \autoref{hard rats reduction} holds and we have reduced $3\sat$ to $\res(q_\triangle)$.
\end{proof}

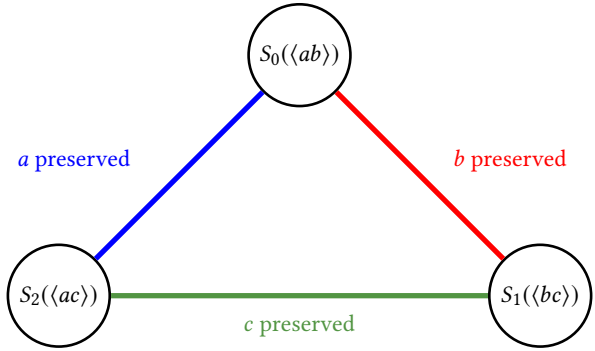
\begin{figure}
\begin{center}
\begin{tikzpicture}[ scale=.4]
\node[circle,draw=black, inner sep=2pt, minimum size=3cc, line width=1pt] (S0) at(0,0)  { $S_0(\angle{ab})$};
\node[circle,draw=black, inner sep=2pt, minimum size=3cc, line width=1pt] (S1) at(8,-8)  { $S_1(\angle{bc})$};
\node[circle,draw=black, inner sep=2pt, minimum size=3cc, line width=1pt] (S2) at(-8,-8)  { $S_2(\angle{ac})$};
\draw (7,-3.5) node {{\color{red}  $b$ preserved}};
\draw (-7.5,-3.5) node {{\color{blue}  $a$ preserved}};
\draw (0,-9) node {{\color{dg}  $c$ preserved}};
\path[line width=2pt,color=red] (S0) edge (S1);
\draw[line width=2pt,color=dg] (S1) -- (S2);
\draw[line width=2pt,color=blue] (S2) -- (S0);
\end{tikzpicture}
\end{center}
\caption{Reduction from $\res(q_\triangle)$ to $\res(q)$ when $q$ contains a triad, $\set{S_0,S_1,S_2}$.}
\label{hard part fig}
\end{figure}

\begin{proposition}[Tripod $q_\Tri$ is hard]\label{prop:tripodQuery}
$\res(q_\Tri)$ is \NP-complete.
\end{proposition}

\begin{proof}[Proof of \cref{prop:tripodQuery}]
We reduce $\res(q_\triangle)$ to $\res(q_\Tri)$.  It will then follow that $\res(q_\Tri)$ is \NP-complete.
Let $(D,k)$ be an instance of $\res(q_\triangle)$.  We construct an instance $(D',k)$ of $\res(q_\Tri)$
by constructing relations $A,B,C$ as copies of $R,S,T$ from $D$.  Define $D'=(A,B,C,W)$ as follows:
\begin{align*}
A &= \bigset{\angle{ab}}{R(a,b)\in D}\\
B &= \bigset{\angle{bc}}{S(b,c)\in D}\\
C &= \bigset{\angle{ca}}{T(c.a)\in D}\\
W &= \bigset{(\angle{ab},\angle{bc},\angle{ac})}{a,b,c\in \textrm{dom(D)}}
\end{align*}

\noindent Here, $\angle{ab}$ stands for a new unique domain value resulting from the concatenation
of domain values $a$ and $b$.  Observe that there is a 1:1 correspondence between the witnesses of $D\models
q_\triangle$ and the witnesses of $D'\models q_\Tri$.
Thus, every contingency set for $q_\triangle$ in $D$ corresponds to a contingency set of the same
size for $q_\Tri$ in $D'$.
Furthermore no minimum $\Gamma'$ from $D'$ needs to choose tuples from~$W$.  If
$\vec t = W(\angle{ab},\angle{bc},\angle{ac})$ were in $\Gamma'$, then we could replace it by
$A(\angle{ab})$, which suffices to remove all the witnesses 
removed by~$\vec t$. As we will explain
later, $A$ ``dominates''
$W$ (\autoref{sj-free domination}).
It follows that $(D,k) \in \res(q_\triangle) \Leftrightarrow (D',k)\in \res(q_\Tri)$.
\end{proof}

\begin{proof}[Proof of \cref{hard part dichotomy}]
Let $q$ be a query with triad ${\mathcal T}=\set{S_0,S_1,S_2}$.  
We build a reduction from $\res(q_\triangle)$ to $\res(q)$. 
Given any $D$ that satisfies $q_\triangle$ we will produce a database $D'$ that satisfies $q$ such
that for all $k$:
\begin{equation}\label{eq res hard case1}
(D,k)\in \res(q_\triangle) \quad\Leftrightarrow\quad (D',k)\in\res(q)
\end{equation}

\noindent
We will assume that no variable is shared by all three elements of ${\mathcal T}$  
(we can ignore any
such variable by setting it to a constant). 
Our proof splits into two cases:

\underline{Case
 1}:  $\var(S_0), \var(S_1), \var(S_2)$ 
are pairwise disjoint. 
Our reduction is similar to the reduction from $q_\triangle$ to $q_\Tri$ (\autoref{prop:tripodQuery}).

We first define the triad relations in $D'$:
\begin{align*}
S_0 &= \bigset{(\angle{ab}, \ldots, \angle{ab})}{R(a,b) \in D} \\
S_1 &= \bigset{(\angle{bc}, \ldots, \angle{bc})}{S(b,c) \in D} \\
S_2 &= \bigset{(\angle{ca}, \ldots, \angle{ca})}{T(c,a) \in D}. 
\end{align*}

\noindent
Thus, each tuple of, for example, $S_0$ consists of identical entries with value $\angle{ab}$ for each
pair $R(a,b) \in D$.  Thus, $S_0,S_1,S_2$ mirror $R,S,T$, respectively. 

To define all the other atoms $A_i$ of $D'$, we first
partition the variables of $q$ into 4 disjoint sets: $\var(q)= \var(S_0)\cup\var(S_1)\cup \var(S_2) \cup
V_3$.  Now for each atom $A_i$, arrange its variables in these four groups. Then define the atom $A'_i$
of $D'$ as follows:
\begin{equation}\label{case1 eq}
A_i' = \bigset{(\angle{ab};\angle{bc};\angle{ca};\angle{abc})}{D \models q_\triangle(a,b,c)}
\end{equation}
\noindent For example, all the variables $v\in \var(S_0)$ are assigned the value $\angle{ab}$
and all the variables $v\in V_3$ are assigned $\angle{abc}$.

By the definition of triad, there is a path from $S_0$ to $S_1$ not using any edges 
(variables) from
$\var(S_2)$.  
Thus, any witness that  $D'\models q$ 
which includes occurrences of $\angle{ab}$ and $\angle{b'c'}$ must have $b =
b'$.  

Similarly, a path from $S_1$ to $S_2$ guarantees that $c$ is preserved and a path from $S_2$ to
$S_0$ guarantees that $a$ is preserved.  It follows that the witnesses 
that $D' \models q$ are essentially identical to the
witnesses that $D\models q_\triangle(x,y,z)$ (See \cref{hard part fig}).

Furthermore, any minimum contingency set only needs tuples from
$S_0, S_1$ or $S_2$. For example, if a tuple contains $\angle{ab}$ or $\angle{abc}$, then it can be
replaced by a tuple from $S_0$.   Thus the sizes of minimum contingency sets are preserved, i.e., \autoref{eq res
  hard case1} holds, as desired.  Thus $\res(q)$ is \NP-complete.

\underline{Case 2}:  $\var(S_i) \cap \var(S_j) \ne \emptyset $  for some $i\ne j$:
We generalize the construction from Case 1 as follows.  Partition $\var(S_i)$ into those 
unshared, those shared with $S_{i-1}$, and those shared with $S_{i+1}$ (Addition is mod 3).

We then assign the
relations of the triad as follows:
\begin{align*}
S_0 &= \bigset{(\angle{ab}; a; b)}{R(a,b) \in D} \\
S_1 &= \bigset{(\angle{bc}; b; c)}{S(b,c) \in D} \\
S_2 &= \bigset{(\angle{ca}; c; a)}{T(c,a) \in D},
\end{align*}

Since none of the $S_i$'s is dominated, in each case both possible values occur, e.g., $a$ and $b$
both occur in the tuples of $S_0$
Thus as in Case 1, $S_0,S_1,S_2$ capture $R,S,T$, respectively.  We now partition 
$\var(q)$ into 7 sets as follows.  The key idea is that 
for each assignment of $x,y,z$ to values $a,b,c$ in $D$, we will make assignments
according to that partition.

\begin{equation}\label{variable partition eq}
\begin{array}{rl}
 \var(S_0)-(\var(S_1) \cup \var(S_2))&\angle{ab} \\  
 \var(S_1)-(\var(S_0) \cup \var(S_2))&\angle{bc} \\ 
 \var(S_2)-(\var(S_0) \cup \var(S_1))& \angle{ca} \\  
 \var(q^*) - (\var(S_0)\cup \var(S_1)\cup \Var(S_2))& \angle{abc} \\ 
 \var(S_2) \cap \var(S_0)& a \\  
 \var(S_0) \cap \var(S_1)& b \\
 \var(S_1) \cap \var(S_2)& c
\end{array}
\end{equation}

We then define each other atom $A$ in
$D'$ to be the following set of tuples, where the only difference between atoms is which
of the 7 members of the partition of variables occurs in $\var(A)$.
\begin{equation}\label{case2 eq}
\hspace*{-.1in}\bigset{(\angle{ab};\angle{bc};\angle{ca};\angle{abc};a;b;c)}
	{D\!\models\! q_\triangle(a,b,c)}
\end{equation}

By the definition of triad, there is a path from $S_0$ to $S_1$ not using any edges (variables) from
$S_2$, i.e., none from $\var(S_2) \cup V_4 \cup V_6$.  Thus, any witness including occurrences of
some of  $\angle{ab},b',\angle{b''c}$ must 
have $b = b' = b''$.  Thus, as in Case 1, 
the witnesses of $D' \models q$ are essentially identical to the
witnesses of $D\models q_\triangle$ and we have reduced
$\res(q_\triangle)$ to $\res(q)$.
\end{proof}

\section{Independent Join Paths: details}
\label{app:IJPs}

 {We give more details on the concept of Independent Join Paths. 
We start with some intuition by providing examples} (\cref{app:IJP:examples}), state our conjecture,
 {and finish by pointing out how this concept could possibly allow an automated search
for hardness proofs} (\cref{app:IJP:automation}), 
 {a prospect we are especially excited about}.

\subsection{IJP Examples}
\label{app:IJP:examples}

We give here examples of IJPs for various queries and earlier hardness reductions, and provide the intuition for our 4 conditions.

\introparagraph{Standard paths} 
The first example shows that IJPs contain standard paths (\cref{unary path}) as a special case.

\begin{example}[$\vc$]
Consider our simplest example for an SJ-path implying hardness: 
$\vc$
from \cref{fig: vc hypergraph}.
The following database of 3 tuples forms an IJP: 
$$
D = \{R(1), S(1,2), R(2)\}
$$
\begin{enumerate}

\item We have $R(1)$ and $R(2)$ 
with
$\{1\} \not \subseteq \{2\}$
and
$\{2\} \not \subseteq \{1\}$.

\item $R(1)$ and $R(2)$ each participate in only one witness, which in this case is the same one. 

\item $R$ being unary, there can't be any other relation with a strict subset of the constants.

\item No exogenous relation.

\item The resilience $\rho(\vc,D)=1$, but becomes 0 after removing either $R(1)$ or $R(2)$ or both.
\end{enumerate}
\end{example}

\introparagraph{Triads} 
The second example shows that any query with a triad can form IJPs. 
We illustrate with our favorite triangle query.

\begin{example}[$q_\triangle$]\label{ex:app:IJP:triangle}
Consider the triangle query as the simplest example of a non-linear SJ-free query containing a triad 
(see \cref{fig:triangleHypergraph}).
The following database of 7 tuples form an IJP:
\begin{align*}
	D=\{R(1,2), R(4,2), R(4,5), S(2,3), S(5,3), T(3,1), T(3,4) \}
\end{align*}	

\begin{enumerate}

\item We have $R(1,2)$ and $R(4,5)$ 
with
$\{1,2\} \not \subseteq \{4,5\}$
and
$\{4,5\} \not \subseteq \{1,2\}$.

\item $R(1,2)$ only participates in witness $w_1 = (1,2,3)$, and $R(4,5)$ only participates
in witness $ w_2 = (4,5,3)$.

\item No other relation has a strict subsets of the constants from $R$

\item No exogenous relation.

\item The resilience $\rho(q_{\triangle},D)=2$, but becomes 1 after removing either $R(1,2)$, or $R(4,5)$, or both.
\end{enumerate}

\Cref{fig:Fig_data_H} illustrates the 3 joins forming the IJP. 
The connection to our idea from \cref{fig:Fig_VC_intuition_b} now becomes clearer.
Also notice that this IJP forms the basic element of our prior hardness proof for triads.
\end{example}	

\begin{figure}[h]
\centering
\includegraphics[width=0.8\linewidth]{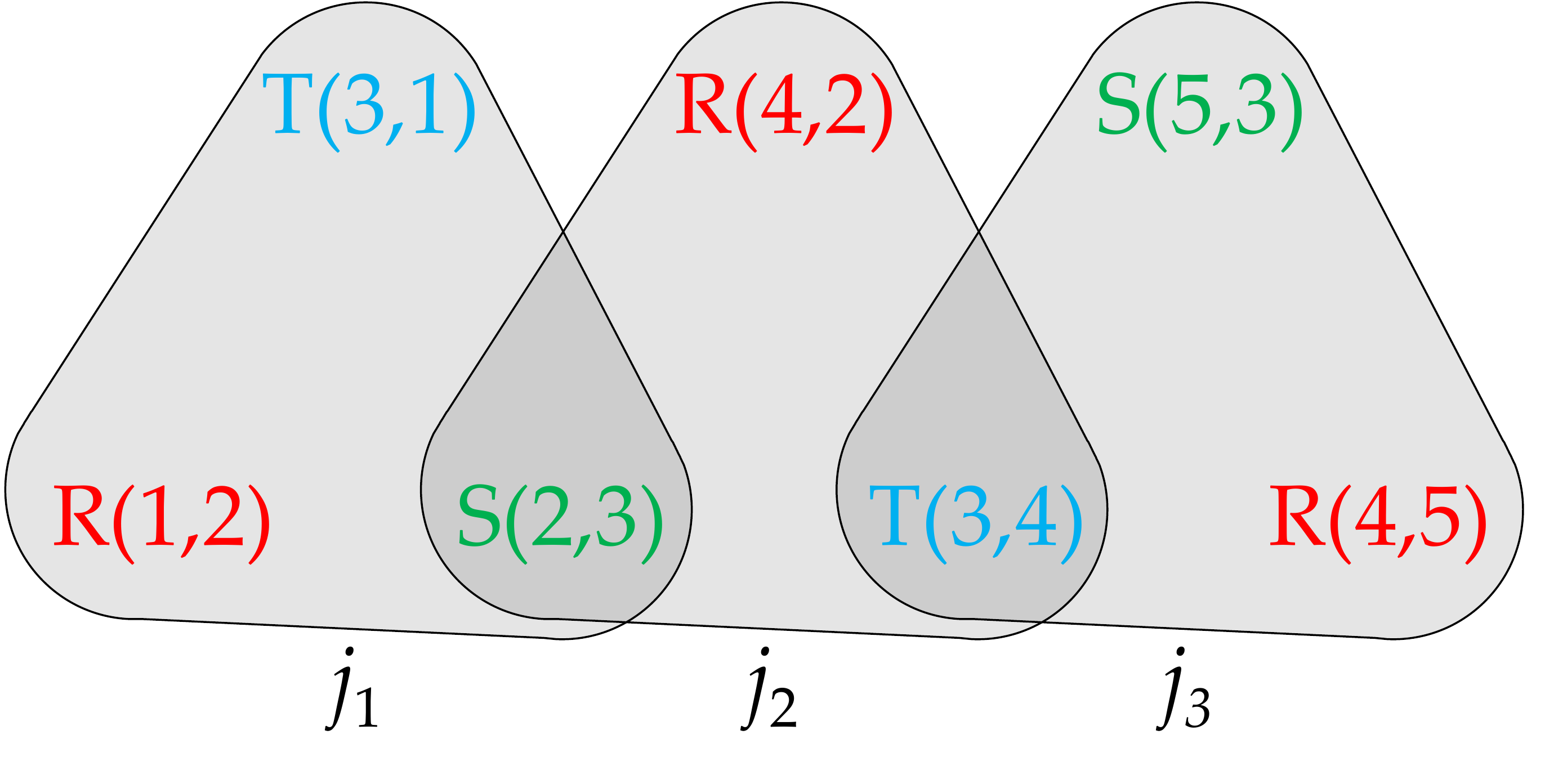}
\caption{\Cref{ex:app:IJP:triangle}: IJP for triangle query $q_\triangle$.}
\label{fig:Fig_data_H}
\end{figure}

\introparagraph{More complicated IJPs}
The third example uses a more complicated IJP.

\begin{example}[more complicated gadget]\label{ex:z5_inverter}
Consider the query
\begin{align*}
	z_5 \datarule A(x), R(x,y), R(y,z), R(z,z)
\end{align*}		
Then following database forms an IJP:
\begin{align*}
	D= \{&A(1), A(4), A(5), A(9), A(13), \\
	      & R(1,2), R(2,2), R(2,3), R(3,3), R(4,1), R(5,2),\\
	      & R(5,6), R(6,7), R(7,7), R(8,7), R(9,8),\\
	      & R(1,10), R(10, 11), R(11,11), R(12,11), R(13,12) \}
\end{align*}		
\begin{enumerate}

\item We have $A(9)$ and $A(13)$.

\item $A(9)$ only participates in witness $w_1 = (9,8,7)$ and $A(13)$ only participates in witness $w_2 = (13,12,11)$.

\item No other relation has a strict subset of the constants from $A$.

\item No exogenous relation.

\item The resilience $\rho(\vc,D)=4$ with $$\Gamma = \{R(1,2), R(2,2),R(7,7),R(11,11)\},$$ 
but becomes 3 after 
($i$) removing $A(9)$ 
with $$\Gamma = \{A(5), R(1,2), R(11,11)\},$$
or ($ii$) removing $A(13)$ 
with $$\Gamma = \{A(1), R(2,2),R(7,7)\},$$
or ($iii$) removing both
with $$\Gamma = \{A(1), A(5),R(1,2)\} \text{ or }\Gamma = \{A(1), A(5),R(2,2)\}.$$
\end{enumerate}

\Cref{fig:Fig_compositequery_H_Rzz}
illustrates how these 21 tuples create 8 different joins,
representing the IJP.
It turns out that this IJP is ``hidden'' 
and can be spotted by the careful reader 
in the crossover part of the variable gadget used in \cref{cf3.1 is npc}. 
\end{example}

\begin{figure}[h]
\centering
\includegraphics[width=0.4\textwidth]{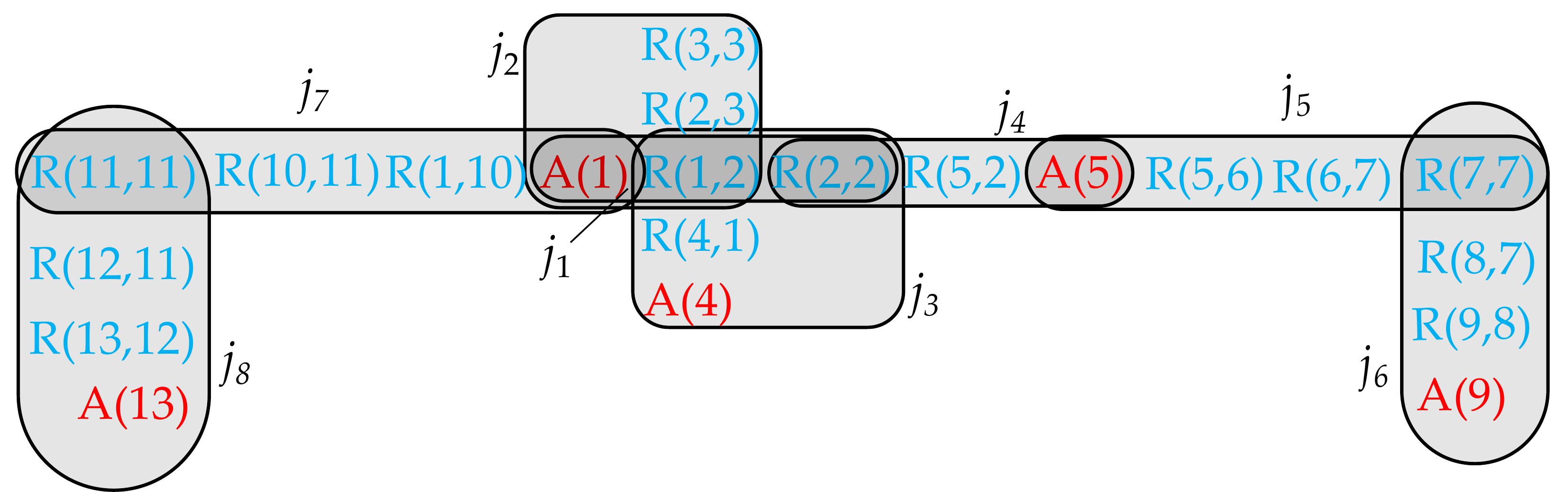}
\caption{\Cref{ex:z5_inverter}: IJP for $z_5$.}
\label{fig:Fig_compositequery_H_Rzz}
\end{figure}

\introparagraph{Condition 4}
We next give one example that illustrates why we need condition 4 of our definition for IJPs.
In particular, this query is an example in which two (instead of only one) relation is repeated.
We know through a dedicated proof that the complexity of this query is in \PTIME.
We illustrate a ``failed attempt'' to create an IJP and point out the problems that would arise if we ignored condition 4.

\begin{example}[Independent paths]\label{ex: ind path}
	Consider the following query 
	$q \datarule $ 
	$\ex{A}(x), R(x), S(x,y), S(z,y), R(z), \ex{B}(z)$
	which contains two repeated relations.
We investigate the canonical database
\begin{align*}
	D= \{ &\{R(1), \ex{A}(1), S(1,2), S(3,2), R(3), \ex{B}(3) \}
\end{align*}		
and its ability to form an IJP.
\begin{enumerate}

\item We have $R(1)$ and $R(3)$.

\item $R(1)$ and $R(3)$ participate in only one witness $w = (1,2,3)$.

\item No other relation has a strict subset of the constants from $A$.

\item Condition 3 requires that $\ex{B}(1)$ and $\ex{A}(3)$ be added to the database, which is currently not the case, and which we ignore for a moment.

\item The resilience is 1, and becomes 0 if any tuple is removed.
\end{enumerate}
	
The crucial condition 4 forces us to add $\ex{B}(1)$ and $\ex{A}(3)$ to the database. And then condition 2 and 5 are not true anymore. 
Addition of these tuples form 2 more joins
$\{R(1), \ex{A}(1), S(1,2), S(1,2), R(1), \ex{B}(1) \}$
and
$\{R(3), \ex{A}(3), S(3,2), S(3,2), R(3), \ex{B}(3) \}$,
which requires \emph{both} tuples $R(1)$ and $R(3)$ to be removed make the query false.

In other words, the canonical database is not enough to succeed with the reduction from VC (recall \cref{fig:Fig_VC_intuition_b}: any two edges incoming and outgoing from vertex $a$ create addition joins.
\end{example}

\subsection{Toward an automated proof construction}
\label{app:IJP:automation}

At its core, each IJP can be considered as a set of 
``canonical databases'' or witnesses, which have been appropriately ``aligned.''
We give the intuition with the triangle query $q_\triangle$
from \cref{ex:app:IJP:triangle}
and \cref{fig:Fig_data_H}.

\begin{example}
Assume we construct three disjoint canonical databases:
\begin{align*}
	j_1 & : R(1,2), S(2,3), T(3,1) \\
	j_2 & : S(a,b), T(b,4), R(4,a) \\
	j_3 & : T(c,d), R(d,5), S(5,d) 
\end{align*}
The total number of constants used is 9, three for each of the three joins.

We can now look at all the possible ways in which these $n=9$ constant can be partitioned into nonempty subsets.
The answer is given by the Bell number and is 21147 for $n=9$. 
Exhaustive enumeration over these 21147 cases will also lead to partition
\begin{align*}
	\{\{1\},\{2, a\},\{3, b, c\},\{4, d\},\{5\}\} 
\end{align*}	
which is isomorph to the IJP from \cref{fig:Fig_data_H}.

Our \cref{def:IJP} now provides a procedure to test that the resulting database indeed forms an IJP.
\end{example}

The more general procedure is now as follows
\begin{enumerate}
	\item for an increasing number of joins $k=1, 2, 3, \ldots$
	\item for all possible partitions
	\item for all pairs of tuples of the same relation that are not dominated
	\item if an exogenous tuple contains a subset of the constants, then possible add a second tuple
	\item calculate the minimal VC of the resulting hypergraph under the 4 cases 
		$\{(0,0), (0,1), (1,0), (1,1)\}$, where 0 and 1 mean that a tuple is present or absent, respectively.
\end{enumerate}

\end{document}